\documentclass[preprint,12pt]{elsarticle}
\usepackage[english]{babel}
\usepackage{bm}
\usepackage{amsmath}
\usepackage{amssymb}
\usepackage{hyperref}
\usepackage{accents}
\usepackage{latexsym}
\usepackage{amsthm}
\usepackage{fullpage}
\usepackage{graphicx}
\usepackage{color}
\usepackage{booktabs}
\usepackage[colorinlistoftodos]{todonotes}
\definecolor{ao(english)}{rgb}{0.0, 0.5, 0.0}
\usepackage{marginnote}
\usepackage{pgfplotstable}
\usepackage{pgfplots}
\usepackage{tikz}
\usepackage{amsmath}
\usepackage{amsfonts}
\usepackage{amssymb}
\usepackage{accents}
\usepackage{latexsym}
\usepackage{amsthm}
\usepackage{appendix}
\usepackage{stmaryrd}

\newcommand{\R}{{\mathbb R}}
\newcommand{\N}{{\bm N}}

\newcommand{\be}{\begin{eqnarray}}

\newcommand{\ee}{\end{eqnarray}}

\renewcommand{\O}{\Omega}
\renewcommand{\det}{{\rm det}\,}
\newcommand{\cof}{{\rm Cof }\,}

 \numberwithin{equation}{section}

\newtheorem{theorem}{Theorem}[section]

\newtheorem{proposition}{Proposition}[section]

\newcommand{\Red}[1]{\textcolor[rgb]{1.00,0.00,0.00}{#1}}

\newcommand{\Cross}{\mathbin{\tikz [x=1.4ex,y=1.4ex,line width=.25ex] \draw (0.1,0.1) -- (0.9,0.9) (0.1,0.9) -- (0.9,0.1);}}
\newcommand{\vect}[1]{\boldsymbol{#1}}

\journal{\Red{}} 
\begin{document}
\begin{frontmatter}
\title{A polyconvex transversely-isotropic invariant-based formulation for electro-mechanics: stability, minimisers and computational implementation}

\author{Martin Hor\'{a}k$^\dagger$\fnref{cor1}} 
\author{Antonio J. Gil$^\ddagger$\fnref{cor2}}  
\author{Rogelio Ortigosa$^{\star}$\fnref{cor3}}
\author{Martin  Kru\v{z}\'{\i}k$^{\ast,\dagger}$\fnref{cor4}}

\fntext[cor1]{Corresponding author: martin.horak@cvut.cz}
\fntext[cor2]{Corresponding author: a.j.gil@swansea.ac.uk}
\fntext[cor3]{Corresponding author: rogelio.ortigosa@upct.es}
\fntext[cor4]{Corresponding author: kruzik@utia.cas.cz}

\address{
$\dagger$ Faculty of Civil Engineering, Czech Technical
University, Th\'{a}kurova 7, CZ-166~ 29~Praha~6, Czech Republic\\
$\ddagger$ Zienkiewicz Centre for Computational Engineering, Faculty of Science and Engineering \\ Swansea University, Bay Campus, SA1 8EN, United Kingdom\\
$\star$Computational Mechanics and Scientific Computing Group,\\
Technical University of Cartagena, Campus Muralla del Mar, 30202, Cartagena (Murcia), Spain\\
$\ast$ Czech Academy of Sciences, Institute of Information Theory and Automation, Pod vod\'{a}renskou
v\v{e}\v{z}\'{\i}~4, CZ-182~00~Praha~8, Czech Republic
}

\begin{abstract}
The fabrication of evermore sophisticated miniaturised soft robotic components made up of Electro-Active Polymers (EAPs) is constantly demanding parallel development from the in-silico simulation point of view. The incorporation of crystallographic anisotropic micro-architectures, within an otherwise nearly uniform isotropic soft polymer matrix, has shown great potential in terms of advanced three-dimensional actuation (i.e. stretching, bending, twisting), especially at large strains, that is, beyond the onset of geometrical pull-in instabilities. To accommodate for this in-silico response, this paper presents a phenomenological invariant-based polyconvex transversely isotropic framework for the simulation of EAPs at large strains. This research expands previous work developed by Gil and Ortigosa for isotropic EAPs \cite{Gil2016_NewFramework} with the help of the pioneering work by Schr\"{o}eder and Neff \cite{schroder2003invariant} in the context of polyconvexity for materials endowed with crystallographic architectures in single physics mechanics. The paper also summarises key important results both in terms of the existence of minimisers and material stability. In addition, a series of numerical examples is presented in order to demonstrate the effect that the anisotropic orientation and the contrast of material properties, as well as the level of deformation and electric field, have upon the response of the EAP when subjected to large three-dimensional stretching, bending and torsion, including the possible development of wrinkling.
\end{abstract}
\begin{keyword} 
	polyconvexity, transversely isotropic, dielectric elastomers, electro-elasticity, finite element method
\end{keyword}
\end{frontmatter}
\section{Introduction}\label{sec:intro}

Electro-Active Polymers (EAPs) are a type of soft smart materials capable of displaying significant change in shape in response to electrical stimuli. EAPs have been thoroughly studied over the years \cite{Bar-Cohen2001, Kornbluh2000,Pelrine1992,deBotton2007} and dielectric elastomers are recognised as one of the most popular EAPs \cite{Pelrine1998, Pelrine2000_High-speed, Lacour2004} due to their outstanding actuation capabilities\footnote{It is worth noting that the maximum achieved voltage-induced area deformation of a dielectric elastomer has been 2200\% \cite{an2015experimental}.} (i.e., lightweight, fast response time, flexibility, low stiffness properties), which makes them ideal for their use as soft robots \cite{Bar-Cohen2004, Bortot2016, Carpi2011, Kornbluh2000, OHalloran2008, Bar-Cohen2001, Carpi2008_Dielectric,Rudykh2012, Li2013_Giant}, flexible energy generators \cite{McKay2010, Kornbluh2011} or tunable optics \cite{aschwanden2006polymeric, rosset2016small}. However, a large electric field is still generally required in order to access the large actuation regime in EAPs, an aspect that very often places them at risk of electromechanical pull-in instabilities or even electrical breakdown \cite{Bertoldi2011}.

It was shown in \cite{xiao2016suppression} that fibre reinforcement can suppress the electro-mechanical pull-in instability and improve EAPs' applicability. A dielectric elastomer stiffened with fibers in the hoop direction leading to large, voltage-induced deformations was proposed in \cite{huang2012large}. In \cite{he2018voltage}, a voltage-driven deformation of dielectric elastomer reinforced with a family of helical fibres leading to torsional behaviour is studied and the effect of fibre orientation on the actuation of a fibre-reinforced bending actuator is analysed in \cite{ahmadi2020nonlinear}. 

In order to reduce the operational voltage required for actuation in EAPs, some authors have advocated for the design of composite-based EAPs \cite{Zhang2002, Huang2004_AllOrganic, Huang2004_Enhanced}, typically combining an ultra-soft and low-permittivity elastomer matrix with a stiffer and high-permittivity inclusion \cite{PonteCastaneda2012, Goshkoderia2017, Huang2004_Enhanced, Stoyanov2011}. Due to the continuous improvement in layer-by-layer fabrication techniques \cite{Galich2017_Shear}, multi-layered laminated EAP composites have gained popularity. From the modelling point of view, References \cite{Li2003, Li2004,deBotton2006,Tian2012,Gei2018,Gei2013, Tian2012, deBotton2007, Bertoldi2011, Rudykh2011} have demonstrated that the contrast between the properties of the composite constituents is one of their critical design factors, where the actuation performance of an EAP composite subjected to simple in-plane stretching can be amplified by several orders of magnitude with respect to that of a single-phase EAP, even in the linearised regime. 

The use of computational methods constructed on the basis of variational principles is nowadays acknowledged as the preferred method for the in-silico simulation of complex actuation. Building upon the works of Toupin \cite{Toupin1956,Toupin1960} and Dorfmann, Ogden, MacMeeking, Suo, and co-workers in
\cite{Dorfmann2005, Dorfmann2006, McMeeking2004, Suo2008, dorfmann2014nonlinear,landis2002new,vu2007numerical,vu20102}, recent contributions in the field of computational electro-mechanics can be found in \cite{Vu2007,Vu2012,Jabareen2015,Franke2019,Bishara2018}. In these works, the constitutive behaviour of a single-phase electro-mechanical material is encoded within a carefully (phenomenologically) defined energy functional which depends upon appropriate strain measures, a Lagrangian electric variable and, if dissipative effects are considered, an electromechanical internal variable \cite{Liao2020}. Bustamante et al. in \cite{bustamante2009transversely, bustamante2010simple} presents an irreducible invariant basis for transversely isotropic EAPs (without piezoelectric effects) as well as several constitutive restrictions that the final energy density must fulfilled, such as the Baker–Ericksen inequality \cite{bustamante2010simple}.

However, in order to design and optimise EAPs for a specific application, a robust and reliable physics-based mathematical theory is needed. Despite great recent progress in the field of constitutive modeling of electro-active materials (see above), the polyconvexity condition \cite{ball1976convexity}, which represents a sufficient condition for a mathematically well-posed formulation, has been introduced only recently\footnote{Note that, in the context of a boundary value problem, polyconvexity together with coercivity guarantees the existence of global minimizers.}, see \cite{gil2016new, ortigosa2016new, siboni2015electromechanical}. In previous publications \cite{Gil2016_NewFramework, Ortigosa2016_ComputationalFramework, Ortigosa2016_NewFramework_ConservationLaws,Poya2018}, the authors put forward a new computational framework for single-phase reversible electro-mechanics, where material stability is always ensured via the selection Convex Multi-Variable (CMV) energy densities, that is, convex with respect to the minors of the deformation gradient tensor $(\vect{F},\vect{H},J)$, the Lagrangian electric displacement $\vect{D}_0$, and the spatial electric displacement $\vect{d}=\vect{F}\vect{D}_0$. CMV energy densities (generally referred to as \textit{polyconvex} or $\mathcal{A}$\textit{-polyconvex} \cite{Silhavy2017}) guarantee existence of minimisers \cite{Silhavy2017}, ellipticity \cite{Gil2016_NewFramework} in the quasi-static case and hyperbolicity in the dynamic case \cite{Ortigosa2016_NewFramework_ConservationLaws}, thus precluding anomalous mesh dependency effects. The theory of polyconvexity for mechanical anisotropic energy densities was magnificently explored by Schr\"{o}eder and Neff \cite{schroder2003invariant} and it was also initially studied for electro-mechanics anisotropic energies in \cite{itskov2016polyconvex}\footnote{However, the differential constraint imposed by the Gauss's law was not accounted for, leading to an incomplete polyconvexity condition, omitting the coupled electro-mechanical contributions.}.

As a result, the aim of this paper is three-fold. \textbf{\textit{First}}, the in-silico modelling of realistic three-dimensional deformation scenarios (namely, combined bending/torsion/stretching), way beyond the onset of geometrical instabilities and without the need to resort to any simplifications in the kinematics of the problem (e.g., plane strain, exact incompressibility). \textbf{\textit{Second}}, the consideration of internal micro-architectures beyond the simpler single-phase isotropy, with focus on the very relevant transverse isotropy groups, with the idea of developing computationally efficient (phenomenological, invariant-based) constitutive models for EAPs, which could be in the future enriched by state-of-the-art data assimilation techniques~ \cite{vetra2018state}. \textbf{\textit{Third}}, to ensure the existence and material stability (ellipticity) of computer simulations by unifying and enhancing key results scattered in the literature.  From the macroscopic point of view, the loss of ellipticity leads to a mathematically and numerically ill-posed problem that needs to be regularised, e.g., by introducing higher-order derivatives into the constitutive law resulting in the so-called higher-order continua, see \cite{benevsova2018note, horak2020gradient} for a recently proposed theory of gradient-polyconvexity. Such regularisation is out of the present article's scope and will be studied in a forthcoming publication.

From the \textit{\textbf{constitutive modelling standpoint}}, this paper will apply  polyconvexity, objectivity, and \textit{isotropicisation}  \cite{vsilhavy2018variational}, in order to develop energy densities constructed on the basis of polyconvex invariants forming an irreducible basis for the characterisation of transversely isotropic EAPs. From the \textbf{\textit{numerical implementation standpoint}}, Finite Element analysis of transversely isotropic EAPs will be presented using an enhanced mixed Hu-Washizu three-field type of formulation $(\vect{\phi},\vect{D}_0,\varphi)$. The high nonlinearity of the quasi-static electro-mechanical problem will be addressed via a monolithic Newton-Raphson scheme, with an arc length technique used to bypass geometrical instabilities. A tensor cross product operation between vectors and tensors \cite{Bonet2015,Bonet2016_TensorCrossProduct} will be used to reduce the complexity of the algebra.

The outline of this paper is as follows. Section \ref{sec:background_theory} describes the necessary elements of nonlinear continuum reversible electro-mechanics and some key variational principles. Section \ref{sec:poly} focusses on the concept of polyconvexity as a basis for the description of phenomenological constitutive models for EAPs. The section starts with some generalised convexity conditions, before introducing the actual definition of polyconvexity or $\mathcal{A}$-polyconvexity (as referred to by \v{S}ilhavý \cite{vsilhavy2018variational}, also referred as Multi-Variable Convexity by Gil and Ortigosa \cite{Gil2016_NewFramework}). The section summarises and puts into context key results in terms of existence of minimisers \cite{Silhavy2017} and material stability for coupled reversible electro-mechanics. Section \ref{sec:invariants} describes the basis for the anisotropic invariant-based constitutive models adopted in this work, where special attention is paid to the important crystallographic groups $\mathcal{D}_{\infty h}$ and $\mathcal{C}_{\infty}$. The section concludes by listing so-called polyconvex basis of invariants for the creation of polyconvex invariant-based constitutive models. Section \ref{sec:examples} presents a series of numerical examples in order to assess the capabilities of the proposed model. Specifically, in a  first example, a local analysis is conducted at a quadrature (Gauss point) level where the effect of orientation of anisotropy upon purely mechanical and electro-mechanical tests is presented. In the second example, complex three-dimensional bending/torsion/stretching combined modes of deformation are studied for a soft robot actuator, monitoring macroscopic stability, and observing the effect that anisotropy has in the predominance of a deformation mode against others. This will be observed only when the material is subjected to extreme deformation as it is the case of the three-dimensional bending of a multi-layered EAP. Finally, in the third example, the effect of anisotropy orientation on the wrinkling patterns is studied. Eventually, section \ref{sec:conclusions} provides some concluding remarks about the paper.

\section{Nonlinear continuum electromechanics}\label{sec:background_theory}
This section briefly summarises the fundamental electro-mechanical equations governing the response of an Electro-Active Polymer (EAP) at large strains.

\subsection{Kinematics: motion and deformation}

Let us consider the deformation of an Electro-Active Polymer (EAP) (see Figure \ref{fig:contact kinematics}) with reference configuration given by a bounded Lipschitz domain $\Omega_0\subset \mathbb{R}^3$ with boundary $\partial\Omega_0$, and its unit outward normal $\vect{N}$. After the deformation, the EAP occupies an actual configuration given by a set $\Omega\subset\mathbb{R}^3$ with boundary $\partial \Omega$ and unit outward normal $\vect{n}$. If the deformation is continuous and bijective  then  $\Omega$ is also open. The deformation of the EAP is defined by a mapping $\vect{\phi}:\Omega_0\to\mathbb{R}^3$ linking material particles $\vect{X}\in\Omega_0$ to  $\vect{x}\in\Omega = \vect{\phi}(\Omega_0)$.
Associated with $\vect{\phi}\left(\vect{X}\right)$, the deformation gradient tensor  $\vect{F}$ is defined as\footnote{Lower case indices $\{i,j,k\}$ will be used to represent the spatial configuration, whereas capital case indices $\{I,J,K\}$ will be used to represent the material description. Zero as the  lower index next to the nabla operator indicates that it is the referential gradient.}
\begin{equation}\label{eqn:definition of F}
\vect{F} = \vect{\nabla}_0\vect{\phi}\left(\vect{X}\right);\qquad
F_{iI}=\frac{\partial \phi_i}{\partial X_I}.
\end{equation}

Related to $\vect{F}$, its co-factor $\vect{H}$ and its Jacobian $J$ \cite{deBoer1982,Bonet2016_TensorCrossProduct} are defined as
\begin{equation}\label{eqn:H and J tensor cross product expression}
	\begin{aligned}
	\vect{H}=\cof \vect{F}=\frac{1}{2}\vect{F}\Cross\vect{F};\qquad  J=\det \vect{F}=\frac{1}{3}\vect{H}:\vect{F}, \end{aligned}
\end{equation}
with $\left(\vect{A}\Cross\vect{B}\right)_{iI}=\mathcal{E}_{ijk}\mathcal{E}_{IJK}A_{jJ}B_{kK}$, $\forall \vect{A},\vect{B}\in\mathbb{R}^{3\times 3}$,
where $\mathcal{E}_{ijk}$ (or $\mathcal{E}_{IJK}$) symbolises the third-order alternating tensor components and the use of repeated indices implies summation\footnote{In addition, throughout the paper, the symbol $(\cdot)$ indicates the scalar product or contraction of a single index $\vect{a}\cdot \vect{b}=a_i b_i$; the symbol $(:)$, double contraction of two indices $\vect{A}:\vect{B}=A_{ij}B_{ij}$; the symbol $(\times)$, the cross product between vectors $(\vect{a} \times \vect{b})_i=\mathcal{E}_{ijk}a_jb_k$; and the symbol $(\otimes)$, the outer or dyadic product $(\vect{a}\otimes \vect{b})_{ij}=a_ib_j$.}. The deformation mapping is assumed to be sufficiently smooth, bijective and orientation
preserving, so that $J>0 \;a.e.$ to avoid interpenetrability.

\begin{figure}[h!]
	\centering
	\begin{tabular}{c}
		\includegraphics[width=0.55\textwidth]{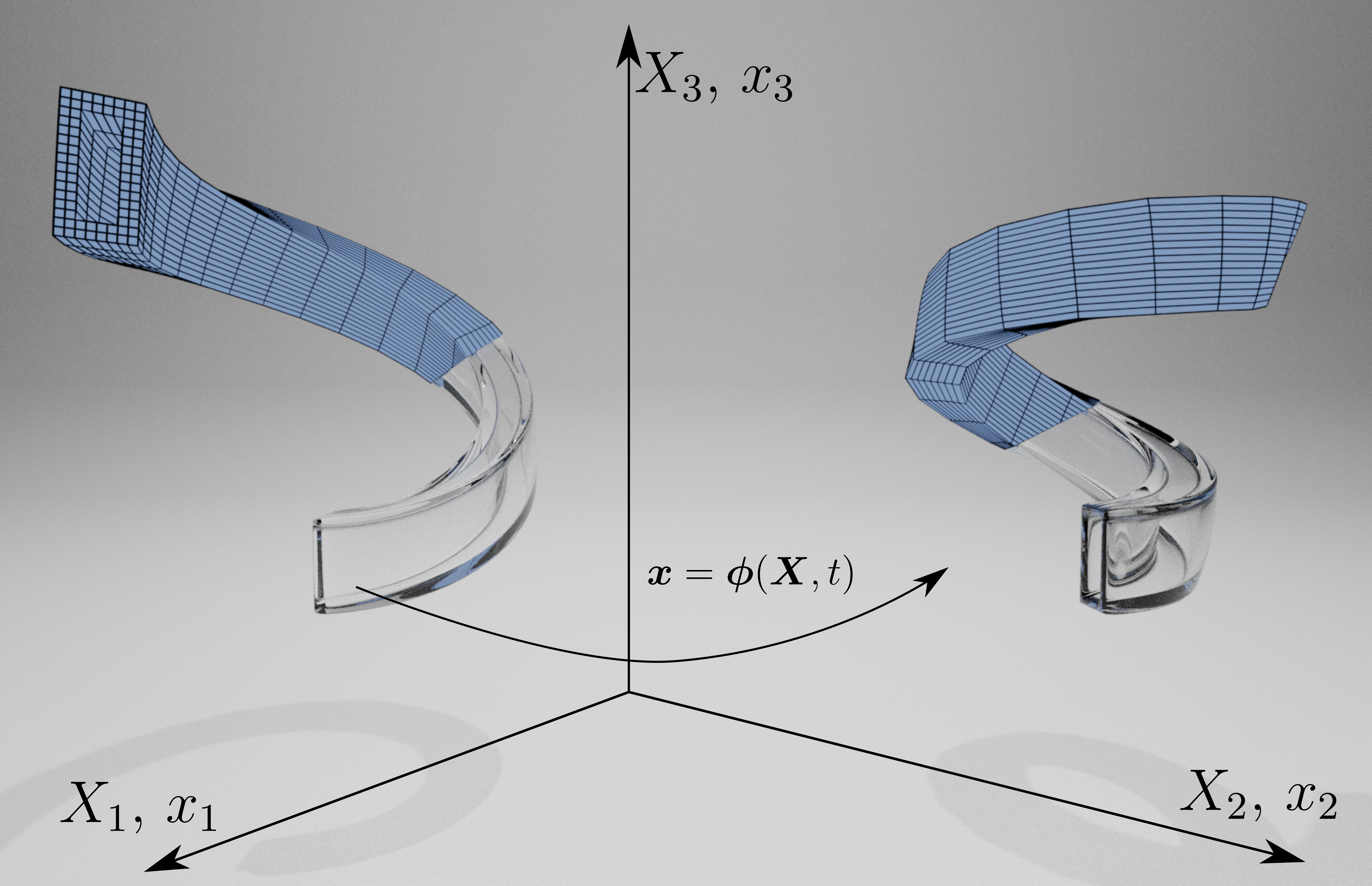} 
	\end{tabular}
	\caption{Deformation mapping $\vect{\phi}\left(\vect{X}\right)$.}
	\label{fig:contact kinematics}
\end{figure}

\subsection{Governing equations in electro-mechanics}

In the absence of inertial effects, the conservation of linear momentum (i.e. equilibrium) can be written in a Total Lagrangian fashion as
\begin{eqnarray}\label{eq:DivP}
{\rm DIV}\; \bm{P} + \bm{b}_0 &=& 0 \quad  \text{in} \quad {\Omega_0}, 
\end{eqnarray}
where ${\rm DIV}$ denotes the divergence operator computed in the material configuration, $\bm{P}$ represents the first Piola-Kirchhoff stress tensor, $\vect{b}_0$ is a body force per unit undeformed volume, and the accompanying boundary conditions are given by
\begin{subequations}
\begin{align}
\bm{P} \bm{N} &= \bar{\bm{t}} \quad  \text{on} \quad {\partial\Omega^t_0}; \\
\bm{\phi} &= \bar{\bm{\phi}}  \quad  \text{on} \quad {\partial\Omega^u_0},\label{eqn:Dirich_displacements}
\end{align}
\end{subequations}
where $\bar{\bm{t}}$ denotes the prescribed traction over part of boundary $\partial \Omega_0^t$, and $\bar{\bm{\phi}}$ are prescribed positions over $\partial \Omega_0^u$, where $\partial \Omega_0^t$ and $\partial \Omega_0^u$ are  two complementary subsets, that is, $\partial\Omega^t \cap \partial\Omega^u_0 = \emptyset$ and $\partial\Omega^t \cup \partial\Omega^u = \partial \Omega$. In addition, conservation of angular momentum leads to the well-known symmetry condition  $\bm{P}\bm{F}^T = \bm{F}\bm{P}^T$ \cite{Bonet2016_Nonlinear_Statics}.

In the absence of magnetic and time-dependent effects, the description of the EAP response is completed with Faraday's and Gauss' laws. Although, in general, the presence of electric fields in the surrounding media, that is, the vacuum, must be accounted for, in the case of EAPs this can be safely neglected due to its negligible influence\footnote{The reader is referred to \cite{Gil2016_NewFramework} for the consideration of the vacuum at a continuum level and \cite{vu2007numerical,vu20102} for its numerical implementation.}. With this in mind, Faraday's law and associated boundary conditions are expressed as 
\begin{subequations}
\begin{align}
\bm{E}_0  &= -\nabla_0 \varphi \quad  \text{in} \quad {\Omega_0};\label{eqn:faraday}
\\
{\varphi} &= \bar{\varphi} \;\;\; \qquad  \text{on} \quad {\partial\Omega^\varphi_0},\label{eqn:Dirichlet_potential}
\end{align}
\end{subequations}
and Gauss' law and associated boundary condition as
\begin{subequations}\label{eq:DivD}
\begin{align}
{\rm DIV}\;\bm{D}_0  &= \rho_0 \qquad  \text{in} \quad {\Omega_0}; \label{eq:DivD_a}
\\ 
\bm{D}_0\cdot\bm{N} &= -\omega_0 \quad  \text{on} \quad {\partial\Omega^\omega_0}, \label{eq:DivD_b}
\end{align}
\end{subequations}
where $\bm{E}_0$ and  $\bm{D}_0$ are the Lagrangian electric field and Lagrangian electric displacement, respectively, $\varphi$ is the electric potential subjected to prescribed value $\bar{\varphi}$ on the boundary $\partial\Omega^\varphi_0$, $\rho_0$ is the electric density charge per unit undeformed volume and  $\omega_0$ is the electric charge per unit undeformed area applied over part of boundary $\partial\Omega_0^\omega$, such that $\partial\Omega^\omega_0 \cap \partial\Omega^\varphi_0 = \emptyset$ and $\partial\Omega^\omega_0 \cup \partial\Omega^\varphi_0 = \partial \Omega$. To close the system of governing equations \eqref{eqn:definition of F} to \eqref{eq:DivD_b}, constitutive equations describing the behaviour of the EAP need to be introduced, which are presented in the following section.

\subsection{Constitutive equations: reversible electro-mechanics}

In the present paper, we resort to reversible behaviour described by electro-hyperelasticity~\cite{DorOgd14NTEM} and introduce an internal energy  $e$ per unit  volume in $\Omega_0$, typically defined in terms of the two-point deformation gradient tensor $\vect{F}$ and the Lagrangian electric displacement vector $\vect{D}_0$, that is, $e=e(\vect{F},\vect{D}_0)$, where $e:\R^{3\times 3}\times \R^{3}\to\R\cup\{+\infty\}$. For simplicity, we write $e=e(\boldsymbol{\mathcal{U}})$ with $\boldsymbol{\mathcal{U}}=\left(\vect{F},\vect{D}_0\right)$. 

Following standard thermodynamic arguments, that is the so-called Coleman--Noll procedure \cite{tadmor2012continuum}, the first directional derivative of the internal energy with respect to variations $\delta \boldsymbol{\mathcal{U}}=(\delta \vect{F},\delta \vect{D}_0)$ gives the first Piola-Kirchhoff stress tensor $\vect{P}$ and the electric field $\vect{E}_0$ as work (dual) conjugates of $\vect{F}$ and $\vect{D}_0$ as
\begin{equation}\label{eqn:directional derivative}
De[\delta \boldsymbol{\mathcal{U}}]=\partial_{\boldsymbol{\mathcal{U}}} e \bullet \delta \boldsymbol{\mathcal{U}}=\vect{P}:\delta\vect{F}+\vect{E}_0\cdot\delta\vect{D}_0, \qquad
\end{equation}
with
\begin{equation}\label{eq:P_and_E}
\vect{P} (\boldsymbol{\mathcal{U}})= {\partial_{\vect{F}} e}(\boldsymbol{\mathcal{U}});\qquad 
\vect{E}_0 (\boldsymbol{\mathcal{U}}) = {\partial_{\vect{D}_0} e}(\boldsymbol{\mathcal{U}}),
\end{equation}
where $\partial_{\vect{A}} e$ denotes the partial derivative of $e$ with respect to the field $\vect{A}$ and $\bullet$ denotes a suitable inner (dual) product. Similarly, assuming sufficient smoothness of the internal energy $e$, the second directional derivative of the internal energy yields the Hessian operator $\mathbb{H}_e$, that is,
\begin{equation}\label{tangent_operator_1}
D^2e[\delta \boldsymbol{\mathcal{U}};\delta \boldsymbol{\mathcal{U}}]=\delta \boldsymbol{\mathcal{U}} \bullet \mathbb{H}_e \bullet \delta \boldsymbol{\mathcal{U}}.
\end{equation}

For the sake of material characterisation, it is convenient to rewrite the Hessian $\mathbb{H}_e$ in terms of a fourth order type of elasticity tensor  $\vect{\mathcal{C}}_e\in\mathbb{R}^{3\times 3\times 3\times 3}$, a third order type of piezoelectric tensor $\vect{\mathcal{Q}}\in\mathbb{R}^{3\times 3\times 3}$ and a second-order type of dielectric tensor $\vect{\theta}\in\mathbb{R}^{3\times 3}$ defined as
\begin{equation}\label{constitutive_tensors_0}
D^2e[\delta \boldsymbol{\mathcal{U}};\delta \boldsymbol{\mathcal{U}}]=\left[\delta\vect{F}:\delta\vect{D}_0\cdot\right]\mathbb{H}_e
\left[\begin{matrix}
:\delta\vect{F}\\
\delta\vect{D}_0
\end{matrix}\right];\qquad%
\mathbb{H}_e=
\left[\begin{matrix}
\vect{\mathcal{C}}_e & \vect{\mathcal{Q}}^T\\
\vect{\mathcal{Q}} & \vect{\mathcal{\theta}}\\
\end{matrix}\right],
\end{equation}
where 
\begin{equation}\label{constitutive_tensors}
\vect{\mathcal{C}}_e(\boldsymbol{\mathcal{U}})=\partial^2_{\vect{FF}}e(\boldsymbol{\mathcal{U}});\qquad
\vect{\mathcal{Q}}(\boldsymbol{\mathcal{U}})=\partial^2_{\vect{D}_0\vect{F}}e(\boldsymbol{\mathcal{U}});\qquad
\vect{\mathcal{\theta}}(\boldsymbol{\mathcal{U}})=\partial^2_{\vect{D}_0\vect{D}_0}e(\boldsymbol{\mathcal{U}}),
\end{equation}
with $\left(\vect{\mathcal{Q}}^T\right)_{jJI}=\left(\vect{\mathcal{Q}}\right)_{IjJ}$. 

\subsubsection{The Helmholtz's free energy density function}

For completeness, it is possible to introduce an alternative Helmholtz's type of energy density $\Psi=\Psi(\boldsymbol{\mathcal{W}})$ with $\boldsymbol{\mathcal{W}}=(\vect{F},\vect{E}_0)$, where $\Psi:\R^{3\times 3}\times \R^{3}\to\R\cup\{+\infty\}$, as
\begin{equation}\label{eqn:Helmholtz_energy}
\Psi(\boldsymbol{\mathcal{W}})=-\sup_{\vect{D}_0}
\left\{\vect{E}_0\cdot\vect{D}_0 - e\left(\boldsymbol{\mathcal{U}}\right)\right\},
\end{equation}
via the use of a partial Legendre transform in \eqref{eqn:Helmholtz_energy}\footnote{A sufficient condition for the existence of the Helmholt'z energy density \eqref{eqn:Helmholtz_energy} is the convexity of the internal energy density $e$ with respect to the field $\vect{D}_0$.}, where the first Piola-Kirchhoff stress tensor and the electric displacement can be alternatively defined as
\begin{equation}
\vect{P}(\boldsymbol{\mathcal{W}})=\partial _{\vect{F}} \Psi(\boldsymbol{\mathcal{W}}) \qquad \vect{D}_0(\boldsymbol{\mathcal{W}})=-\partial_{\vect{E}_0}\Psi(\boldsymbol{\mathcal{W}}).
\end{equation}

Similarly to \eqref{tangent_operator_1}, the second directional derivative of the Helmholtz's free energy density gives rise to the Hessian operator $\mathbb{H}_{\Psi}$ of $\Psi$ as
\begin{equation}
D^2 \Psi[\delta \boldsymbol{\mathcal{W}};\delta \boldsymbol{\mathcal{W}}]=\delta \boldsymbol{\mathcal{W}} \bullet \mathbb{H}_{\Psi} \bullet \delta \boldsymbol{\mathcal{W}},
\end{equation}
which can be re-expressed in terms of the elasticity tensor  $\vect{\mathcal{C}}_{\Psi}\in\mathbb{R}^{3\times 3\times 3\times 3}$, the piezoelectric tensor $\vect{\mathcal{P}}\in\mathbb{R}^{3\times 3\times 3}$ and the  dielectric tensor $\vect{\epsilon}\in\mathbb{R}^{3\times 3}$ defined as
\begin{equation}
D^2\Psi[\delta \boldsymbol{\mathcal{W}};\delta \boldsymbol{\mathcal{W}}]=\left[\delta\vect{F}:\delta\vect{E}_0\cdot\right]\mathbb{H}_{\Psi}
\left[\begin{matrix}
:\delta\vect{F}\\
\delta\vect{E}_0
\end{matrix}\right];\qquad%
\mathbb{H}_{\Psi}=
\left[\begin{matrix}
\vect{\mathcal{C}}_{\Psi} & -\vect{\mathcal{P}}^T\\
-\vect{\mathcal{P}} &   -\vect{\mathcal{\epsilon}}\\
\end{matrix}\right],
\end{equation}
where 
\begin{equation}\label{constitutive_tensors_2}
\vect{\mathcal{C}}_{\Psi}(\boldsymbol{\mathcal{W}})=\partial^2_{\vect{FF}}\Psi(\boldsymbol{\mathcal{W}});\qquad
\vect{\mathcal{P}}(\boldsymbol{\mathcal{W}})=-\partial^2_{\vect{E}_0\vect{F}}\Psi(\boldsymbol{\mathcal{W}});\qquad
\vect{\mathcal{\epsilon}}(\boldsymbol{\mathcal{W}})=-\partial^2_{\vect{E}_0\vect{E}_0}\Psi(\boldsymbol{\mathcal{W}}),
\end{equation}
with $\left(\vect{\mathcal{P}}^T\right)_{jJI}=\left(\vect{\mathcal{P}}\right)_{IjJ}$. It is possible \cite{Ortigosa2016_ComputationalFramework} to relate the components of the Hessian operators of $\Psi$ \eqref{constitutive_tensors} and $e$ \eqref{constitutive_tensors_2} as
\begin{equation}\label{eqn:constitutive tensors relationship}
\begin{aligned}
\vect{\epsilon} =\vect{\theta}^{-1};\qquad
\vect{\mathcal{P}}^T=-\vect{\mathcal{Q}}\cdot\vect{\epsilon};\qquad
\vect{\mathcal{C}}_{\Psi}=\vect{\mathcal{C}}_e + \vect{\mathcal{Q}}^T\cdot\vect{\mathcal{P}},
\end{aligned}
\end{equation}
where the inner product $\cdot$ above indicates contraction of the indices placed immediately before and after it.

\subsection{Variational principles in reversible electro-mechanics}

The solution of the boundary value problem defined by the governing equations \eqref{eqn:definition of F} to \eqref{eq:DivD_b}, in conjunction with a suitable definition of the internal energy density $e$ or Helmholtz's free energy density $\Psi$, is typically solved based on the direct methods of the calculus of variations. With that in mind, the \textit{\textbf{first mixed variational principle}} can be introduced as
\begin{equation}\label{eqn:mixed variational principle I}
(\vect{\phi}^{\ast},\vect{D}_0^{\ast})=\text{arg} \inf_{\mathcal{A}_{\Pi_e}} \Pi_e\left(\vect{\phi},\vect{D}_0\right) ,
\end{equation}
where
\begin{equation}\label{eqn:mixed variational principle Ib}
\Pi_e\left(\vect{\phi},\vect{D}_0\right)=\int_{{\Omega}_0}e\left(\vect{\vect{F}},\vect{D}_0\right)\,dV - \Pi_{\text{ext}}^m\left(\vect{\phi}\right),
\end{equation}
$\Pi_{\text{ext}}^m(\vect{\phi})$ representing the external work done by the mechanical actions, defined as
\begin{equation}\label{eqn:mechanical_work}
\Pi_{\text{ext}}^m\left(\vect{\phi}\right) = \int_{\Omega_0}\vect{b}_0\cdot\vect{\phi}\,dV + 
\int_{\partial\Omega_0^t}\bar{\vect{t}}\cdot\vect{\phi}\,dA,
\end{equation}
and with the functional space $\mathcal{A}_{\Pi_e}$ given by
\begin{equation}
\mathcal{A}_{\Pi_e}=\left\{\vect{\phi} \subset \mathcal{G}_{\vect{\phi}}(\Omega_0;\R^{3}),\vect{D}_0 \subset \mathcal{G}_{\vect{D}_0}(\Omega_0;\R^{3}); \;\text{s.t}. \eqref{eqn:Dirich_displacements},\eqref{eq:DivD_a},\eqref{eq:DivD_b} \right\}.
\end{equation}
with $\mathcal{G}_{\vect{\phi}},\mathcal{G}_{\vect{D}_0}$ suitable functional spaces to be defined. The use of the asterisk $(\cdot)^{\ast}$ in \eqref{eqn:mixed variational principle I} denotes the solution (minimiser) of the variational principle. As it is well-known, for the existence of minimisers of the above variational problem, some mathematical restrictions must be imposed on 1) $\mathcal{A}_{\Pi_e}$ and 2) the form that the internal energy density $e$ must adopt \cite{ball1976convexity,vsilhavy2018variational}, which will be the focus of the following section. 

From the computational standpoint, it is convenient to resort to an alternative variational principle in order to by-pass the strong imposition of conditions \eqref{eq:DivD} and allow the possibility of imposing boundary conditions directly on the electric potential (i.e. via electrodes attached to the EAP). Indeed, after the introduction of the electric potential $\varphi$ (which can be understood as a Lagrange multiplier for the weak enforcement of Gauss' law \eqref{eq:DivD}), the \textit{\textbf{second mixed variational principle}} emerges as
\begin{equation}\label{eqn:mixed variational principle II}
(\vect{\phi}^{\ast},\varphi^{\ast},\vect{D}_0^{\ast})=\text{arg} \inf_{\mathcal{A}_{\tilde{\Pi}_{e,\vect{\phi}}}}
\inf_{\mathcal{A}_{\tilde{\Pi}_{e,\vect{D}_0}}}
\sup_{\mathcal{A}_{\tilde{\Pi}_{e,\varphi}}}
\tilde{\Pi}_e\left(\vect{\phi},\varphi,\vect{D}_0\right) ,
\end{equation}
where
\begin{equation}
\tilde{\Pi}_e\left(\vect{\phi},\varphi,\vect{D}_0\right)=\int_{{\Omega}_0}e\left(\vect{F},\vect{D}_0\right)\,dV + \int_{{\Omega}_0}\vect{D}_0 \cdot \vect{\nabla}_0 \varphi\,dV- \Pi_{\text{ext}}\left(\vect{\phi},\varphi\right),
\end{equation}
where  $\Pi_{\text{ext}}(\vect{\phi},\varphi)$ represents the combined external work done by the mechanical and electrical actions, defined as
\begin{equation}\label{eqn:total_work}
\Pi_{\text{ext}}\left(\vect{\phi},\varphi\right)=\Pi_{\text{ext}}^m\left(\vect{\phi}\right)+\Pi_{\text{ext}}^e\left({\varphi}\right);\qquad 
\Pi_{\text{ext}}^e\left({\varphi}\right) = -\int_{\Omega_0}\rho_0 {\varphi}\,dV - 
\int_{\partial\Omega_0^t}\omega_0{\varphi}\,dA,
\end{equation}
and with 

\begin{equation}\label{eqn:spaces}
\begin{aligned}
\mathcal{A}_{\tilde{\Pi}_{e,\vect{\phi}}}&=\left\{\vect{\phi}\subset \tilde{\mathcal{G}}_{\vect{\phi}}(\Omega_0;\R^{3});\;\text{s.t}. \eqref{eqn:Dirich_displacements} \right\};\;\\
\mathcal{A}_{\tilde{\Pi}_{e,\vect{D}_0}}&=\left\{\vect{D}_0 \subset \tilde{\mathcal{G}}_{\vect{D}_0}(\Omega_0;\R^{3})\right\};\; \\
\mathcal{A}_{\tilde{\Pi}_{e,\varphi}}&=\left\{\varphi \subset \tilde{\mathcal{G}}_{\varphi}(\Omega_0;\R);\;\text{s.t}. \eqref{eqn:Dirichlet_potential}\right\},
\end{aligned}
\end{equation}
with $\tilde{\mathcal{G}}_{\vect{\phi}},\tilde{\mathcal{G}}_{\vect{D}_0},\tilde{\mathcal{G}}_{\varphi}$ suitable functional spaces to be defined. This variational principle \eqref{eqn:mixed variational principle II} will be the one selected for the numerical examples presented later on in this paper, due to its excellent accuracy as a result of its extended nature. 

Furthermore, use of the Helmholtz's free energy density functional \eqref{eqn:Helmholtz_energy} in \eqref{eqn:mixed variational principle II} leads to the \textit{\textbf{third mixed variational principle}} as follows,
\begin{equation}\label{eqn:mixed variational principle III}
(\vect{\phi}^{\ast},\varphi^{\ast})=\text{arg} \inf_{\mathcal{A}_{\Psi_{\vect{\phi}}}}\sup_{\mathcal{A}_{\Psi_{\varphi}}} \Pi_{\Psi}\left(\vect{\phi},\varphi\right) ,
\end{equation}
where
\begin{equation}\label{eqn:displacement based formulation}
\Pi_{\Psi}\left(\vect{\phi},\varphi\right)=\int_{{\Omega}_0}\Psi\left(\vect{F},\vect{E}_0\right)\,dV - \Pi_{\text{ext}}\left(\vect{\phi},\varphi\right),
\end{equation}
where $\mathcal{A}_{\Psi_{\vect{\phi}}}$ and $\mathcal{A}_{\Psi_{\varphi}}$ are suitably adapted functional spaces analogous to their counterparts $\mathcal{A}_{\tilde{\Pi}_{e,\vect{\phi}}}$ and $\mathcal{A}_{\tilde{\Pi}_{e,\varphi}}$ in \eqref{eqn:spaces}. The latter variational principle is typically preferred for numerical simulations due to its compact form and reduced number of degrees of freedom, described solely in terms of the mapping $\phi^{\ast}$ and the electrical potential $\varphi^{\ast}$. 

\section{Polyconvexity in electromechanics}\label{sec:poly}
To ensure the existence of solutions of the above variational principles $\Pi_e$ \eqref{eqn:mixed variational principle I}, $\tilde{\Pi}_e$ \eqref{eqn:mixed variational principle II} or $\Pi_{\Psi}$ \eqref{eqn:mixed variational principle III}, the energy density $e$ must fulfil certain mathematical requirements. In particular, we require, that the energy functionals are lower semicontinuous in a suitable topology.  This relates to the convexity properties of the integrands.   Prior to introducing the notion of polyconvexity in the context of electromechanics, the section starts by introducing some alternative mathematical restrictions (the reader is referred to the  monograph \cite{Silhavy_book} for an in-depth discussion in the context of  mechanics). The section continues by introducing the requirement of polyconvexity and its most relevant implications, namely, the existence of minimisers and the assurance of material stability.

\subsection{Some notions on generalised convexity conditions}

One of the simplest conditions is that of \textit{\textbf{convexity}} of $e(\boldsymbol{\mathcal{U}})$, that is 
\begin{equation}\label{convexity}
e(\lambda\boldsymbol{\mathcal{U}}_1+(1-\lambda)\boldsymbol{\mathcal{U}}_2)\leq\lambda e(\boldsymbol{\mathcal{U}}_1) +(1-\lambda)e(\boldsymbol{\mathcal{U}}_2);\quad \forall\; \boldsymbol{\mathcal{U}}_1,\boldsymbol{\mathcal{U}}_2 ;\quad \lambda \in [0,1], 
\end{equation}
which for functions with first order differentiability can be alternatively written as
\begin{equation}
e(\boldsymbol{\mathcal{U}}+\delta \boldsymbol{\mathcal{U}})-e(\boldsymbol{\mathcal{U}})-De(\boldsymbol{\mathcal{U}})[\delta \boldsymbol{\mathcal{U}}]\geq0;\qquad \forall\, \boldsymbol{\mathcal{U}}, \delta \boldsymbol{\mathcal{U}},
\end{equation}
and for functions with second order differentiability (in terms of the Hessian $[\mathbb{H}_e]$) as
\begin{equation}\label{eqn:convexity2}
D^2e(\boldsymbol{\mathcal{U}})[\delta \boldsymbol{\mathcal{U}};\delta \boldsymbol{\mathcal{U}}]=\delta \boldsymbol{\mathcal{U}} \bullet [\mathbb{H}_e] \bullet \delta \boldsymbol{\mathcal{U}}
\geq 0;\qquad \forall \,\boldsymbol{\mathcal{U}},\delta \boldsymbol{\mathcal{U}},
\end{equation}
which requires (refer to \eqref{constitutive_tensors_0}, \eqref{constitutive_tensors}) positive semi-definiteness of $\vect{\mathcal{C}}_e$ and $\vect{\theta}$, independently, and further restrictions on the form that the third-order tensor $\vect{\mathcal{Q}}$ can adopt. However, convexity away from the origin (i.e. $\vect{F} \approx \vect{I}$ and $\vect{D}_0\approx \vect{0}$) is not a suitable physical restriction as it precludes the realistic behavior of materials such as buckling in the purely mechanical context, as well as the possibility of voltage-induced buckling, inherent to soft dielectric materials \cite{Vera_thesis}. An alternative mathematical restriction is that of $\mathcal{A}$-quasiconvexity of $e$ \cite{FoMu99}. Unfortunately, $\mathcal{A}$-quasiconvexity is a nonlocal condition that is very difficult even impossible to be verified. 

A necessary restriction implied by $\mathcal{A}$-quasiconvexity is that of \textbf{\textit{generalised rank-one convexity}} of $e$. A generalised  rank-one convex energy density verifies 
\begin{equation}\label{rank-convexity}
e(\lambda\boldsymbol{\mathcal{U}}+(1-\lambda)\tilde{\boldsymbol{\mathcal{U}}})\leq\lambda e(\boldsymbol{\mathcal{U}}) +(1-\lambda)e(\tilde{\boldsymbol{\mathcal{U}}}); \quad \forall\; \boldsymbol{\mathcal{U}} ;\quad \lambda \in [0,1], 
\end{equation}
and with $\tilde{\boldsymbol{\mathcal{U}}}=\boldsymbol{\mathcal{U}}+\delta\boldsymbol{\mathcal{U}}$ and $\delta \boldsymbol{\mathcal{U}}=(\vect{u}\otimes \vect{V},\vect{V}_{\perp})$, where $\vect{V}\cdot \vect{V_{\perp}}=0$ and with $\vect{u}$, $\vect{V}$ and $\vect{V}_{\perp}$ any arbitrary vectors. For the case of energies with first order differentiability, generalised rank-one convexity can alternatively be written as
\begin{equation}\label{eqn:rank1_der1}
e(\boldsymbol{\mathcal{U}}+\delta \boldsymbol{\mathcal{U}})-e(\boldsymbol{\mathcal{U}})-De(\boldsymbol{\mathcal{U}})[\delta \boldsymbol{\mathcal{U}}]\geq0;\qquad\delta \boldsymbol{\mathcal{U}}=(\vect{u}\otimes \vect{V},\vect{V}_{\perp});\qquad \forall\, \boldsymbol{\mathcal{U}},\vect{u},\vect{V},\vect{V}_{\perp},
\end{equation}
and for energies with second order differentiability,
\begin{equation}\label{eqn:ellipticity condition_1}
D^2e(\boldsymbol{\mathcal{U}})[\delta \boldsymbol{\mathcal{U}};\delta \boldsymbol{\mathcal{U}}]=\delta \boldsymbol{\mathcal{U}} \bullet \mathbb{H}_e \bullet \delta \boldsymbol{\mathcal{U}}
\geq 0;\qquad \delta \boldsymbol{\mathcal{U}}=(\vect{u}\otimes \vect{V},\vect{V}_{\perp}); \qquad \forall \,\boldsymbol{\mathcal{U}},\vect{u},\vect{V},\vect{V}_{\perp}.
\end{equation}

The generalised rank-one convexity restriction in its form \eqref{eqn:rank1_der1} plays a key role in the suitable definition of jump conditions across a travelling interface defined by a normal vector $\vect{V}$ moving with speed $c$ (i.e. singular surface \cite{Silh97MTCM}). Analogously, condition \eqref{eqn:ellipticity condition_1} is referred to as the Legendre-Hadamard condition or \textbf{\textit{ellipticity}} of $e$, linked to the propagation of travelling plane wave within the material defined by a vector $\vect{V}$ and speed $c$. These two aspects will be presented in the following two subsections.

\subsubsection{Generalised rank-one convexity conditions and jump conditions}
Consider a travelling discontinuity within an EAP, moving at speed $c$ and defined by a normal vector $\vect{V}$ separating states $(\bullet)^a$ and $(\bullet)^b$,  where $\llbracket\bullet\rrbracket =(\bullet)^a-(\bullet)^b$ denotes the jump across the interface and $(\bullet)^{{\text{Av}}}=1/2((\bullet)^a+(\bullet)^b)$ the average state. Following \cite{Ortigosa2016_NewFramework_ConservationLaws}, the jump or Rankine-Hugoniot conditions for the combined set of Cauchy-Maxwell (electro-magneto-mechanical) equations are given by
\begin{subequations}
\begin{align}
-c\llbracket  \vect{F}\rrbracket&= \llbracket \vect{v}\rrbracket \otimes \vect{V};\label{eqn:jump1}\\
-c\rho_R \llbracket  \vect{v}\rrbracket&= \llbracket \vect{P}\rrbracket  \vect{V};\label{eqn:jump2}\\
-c\llbracket \vect{B}_0\rrbracket&= \llbracket \vect{E}_0\rrbracket \times \vect{V};\label{eqn:jump3}\\
c\llbracket  \vect{D}_0\rrbracket&= \llbracket \vect{H}_0\rrbracket \times \vect{V};\label{eqn:jump4}\\
-c\llbracket e +1/2\rho_R \vect{v}\cdot \vect{v}\rrbracket&=\llbracket\vect{v}\cdot \vect{P} \rrbracket \cdot \vect{V} -\llbracket \vect{E}_0 \times \vect{H}_0\rrbracket \cdot \vect{V},\label{eqn:jump5}
\end{align}
\end{subequations}
where, in addition to previously defined fields, $\rho_R$, $\vect{v}$, $\vect{B}_0$ and $\vect{H}_0$ denote the density of the material per unit undeformed volume, the velocity, the magnetic induction vector and the magnetic field, respectively. Substitution of \eqref{eqn:jump1}-\eqref{eqn:jump4} into \eqref{eqn:jump5} renders
\begin{equation} \label{eqn:jump5b}
\begin{aligned}
\llbracket e \rrbracket&=\vect{P}^{\text{Av}}:\llbracket\vect{F} \rrbracket + \vect{E}_0^{\text{Av}}\cdot \llbracket\vect{D}_0 \rrbracket +\vect{H}_0^{\text{Av}}\cdot\llbracket\vect{B}_0 \rrbracket,
\end{aligned}
\end{equation}
which, by re-defining $(\bullet)^{\text{Av}}=(\bullet)^a+\llbracket \bullet \rrbracket/2$ and removing the upper index $(\bullet)^a$ for simplicity, can be rewritten as
\begin{equation} \label{eqn:jump5c}
\begin{aligned}
\llbracket e \rrbracket-\vect{P}:\llbracket\vect{F} \rrbracket - \vect{E}_0\cdot\llbracket\vect{D}_0 \rrbracket -\vect{H}_0\cdot\llbracket\vect{B}_0 \rrbracket &= \llbracket\vect{P}\rrbracket:\llbracket\vect{F} \rrbracket + \llbracket\vect{E}_0\rrbracket\cdot\llbracket\vect{D}_0 \rrbracket +\llbracket\vect{H}_0\rrbracket\cdot\llbracket\vect{B}_0 \rrbracket.
\end{aligned}
\end{equation}

For the case when magnetic effects are neglected (i.e. $\llbracket \vect{B}_0 \rrbracket \approx \vect{0}$) the last term on both sides of above equation \eqref{eqn:jump5c} can be neglected. Moreover, notice that the term $\llbracket\vect{E}_0\rrbracket\cdot\llbracket\vect{D}_0 \rrbracket \approx \vect{0}$ as $\vect{E}_0$ can only jump tangentially across the discontinuity (refer to \eqref{eqn:jump3}) whilst $\vect{D}_0$ can only jump along $\vect{V}$ (refer to \eqref{eqn:jump4}), that is, $\llbracket\vect{D}_0 \rrbracket=\vect{V}_{\perp}$, $\vect{V}\cdot \vect{V}_{\perp}=0$. Thus, we arrive at the simpler equation 
\begin{equation} \label{eqn:jump5d}
\begin{aligned}
\llbracket e \rrbracket-\vect{P}:\llbracket\vect{F} \rrbracket - \vect{E}_0\cdot \vect{V}_{\perp}   &= \llbracket\vect{P}\rrbracket:\llbracket\vect{F} \rrbracket.
\end{aligned}
\end{equation}

Redefining $-\llbracket  \vect{v}\rrbracket/c=\vect{u}$, that is $\llbracket\vect{F}\rrbracket=\vect{u} \otimes \vect{V}$, and making use of \eqref{eqn:jump1} and \eqref{eqn:jump2}, it is straightforward to show that $\llbracket\vect{P}\rrbracket:\llbracket\vect{F} \rrbracket=\rho_R c^2 \vect{u}\cdot \vect{u}$, which finally gives
\begin{equation} \label{eqn:jump5e}
\begin{aligned}
\llbracket e \rrbracket-\vect{P}:(\vect{u} \otimes \vect{V}) - \vect{E}_0\cdot \vect{V}_{\perp}   &= \rho_R c^2 \vect{u}\cdot \vect{u}\geq 0,
\end{aligned}
\end{equation}
which can be equivalently written as
\begin{equation}\label{eqn:rank-1}
e(\vect{F}+\vect{u}\otimes \vect{V},\vect{D}_0 + \vect{V}_{\perp})-e(\vect{F},\vect{D}_0)-\partial_{\vect{F}}e(\vect{F}, \vect{D}_0):(\vect{u}\otimes \vect{V})-\partial_{\vect{D}_0}e(\vect{F}, \vect{D}_0) \cdot \vect{V}_{\perp}\geq 0.
\end{equation}

Notice that above inequality \eqref{eqn:rank-1} holds for an energy density $e$ defined as generalised rank-one convexity (refer to equation \eqref{eqn:rank1_der1}). 

\subsubsection{Ellipticity and Legendre-Hadamard condition}

Condition \eqref{eqn:ellipticity condition_1} is referred to as the Legendre-Hadamard condition or ellipticity of $e$ and is related to the propagation of  travelling plane waves\footnote{This concept is also strongly related to that of the propagation of \textbf{\textit{acceleration waves}} within a discontinuity defined by a normal vector $\vect{V}$ with speed $c$.} within the material defined by a vector $\vect{V}$ and speed $c$. Notice that the above expression \eqref{eqn:ellipticity condition_1} is a generalisation of the concept of ellipticity from elasticity to the more general case of electro-elasticity, which can be easily verified by neglecting the piezoelectric and dielectric components of the Hessian operator, that is
\begin{equation}\label{eqn:ellip}
\left(\vect{u}\otimes\vect{V}\right):
\boldsymbol{\mathcal{C}}_e(\boldsymbol{\mathcal{U}})
:\left(\begin{matrix}
\vect{u}\otimes\vect{V}
\end{matrix}\right)\geq 0; \qquad \forall \, \boldsymbol{\mathcal{U}},\vect{u},\vect{V}.
\end{equation}

Similarly, particularising for $\vect{u}=\vect{0}$ in either \eqref{eqn:rank1_der1} or \eqref{eqn:ellipticity condition_1}, it reduces to the condition of convexity with respect to $\vect{D}_0$. This condition is crucial in order to guarantee the existence of the Helmholtz's free energy density function \eqref{eqn:Helmholtz_energy} and thus that of the variational principle $\Pi_{\Psi}$ \eqref{eqn:mixed variational principle III}.

The derivation of the above Legendre-Hadamard condition \eqref{eqn:ellipticity condition_1} has its roots in the study of the hyperbolicity (or stability in the quasi-static case) of the system of the generalised Cauchy-Maxwell equations \cite{Ortigosa2016_NewFramework_ConservationLaws} in order to ensure the existence of real wave speeds propagating throughout the domain. Crucially, the \textit{ab initio} existence of real wave speeds for the specific system of governing equations \eqref{eqn:definition of F} to \eqref{eq:DivD_b} (i.e. magnetic and time-dependent effects are neglected in the Maxwell equations) can be monitored via the study of the so-called electro-mechanical acoustic tensor $\vect{Q}$ defined as
\begin{equation}\label{eqn: Exp1 acoustic tensor}
\vect{Q}=\vect{\mathcal{C}_{e,\vect{V}\vect{V}}}+\vect{\mathcal{Q}_{\vect{V}}}^{T}\vect{\mathcal{\theta}}^{-1}\left[\frac{\vect{\vect{V}}\otimes\vect{\mathcal{\theta}}^{-1}\vect{V}}{\vect{V}\cdot\vect{\mathcal{\theta}}^{-1}\vect{V}}-\vect{I}\right]\vect{\mathcal{Q}_{\vect{V}}},
\end{equation}
where
\begin{equation}
\left[\vect{\mathcal{C}_{e,\vect{V}\vect{V}}}\right]_{ij}=\mathcal{C}_{e,iIjJ}V_{I}V_{J},\qquad
\left[\vect{\mathcal{Q}_{\vect{V}}}\right]_{Ij}=\mathcal{Q}_{IjJ}V_{J}, \qquad
\left[\vect{I}\right]_{IJ}=\delta_{IJ}.
\end{equation}

Specifically, the eigenvalues of the acoustic tensor $\vect{Q}$ are proportional to the squared of the volumetric and shear wave speeds of the electroactive material. Hence, the above tensor $\vect{Q}$ can be used as a suitable localisation measure for the onset of material instabilities by ensuring that the wave speeds are kept real throughout the entire electro-deformation process. This can be achieved by ensuring that
\begin{equation}\label{eqn:eqb:acoustic_tensor_check1}
\vect{u}\cdot \vect{Q}(\boldsymbol{\mathcal{U}}) \vect{u}\geq0; \qquad \forall\, \boldsymbol{\mathcal{U}}, \vect{u}, \vect{V}
\end{equation}
or alternatively written as
\begin{equation}\label{eqb:acoustic_tensor_check2}
\left(\vect{u}\otimes \vect{V} \right):\left[\vect{\mathcal{C}_{e}}+\underbrace{\vect{\mathcal{Q}}^{T}\vect{\mathcal{\theta}}^{-1}\left[\frac{\vect{\vect{V}}\otimes\vect{\mathcal{\theta}}^{-1}\vect{V}}{\vect{V}\cdot\vect{\mathcal{\theta}}^{-1}\vect{V}}-\vect{I}\right]\vect{\mathcal{Q}}}_{\boldsymbol{\mathcal{C}_{e}'}}\right] :\left(\vect{u}\otimes \vect{V} \right)\geq 0; \qquad \forall\, \vect{u},\vect{V}.
\end{equation}

Comparison of \eqref{eqb:acoustic_tensor_check2} against \eqref{eqn:ellip} permits to identify the modification of the traditional ellipticity condition in elasticity with that of electro-elasticity via the additional contribution $\boldsymbol{\mathcal{C}}'_e$, playing a key role in the modification of the dilatational and shear wave speeds from those of simple elasticity. The reader is referred to \cite{Ortigosa2016_NewFramework_ConservationLaws} for a comprehensive derivation of the acoustic tensor for the combined electro-magneto-mechanical case. A more compact derivation for the specific electro-mechanical case of interest in this paper is included in Appendix \ref{appendix:acoustic_tensor}.

\subsection{$\mathcal{A}$-polyconvexity}

Motivated by considerations of material stability (i.e. existence of real wave speeds) first and existence of minimisers subsequently, Gil and Ortigosa in \cite{Gil2016_NewFramework, Ortigosa2016_NewFramework_FiniteElements, Ortigosa2016_NewFramework_ConservationLaws, Ortigosa2016_ComputationalFramework} and \v{S}ilhav\'{y} in \cite{Silhavy2017, vsilhavy2018variational} extended the concept of polyconvexity \cite{Ball1976, Ball2002, Schroeder2003, Schroeder2005,kambouchev2007polyconvex, Itskov2016} from elasticity to electro-magneto-elasticity and proposed a new convexity restriction on the form that the internal energy can adopt, that is, \textbf{$\mathcal{A}$-\textit{polyconvexity}} of $e$. Following \cite{Silhavy2017}, the internal energy $e:\R^{3\times 3}\times \R^{3}\to\R\cup\{+\infty\}$ is defined as  $\mathcal{A}$-polyconvex if there exists a convex and lower semicontinuous function $W :\mathbb{R}^{3\times3}\times\mathbb{R}^{3\times3}\times\mathbb{R}\times\mathbb{R}^3\times\mathbb{R}^3 \to \mathbb{R}\cup\{+\infty\}$ (in general non-unique) defined as
\begin{align}\label{polyconvexity_emm}
e(\boldsymbol{\mathcal{U}}) = W(\boldsymbol{\mathcal{V}}); \quad \boldsymbol{\mathcal{V}}=(\vect{F},\vect{H},J,\vect{D}_0,\vect{d}); \quad \vect{d}=\vect{FD}_0.
\end{align}

Notice that convexity of $W(\boldsymbol{\mathcal{V}})$ implies that
\begin{equation}
W(\lambda\boldsymbol{\mathcal{V}}_1+(1-\lambda)\boldsymbol{\mathcal{V}}_2)\leq\lambda W(\boldsymbol{\mathcal{V}}_1) +(1-\lambda)W(\boldsymbol{\mathcal{V}}_2);\quad \forall\; \boldsymbol{\mathcal{V}}_1,\boldsymbol{\mathcal{V}}_2;\quad \lambda \in [0,1]. 
\end{equation}
which for functions with sufficient differentiability, can be alternatively written as
\begin{equation}\label{eqn:pol2}
W(\boldsymbol{\mathcal{V}}+\delta \boldsymbol{\mathcal{V}})-W(\boldsymbol{\mathcal{V}})-DW(\boldsymbol{\mathcal{V}})[\delta \boldsymbol{\mathcal{V}}]\geq0;\qquad \forall\, \boldsymbol{\mathcal{V}}, \delta \boldsymbol{\mathcal{V}},
\end{equation}
or even in terms of the Hessian operator $[\mathbb{H}_W]$ as
\begin{equation}\label{eqn:polyconvexity2}
D^2W(\boldsymbol{\mathcal{V}})[\delta \boldsymbol{\mathcal{V}};\delta \boldsymbol{\mathcal{V}}]=\delta \boldsymbol{\mathcal{V}} \bullet [\mathbb{H}_W] \bullet \delta \boldsymbol{\mathcal{V}}
\geq 0;\qquad \forall \,\boldsymbol{\mathcal{V}},\delta \boldsymbol{\mathcal{V}}.
\end{equation}

As presented by \v{S}ilhav\'{y} in \cite{Silhavy2017},  $\mathcal{A}$-polyconvexity in conjunction with suitable growth conditions, ensures the existence of minimisers in nonlinear electro-magneto-elasticity, which is briefly recalled in the following section. 

\subsection{$\mathcal{A}$-polyconvexity and existence of minimisers}

The extended set of arguments in $W$, that is $\boldsymbol{\mathcal{V}}$, has a special property; namely, they are weakly sequentially continuous \cite{Ciar88ME1}. 
\begin{proposition}\label{Prop:det-cof}
Let $\Omega_0\subset\R^3$ be a bounded Lipschitz domain. Let $(\bm{\phi}_k)_{k\in\N}\subset W^{1,p}(\Omega_0;\R^3)$ for some $p>3$ and let $\bm{\phi}_k\to \bm{\phi}$ weakly for some $\bm{\phi}\in W^{1,p}(\Omega_0;\R^3)$. 
Then 
\begin{subequations}
\begin{align}
    \cof\bm{\nabla}_0 \bm{\phi}_k &\to \cof\bm{\nabla}_0 \bm{\phi}
    \,\text{weakly in }\, L^{p/2}(\Omega_0;\R^{3\times 3})\ ,\\
    \det\bm{\nabla}_0 \bm{\phi}_k &\to \det\bm{\nabla}_0 \bm{\phi}
    \,\text{weakly in }\, L^{p/3}(\Omega_0)\ .
\end{align}
\end{subequations}
\end{proposition}
%


\begin{proposition}\label{prop:div-curl}
Let $\Omega_0\subset\R^3$ be a bounded Lipschitz domain. Let $\bm{\phi}_k\to \bm{\phi}$ weakly  in $W^{1,p}(\Omega_0;\R^3)$ for some $1<p<+\infty$ and let $\bm{D}_{0,k}\to \bm{D}_0$ weakly in $L^q(\Omega_0;\R^3)$  where $1/q+1/p<1$  and $\mathrm{DIV}\, \bm{D}_{0,k}=\mathrm{DIV}\, \bm{D}_0=\rho_0$ for every $k\in\N$ and for some $\rho_0\in W^{-1,1}(\Omega)$. Then $\vect{d}_k=\vect{F}_k\bm{D}_{0,k} \to \vect{d}=\vect{F}\bm{D}_0$ weakly in $L^1(\Omega_0;\R^3)$, where $\vect{F}_k=\bm{\nabla} \bm{\phi}_k$ and $\vect{F}=\bm{\nabla} \bm{\phi}$.
\end{proposition}


Referring now to the variational principle $\Pi_e$ previously presented in \eqref{eqn:mixed variational principle Ib}, it is now possible to introduce the following theorem of existence of minimisers. Indeed,

\begin{theorem}\label{th:Existence}
Let $\Omega_0\subset\R^3$ be a bounded Lipschitz domain. Let $p>2$, $p_c\ge p/(p-1)$, $r>1$, and  $q>p/(p-1)$, and let the following \textbf{growth conditions} be satisfied for some constant $C>0$.
\begin{align}\label{growth}
    e(\boldsymbol{\mathcal{U}})\begin{cases}\ge C(||\bm{F}||^p+ ||\vect{H}||^{p_c}+J^r+||{\bm D}_0||^q)  \text{ if } J>0\ ,\\
    =+\infty \text{ otherwise.}
    \end{cases}
    \end{align}
    Moreover, assume that $e$ is \textbf{$\mathcal{A}$-polyconvex}, i.e., that \eqref{polyconvexity_emm} holds for some convex and lower semicontinuous function $W:\mathbb{R}^{3\times3}\times\mathbb{R}^{3\times3}\times\mathbb{R}\times\mathbb{R}^3\times\mathbb{R}^3 \to \mathbb{R}\cup\{+\infty\}$ and that 
    \begin{align}
    e(\boldsymbol{\mathcal{U}})=W(\boldsymbol{\mathcal{V}})\to +\infty \text{ if } J\to 0.
    \end{align}
    Let 
\begin{equation}
\begin{aligned}
\mathcal{A}=\{&\bm{\phi}\in  W^{1,p}(\Omega_0;\R^3),\;\; \bm{D}_0\in L^q(\Omega_0;\R^3); \nonumber\\
&s.t. \;\vect{H}\in L^{p_c}(\Omega_0;\R^{3\times 3});\;\;J\in L^r(\Omega_0)\, ,\,J>0 \text{ a.e.};\nonumber\\
&\text{s.t}.\; \eqref{eqn:Dirich_displacements},\eqref{eq:DivD_a},\eqref{eq:DivD_b}\}
\end{aligned}
\end{equation}
be non-empty and such that ${\rm inf}_{\mathcal{A}}\Pi_e<+\infty$. Then there is a \textbf{minimiser} of $\Pi_e$ on $\mathcal{A}$.
\end{theorem}

\begin{proof}
The proof follows  the same lines as in \cite{Ciar88ME1}. If $(\vect{\phi}_k,(\bm{D}_0)_k)\subset W^{1,p}(\Omega_0;\R^3)\times L^q(\Omega_0;\R^3)$ is a minimizing sequence then \eqref{growth} and the Dirichlet boundary conditions ensure that this minimizing sequence is bounded  in $W^{1,p}(\O;\R^3)\times L^q(\O;\R^3)$. Therefore, it holds for a non-relabeled subsequence that $\vect{\phi}_k\to \vect{\phi}$ weakly and $(\bm{D}_0)_k\to\bm{D}_0$ weakly. Further,   
$\vect{F}_k(\bm{D}_0)_k\to \vect{F}(\bm{D}_0)$ weakly in $L^{pq/(p+q)}(\O;\R^3)$ due to Proposition~\ref{prop:div-curl}. 
\end{proof}

\subsection{$\mathcal{A}$-polyconvexity and material stability: extended set of work conjugates and the tangent operator}

As presented in \cite{Gil2016_NewFramework, Ortigosa2016_NewFramework_FiniteElements, Ortigosa2016_NewFramework_ConservationLaws, Ortigosa2016_ComputationalFramework}, the consideration of the function $W$ (instead of its equivalent energetic expression $e)$ permits the introduction of the set $\boldsymbol{\Sigma}_{\boldsymbol{\mathcal{V}}}$ of work conjugate fields $\boldsymbol{\Sigma}_{\boldsymbol{\mathcal{V}}}=(\boldsymbol{\Sigma}_{\vect{F}},\boldsymbol{\Sigma}_{\vect{H}},\Sigma_J,\boldsymbol{\Sigma}_{\vect{D_0}},\boldsymbol{\Sigma}_{\vect{d}})$ defined as
\begin{equation}\label{eqn:work_conjugates}
DW[\delta \boldsymbol{\mathcal{V}}]=\boldsymbol{\Sigma}_{\boldsymbol{\mathcal{V}}} \bullet \delta \boldsymbol{\mathcal{V}}; \qquad\boldsymbol{\Sigma}_{\boldsymbol{\mathcal{V}}}(\boldsymbol{\mathcal{V}})=\partial_{\boldsymbol{\mathcal{V}}}W(\boldsymbol{\mathcal{V}}).
\end{equation}

Moreover, for the particular case when $\delta \boldsymbol{\mathcal{V}}=D\boldsymbol{\mathcal{V}}[\delta \boldsymbol{\mathcal{U}}]$, comparing equations \eqref{eqn:directional derivative} and \eqref{eqn:work_conjugates} yields \cite{Gil2016_NewFramework}
\begin{equation}
\vect{P}=\boldsymbol{\Sigma}_{\vect{F}} + \boldsymbol{\Sigma}_{\vect{H}} \Cross \vect{F} + \Sigma_J \vect{H} + \boldsymbol{\Sigma}_{\vect{d}} \otimes \vect{D}_0;\qquad \vect{E}_0=\boldsymbol{\Sigma}_{\vect{D_0}} +\vect{F}^T \boldsymbol{\Sigma}_{\vect{d}},
\end{equation}
which, for the case of a suitably defined convex energy function $W$, facilitates the evaluation of $\vect{P}$ and $\vect{E}_0$ in terms of the sets $\boldsymbol{\mathcal{V}}$ and $\boldsymbol{\Sigma_{\boldsymbol{\mathcal{V}}}}$. Assuming sufficient regularity of the function $W$, its second directional derivative yields the Hessian operator $\mathbb{H}_{W}$ as
\begin{equation}
D^2W[\delta \boldsymbol{\mathcal{V}};\delta \boldsymbol{\mathcal{V}}]=\delta \boldsymbol{\mathcal{V}} \bullet \mathbb{H}_{W} \bullet \delta \boldsymbol{\mathcal{V}},
\end{equation}
By definition, $\mathcal{A}$-polyconvexity implies semi-positive definiteness of the Hessian operator $\mathbb{H}_{W}$ (refer to \eqref{eqn:polyconvexity2}). Moreover, as also shown in \cite{Gil2016_NewFramework, Ortigosa2016_NewFramework_FiniteElements, Ortigosa2016_NewFramework_ConservationLaws, Ortigosa2016_ComputationalFramework}, the second directional derivative of the internal energy $e$ \eqref{tangent_operator_1} can be equivalently expressed in terms of its extended representation $W$ as 
\begin{equation}\label{tangent_operator}
\begin{aligned}
D^2e\left[\delta\boldsymbol{\mathcal{U}};\delta\boldsymbol{\mathcal{U}}\right]&= 
D^2 W[D\boldsymbol{\mathcal{V}}[\delta \boldsymbol{\mathcal{U}}];D\boldsymbol{\mathcal{V}}[\delta \boldsymbol{\mathcal{U}}]]+
\left(\boldsymbol{\Sigma}_{\vect{H}}+\Sigma_{J}\vect{F}\right):\left(\delta\vect{F}\Cross\delta\vect{F}\right)
+2\boldsymbol{\Sigma}_{\vect{d}}\cdot\delta\vect{F}\,\delta\vect{D}_0,
\end{aligned}
\end{equation}
where
\begin{equation}
D^2 W[D\boldsymbol{\mathcal{V}}[\delta \boldsymbol{\mathcal{U}}];D\boldsymbol{\mathcal{V}}[\delta \boldsymbol{\mathcal{U}}]]=\mathbb{S}^T\mathbb{H}_{W}\mathbb{S},
\end{equation}
with $\mathbb{S}$ defined as 
\begin{equation}
\mathbb{S}=\left[\begin{matrix}
: \delta\vect{F}\\
: (\delta\vect{F}\times\vect{F})\\
\delta\vect{F}:\vect{H}\\
\cdot \delta\vect{D}_0\\
\cdot \left(\delta\vect{F}\vect{D}_0+\vect{F}\delta\vect{D}_0\right)\ 
\end{matrix}\right].
\end{equation}

As presented in \cite{Ortigosa2016_ComputationalFramework}, it is possible to relate the components of the Hessian $\mathbb{H}_e$ in \eqref{tangent_operator_1} to those of the Hessian $\mathbb{H}_W$ via appropriate algebraic transformations, which is advantageous for the case of energetic expressions defined in terms of $W$. Moreover, replacing $\delta\vect{F}$ and $\delta\vect{D}_0$ in \eqref{tangent_operator} with 
$\delta\vect{F}=\vect{u}\otimes\vect{V}$ and $\delta\vect{D}_0=\vect{V}_{\perp}$, respectively, and noticing that $De[\delta \boldsymbol{\mathcal{U}}]=DW[D\boldsymbol{\mathcal{V}}[\delta \boldsymbol{\mathcal{U}}]]$, implies that equation \eqref{eqn:pol2} reduces to  \eqref{eqn:rank1_der1}, which shows that $\mathcal{A}$-polyconvexity implies generalised rank-one convexity (in the more generic sense of electromechanics). Moreover, replacing $\delta\vect{F}$ and $\delta\vect{D}_0$ in \eqref{tangent_operator} with 
$\delta\vect{F}=\vect{u}\otimes\vect{V}$ and $\delta\vect{D}_0=\vect{V}_{\perp}$, respectively, permits to cancel the last two terms (also known as geometric terms) on the right-hand side of \eqref{tangent_operator}, and leads to
\begin{equation}\label{eqn:ellipticity extended}
D^2 e[\vect{u}\otimes\vect{V},\vect{V}_{\perp};\vect{u}\otimes\vect{V},\vect{V}_{\perp}] = \mathbb{S}^T\mathbb{H}_{W}\mathbb{S} \geq 0.
\end{equation}

Relation \eqref{eqn:ellipticity extended} clearly illustrates that smooth $\mathcal{A}$-polyconvex internal energy functionals, characterised by a semi-positive definite Hessian operator $\mathbb{H}_{W}$, guarantee semi-positiveness of the left-hand side of equation \eqref{eqn:ellipticity extended}, and hence the fulfilment of the ellipticity condition in \eqref{eqn:ellipticity condition_1}.

\section{Polyconvex transversely isotropic invariant-based electromechanics}\label{sec:invariants}
In addition to the requirement of $\mathcal{A}$-polyconvexity (i.e. existence of minimisers and material stability), the internal energy density $e$ must be defined satisfying additional requirements. The following subsection summarises them leading to an invariant-based representation of the energy density.

\subsection{Invariant energy representation}

Objectivity or frame invariance implies independence of the energy density with respect to arbitrary rotations $\vect{Q}$ of the spatial configuration, which can be formulated as
\begin{equation}\label{eqn:frame_invariance}
e(\vect{Q}\vect{F},\vect{D}_0)=e(\vect{F},\vect{D}_0);\qquad \forall \,\vect{Q}\in \text{SO}(3),
\end{equation}
with $\text{SO}(3)$ the special (proper) orthogonal group \cite{Gurt81ICM}. Notice that in above objectivity requirement, the rotation tensor $\vect{Q}$ is only applied to the two-point tensor $\vect{F}$ and not to the Lagrangian vector $\vect{D}_0$. As it is well-known, the requirement of objectivity implies that the internal energy $e$ must be expressible in terms of an objective set of arguments as  $e (\vect{F},\vect{D}_0) = \tilde{e} (\vect{C},\vect{D}_0)$, where $\vect{C}=\vect{F}^T\vect{F}$ denotes the right Cauchy-Green strain tensor and with $e$ and $\tilde{e}$ denoting alternative functional representations of the same internal energy density. 

In addition, for the case of transverse isotropy\footnote{Or any other kind of anisotropy characterised by a possibly more complex crystal symmetry structure.} characterised by a material unit vector $\vect{M}$, the energy must be independent with respect to rotations/reflections $\vect{Q}\in \mathcal{G}_{\vect{m}}\subset \text{O}(3)$ of the material configuration, where $ \mathcal{G}_{\vect{m}}$ denotes the corresponding symmetry group and $\text{O}(3)$ the full orthogonal group.  This anisotropic restriction is formulated as
\begin{equation}\label{eqn:transverse_iso}
\tilde{e}(\vect{Q}\vect{C}\vect{Q}^T,\vect{Q}\vect{D}_0)=\tilde{e}(\vect{C},\vect{D}_0);\qquad \forall \,\vect{Q}\in \mathcal{G}_{\vect{m}} \subset \text{O}(3).
\end{equation}

As stated e.g.~in \cite{SchroederGross04}, there exist five types of transverse isotropy. In this work, we restrict ourselves to those that affect most EAPS of interest, that is, the symmetry $\mathcal{D}_{\infty h}$, which corresponds to the usual definition of transverse isotropy (e.g. applicable to electro-active polymers), and the symmetry $\mathcal{C}_{\infty}$, the so-called rotational symmetry (e.g. applicable to electro-active materials exhibiting piezoelectric effects). Please refer to \cite{Zheng_invariants} for further details regarding the various types of transverse isotropy. Moreover, these two symmetry groups can be characterised by appropriate \textit{structural tensors} which encapsulate the symmetry attributes of their corresponding groups, which in this case are given by the second-order tensor $\vect{M}\otimes \vect{M}$ (for $\mathcal{D}_{\infty h}$) and the first-order tensor (vector) $\vect{M}$ (for $\mathcal{C}_{\infty}$).  

By making use of the \textit{isotropicisation theorem} \cite{Zheng_invariants}, it is now possible to re-express above anisotropic restriction \eqref{eqn:transverse_iso} for every $\vect{Q}\in \text{O}(3)$ as
\begin{subequations}
\begin{align}
\bar{e}_{\mathcal{D}_{\infty h}}(\vect{Q}\vect{C}\vect{Q}^T, \vect{Q}(\vect{D}_0\otimes\vect{D}_0)\vect{Q}^T, \vect{Q}(\vect{M}\otimes \vect{M})\vect{Q}^T)&=\bar
{e}_{{\mathcal{D}_{\infty h}}}(\vect{C},\vect{D}_0\otimes\vect{D}_0,\vect{M}\otimes \vect{M})
\\
\bar{e}_{\mathcal{C}_{\infty}}(\vect{Q}\vect{C}\vect{Q}^T, \vect{Q}\vect{D}_0, \vect{Q}\vect{M})&=\bar{e}_{\mathcal{C}_{\infty}}(\vect{C},\vect{D}_0,\vect{M})
\end{align}
\end{subequations}
where $\bar{e}_{\mathcal{D}_{\infty h}}$ and $\bar{e}_{\mathcal{C}_{\infty}}$ represent alternative isotropic functional representations of the internal energy density $e$ particularised for their corresponding transverse isotropy groups ${\mathcal{D}_{\infty h}}$ and ${\mathcal{C}_{\infty}}$, respectively. It is now straightforward to obtain an irreducible list of isotropic invariants (a so-called \textit{integrity basis}) for the characterisation of above energy densities $\bar{e}_{\mathcal{D}_{\infty h}}$ and $\bar{e}_{\mathcal{C}_{\infty}}$.  For the case of $\bar{e}_{\mathcal{D}_{\infty h}}$ (refer to Table 12 in \cite{Zheng_invariants}), we obtain
\begin{equation}\label{eqn:basis1}
\begin{aligned}
\text{isotropic elasticity:} \quad &\text{tr} [\vect{C}];  \text{tr} [\vect{C}^2]; \text{tr} [\vect{C}^3]; \\
+\text{isotropic electro:} \quad &\text{tr} [\vect{D}_0 \otimes \vect{D}_0]; \text{tr} [\vect{C}(\boldsymbol{D}_0 \otimes \vect{D}_0)]; \text{tr} [\vect{C}^2(\boldsymbol{D}_0 \otimes \vect{D}_0)];\\
+\text{transverse-isotropic elasticity:} \quad&\text{tr} [\vect{C}(\vect{M}\otimes \vect{M})]; \text{tr} [\vect{C}^2(\vect{M}\otimes \vect{M})];\\
+\text{transverse-isotropic electro:} \quad&\text{tr}[(\boldsymbol{D}_0\otimes \vect{D}_0)(\vect{M}\otimes \vect{M})];\text{tr}[(\vect{D}_0 \otimes \vect{D}_0)\vect{C}(\vect{M}\otimes \vect{M})],
\end{aligned}
\end{equation}
and for the case of  $\bar{e}_{\mathcal{C}_{\infty}}$ (refer to Table 12 in \cite{Zheng_invariants}), we obtain\footnote{In the original Table 12 in \cite{Zheng_invariants}, the invariant $\text{tr} [\vect{C}^2(\boldsymbol{D}_0 \otimes \vect{D}_0)]$ is missing whilst the invariant $\text{tr} [\vect{C}^2(\boldsymbol{D}_0 \otimes \vect{M})]$ is included in the basis. Notice that it is possible to show (refer to \cite{Bustamante2008} and references therein) that either invariant can be expressed in terms of the rest of the invariants of the basis.}
\begin{equation}\label{eqn:basis2}
\begin{aligned}
\text{isotropic elasticity:} \quad &\text{tr} [\vect{C}];  \text{tr} [\vect{C}^2]; \text{tr} [\vect{C}^3]; \\
+\text{isotropic electro:} \quad &\text{tr} [\vect{D}_0 \otimes \vect{D}_0]; \text{tr} [\vect{C}(\boldsymbol{D}_0 \otimes \vect{D}_0)]; \text{tr} [\vect{C}^2(\boldsymbol{D}_0 \otimes \vect{D}_0)]\\
+\text{transverse-isotropic elasticity:} \quad&\text{tr} [\vect{C}(\vect{M} \otimes \vect{M})]; \text{tr} [\vect{C}^2(\vect{M}\otimes \vect{M})];\\
+\text{transverse-isotropic electro:} \quad& \text{tr} [\vect{D}_0 \otimes \vect{M}];  \text{tr} [\vect{C}(\boldsymbol{D}_0 \otimes \vect{M})].
\end{aligned}
\end{equation}

Notice that both integrity bases \eqref{eqn:basis1} and \eqref{eqn:basis2} are comprised of ten invariants, only differing in the expression of the two invariants responsible for the transversely isotropic electro-mechanic effect. It is now possible to re-express the above energy densities in terms of the invariants of the above integrity bases. Specifically, for the case of $\bar{e}_{\mathcal{D}_{\infty h}}$, it yields
\begin{equation}
\bar{e}_{\mathcal{D}_{\infty h}}(\vect{C},\vect{D}_0\otimes\vect{D}_0,\vect{M}\otimes \vect{m})=\hat{e}_{\mathcal{D}_{\infty h}}(J_1,J_2,J_3,J_4,J_5,J_6,J_7,J_8,K_1^{\mathcal{D}_{\infty h}},K_2^{\mathcal{D}_{\infty h}}),
\end{equation}
where $J_i,\{i=1\ldots8\}$ denote the following set of invariants 
\begin{equation}
\begin{gathered}
J_1=\text{tr} [\vect{C}];\quad
J_2=\text{tr} [\vect{C}^2];\quad
J_3=\text{tr} [\vect{C}^3];\quad
J_4=\text{tr} [\vect{C}(\vect{M}\otimes \vect{M})];\quad
J_5=\text{tr} [\vect{C}^2(\vect{M}\otimes \vect{M})];\\
J_6=\text{tr} [\vect{D}_0 \otimes \vect{D}_0]; \quad
J_7=\text{tr} [\vect{C}(\boldsymbol{D}_0 \otimes \vect{D}_0)];\quad
J_8=\text{tr} [\vect{C}^2(\vect{D}_0\otimes \vect{D}_0)],
\end{gathered}
\end{equation}
and $K_1^{\mathcal{D}_{\infty h}},K_2^{\mathcal{D}_{\infty h}}$ are
\begin{equation}
\begin{gathered}
K_1^{\mathcal{D}_{\infty h}}=\text{tr} [(\boldsymbol{D}_0\otimes \vect{D}_0)(\vect{M}\otimes \vect{M})];\quad
K_2^{\mathcal{D}_{\infty h}}=\text{tr}[(\vect{D}_0 \otimes \vect{D}_0)\vect{C}(\vect{M}\otimes \vect{M})],
\end{gathered}
\end{equation}
note that $J_{1}\ldots J_5$ denote the mechanical invariants, $J_6\ldots J_8$ denote the invariants introducing (isotropic) electro-mechanical effects and $K_1^{\mathcal{D}_{\infty h}}$ and $K_2^{\mathcal{D}_{\infty h}}$ are the two invariants responsible for introducing fibre dependent electro-mechanical effects. 

Similarly, in the case of $\bar{e}_{\mathcal{C}_{\infty}}$, it yields
\begin{equation}
\bar{e}_{\mathcal{C}_{\infty}}(\vect{C},\vect{D}_0,\vect{M})=\hat{e}_{\mathcal{C}_{\infty}}(J_1,J_2,J_3,J_4,J_5,J_6,J_7,J_8,K_1^{\mathcal{C}_{\infty }},K_2^{\mathcal{C}_{\infty }}),
\end{equation}
where the alternative fibre dependent electro-mechanical invariants are
\begin{equation}
K_1^{\mathcal{C}_{\infty }}=\text{tr} [\vect{D}_0 \otimes \vect{M}];\quad
K_2^{\mathcal{C}_{\infty }}=\text{tr} [\vect{C}(\boldsymbol{D}_0 \otimes \vect{M})].
\end{equation}

A further restriction, albeit not strictly necessary, is that of zero energy at the origin (i.e. $\boldsymbol{\mathcal{U}}_0=(\vect{F}=\vect{I},\vect{D}_0=\vect{0})$), that is, 
$e(\boldsymbol{\mathcal{U}}_0)=0$. However, an important restriction is that of zero stresses and electric field at the origin, namely,
\begin{equation}
\vect{P}(\boldsymbol{\mathcal{U}}_0)=\left.\partial_{\vect{F}} e\right|_{\boldsymbol{\mathcal{U}}_0}=\vect{0};\qquad
\vect{E}_0(\boldsymbol{\mathcal{U}}_0)=\left.\partial_{\vect{D}_0} e\right|_{\boldsymbol{\mathcal{U}}_0}=\vect{0}.
\end{equation}

It is interesting to observe that for EAPs exhibiting symmetry of the group $\mathcal{D}_{\infty h}$, the electric field $\vect{E}_0$ evaluated at the origin of electric displacements, that is,  $\vect{D}_0=\vect{0}$ (regardless of the value of deformation), is zero, which is indeed the case as these materials do not experience any piezoelectric effects. This is obtained implicitly due to the dependence of the integrity basis on the second order tensor $\vect{D}_0 \otimes \vect{D}_0$ (refer to invariants $J_6\ldots J_{10}$), that is,
\begin{equation}\label{eqn:piezo}
\vect{E}_0(\vect{F},\vect{D}_0=\vect{0})=\sum_{i=6}^{10}\left.[(\partial_{I_i} \bar{e}_{\mathcal{D}_{\infty h}}) (\partial_{\vect{D}_0 \otimes \vect{D}_0} I_i)]\right|_{(\vect{F},\vect{0})} 2\underbrace{\vect{D}_0}_{=\vect{0}}=\vect{0}.
\end{equation}

However, for electro-active materials of the group $\mathcal{C}_{\infty}$, this restriction does not apply. Indeed, invariants $K^{\mathcal{C}_{\infty }}_1, K^{\mathcal{C}_{\infty }}_2$ depend linearly on $\vect{D}_0$ and not on $\vect{D}_0 \otimes \vect{D}_0$, which permits lifting of the restriction \eqref{eqn:piezo}.

Finally, material characterisation typically requires calibration of material parameters at $\boldsymbol{\mathcal{U}}_0$, that is
\begin{equation}
\vect{\mathcal{C}}_e(\boldsymbol{\mathcal{U}}_0)=\left.\partial^2_{\vect{FF}}e\right|_{\boldsymbol{\mathcal{U}}_0}=\vect{\mathcal{C}}_e^{\text{lin}};\quad
\vect{\mathcal{Q}}(\boldsymbol{\mathcal{U}}_0)=\left.\partial^2_{\vect{D}_0\vect{F}}e\right|_{\boldsymbol{\mathcal{U}}_0}=\vect{\mathcal{Q}}^{\text{lin}};\quad
\vect{\mathcal{\theta}}(\boldsymbol{\mathcal{U}}_0)=\left.\partial^2_{\vect{D}_0\vect{D}_0}e\right|_{\boldsymbol{\mathcal{U}}_0}=\vect{\theta}^{\text{lin}}, 
\end{equation}
or in case of using the Hessian components of the Helmholtz's free energy function $[\mathbb{H}_{\Psi}]$
\begin{equation}
\vect{\mathcal{C}}_{\Psi}(\boldsymbol{\mathcal{W}}_0)=\left.\partial^2_{\vect{FF}}\Psi\right|_{\boldsymbol{\mathcal{W}}_0}=\vect{\mathcal{C}}_{\Psi}^{\text{lin}};\quad
\vect{\mathcal{P}}(\boldsymbol{\mathcal{W}}_0)=\left.\partial^2_{\vect{D}_0\vect{F}}\Psi\right|_{\boldsymbol{\mathcal{W}}_0}=\vect{\mathcal{P}}^{\text{lin}};\quad
\vect{\mathcal{\epsilon}}(\boldsymbol{\mathcal{W}}_0)=\left.\partial^2_{\vect{E}_0\vect{E}_0}\Psi\right|_{\boldsymbol{\mathcal{W}}_0}=\vect{\mathcal{\epsilon}}^{\text{lin}},
\end{equation}
where $\boldsymbol{\mathcal{W}}_0=(\vect{F}=\vect{I},\vect{E}_0=\vect{0})$.

\subsection{$\mathcal{A}$-polyconvex transversely isotropic invariant representations}

It is possible to re-express the above integrity bases of invariants $J_i,\{i=1\ldots8\}$, $K_j^{\mathcal{D}_{\infty h}}$, $K_j^{\mathcal{C}_{\infty }}$, $\{j=1,2\}$ in terms of an alternative basis of invariants more amenable to the study of $\mathcal{A}$-polyconvexity. The purely mechanical isotropic invariants can be re-defined as
\begin{equation}
I_1=J_1=\vect{F}:\vect{F};\qquad I_2=\vect{H}:\vect{H};\qquad I_3=J,
\end{equation}
where
\begin{equation}
J_2=I_1^2 -2I_2; \qquad J_3=I_1^3-3I_1I_2+3I_3^2.
\end{equation}

The transversely isotropic mechanical invariants can be re-defined as
\begin{equation}
I_4=J_4=(\vect{F}\vect{M}) \cdot (\vect{F}\vect{M});\qquad I_5=(\vect{H}\vect{M}) \cdot (\vect{H}\vect{M});
\end{equation}
where
\begin{equation}
J_5=I_1I_4-I_2+I_5.
\end{equation}

Similarly, the isotropic electro-mechanical invariants can be re-defined as
\begin{equation}
I_6=J_6=\vect{D}_0\cdot \vect{D}_0; \qquad I_7=J_7=\vect{d}\cdot \vect{d};\qquad I_8=(\vect{H}\vect{D}_0) \cdot (\vect{H} \vect{D}_0),
\end{equation}
where
\begin{equation}
J_8=I_1I_4-I_2+I_8,
\end{equation}
and the transversely isotropic electro-mechanical contributions as
\begin{equation}
K_1^{\mathcal{D}_{\infty h}}=(\vect{D}_0\cdot \vect{M})^2;\qquad 
K_2^{\mathcal{D}_{\infty h}}=(\vect{d} \cdot \vect{F}\vect{M})^2,
\end{equation}
and
\begin{equation}\label{eq:K^C}
K_1^{\mathcal{C}_{\infty }}=\vect{D}_0\cdot \vect{M};\qquad 
K_2^{\mathcal{C}_{\infty }}=\vect{d}\cdot \vect{F}\vect{M}.
\end{equation}

Thus, energies $\bar{e}_{\mathcal{D}_{\infty h}}$ and $\bar{e}_{\mathcal{C}_{\infty}}$ can be re-written as
\begin{subequations}\label{eqn:energies1}
\begin{align}
\bar{e}_{\mathcal{D}_{\infty h}}(\vect{C},\vect{D}_0,\vect{M})&=\hat{\hat{e}}_{\mathcal{D}_{\infty h}}(I_1,I_2,I_3,I_4,I_5,I_6,I_7,I_8,K_1^{\mathcal{D}_{\infty h}},K_2^{\mathcal{D}_{\infty h}});\\
\bar{e}_{\mathcal{C}_{\infty}}(\vect{C},\vect{D}_0,\vect{M})&=\hat{\hat{e}}_{\mathcal{C}_{\infty}}(I_1,I_2,I_3,I_4,I_5,I_6,I_7,I_8,K_1^{\mathcal{C}_{\infty}},K_2^{\mathcal{C}_{\infty}}).
\end{align}
\end{subequations}

It is customary to split the energy density into mechanical and electro-mechanical components, further splitting the former into deviatoric and volumetric components, via suitable modification of the invariants $I_1$ and $I_2$. This can be formulated as follows
\begin{subequations}\label{eqn:energies2}
\begin{align}
\bar{e}_{\mathcal{D}_{\infty h}}(\vect{C},\vect{D}_0,\vect{m})&=\tilde{\hat{e}}^{\text{dev,mec}}_{\mathcal{D}_{\infty h}}(I_1^{\text{dev}},I_2^{\text{dev}},I_4,I_5) + e^{\text{vol}}(I_3)\nonumber\\
&+\tilde{\hat{e}}^{\text{ele}}_{\mathcal{D}_{\infty h}}(I_1,I_2,I_3,I_4,I_5,I_6,I_7,I_8,K_1^{\mathcal{D}_{\infty h}},K_2^{\mathcal{D}_{\infty h}});\\
\bar{e}_{\mathcal{C}_{\infty }}(\vect{C},\vect{D}_0,\vect{m})&=\tilde{\hat{e}}^{\text{dev,mec}}_{\mathcal{C}_{\infty }}(I_1^{\text{dev}},I_2^{\text{dev}},I_4,I_5) + e^{\text{vol}}(I_3) \nonumber\\
&+ \tilde{\hat{e}}^{\text{ele}}_{\mathcal{C}_{\infty }}(I_1,I_2,I_3,I_4,I_5,I_6,I_7,I_8,K_1^{\mathcal{C}_{\infty }},K_2^{\mathcal{C}_{\infty }}),
\end{align}
\end{subequations}
where the different energetic contributions are identified by corresponding upper indices and where
\begin{equation}
I_1^{\text{dev}}=I_3^{-2/3}I_1=J^{-2/3}\vect{F}:\vect{F};\qquad I_2^{\text{dev}}=I_3^{-4/3}I_2=J^{-4/3}\vect{H}:\vect{H}.
\end{equation}

It is now the objective to derive  that is, 1) the volumetric energy $e^{\text{vol}}$ is selected as a convex function and 2) the deviatoric energy densities are selected as a linear combination of invariants with suitable coefficients such that the overall energy density is $\mathcal{A}$-polyconvex. One possibility how to construct such an energy is to consider only $\mathcal{A}$-polyconvex invariants  $I_1,I_2,I_1^{\text{dev}}$, $I_4$, $I_5$, $I_6$, $I_7$, $K_1^{\mathcal{D}_{\infty h}}$, and $K_1^{\mathcal{C}_{\infty}}$ with non-negative weights. However, invariants $I_2^{\text{dev}}$, $I_8$, $K_2^{\mathcal{D}_{\infty h}}$, and $K_2^{\mathcal{C}_{\infty}}$ are not rank-one convex and, hence, not polyconvex. To incorporate also non $\mathcal{A}$-polyconvex invariants into the energy, one needs to make sure that the overall energy is $\mathcal{A}$-polyconvex.

Yet another possibility is to polyconvexify the non $\mathcal{A}$-polyconvex invariants, cf \cite{vsilhavy2018variational}. However, polyconvexification is in general difficult and may lead to non-coercive terms. Thus, some authors, e.g., \cite{Gil2016_NewFramework}, propose ad-hoc nonlinear modification of non-polyconvex invariants leading to polyconvex terms. Such modification has to be designed case by case, e.g., 
$I_{2}^{\text{dev}}$ can be modified to $I_{2,\text{pol}}^{\text{dev}} = \left(I_{2}^{\text{dev}}\right)^{3/2} =  J^{-2}\left(\vect{H}:\vect{H}\right)^{3/2}$ which is polyconvex. Similarly, a possible polyconvex invariant expression for $I_{8}$ is  
\begin{equation}
I_{8,\text{pol}}=\alpha^2 \left(\vect{H}:\vect{H}\right)^2 + \beta^2 \left(\vect{D}_0\cdot \vect{D}_0\right)^2 +\alpha\beta (\vect{H}\vect{D}_0)\cdot(\vect{H}\vect{D}_0)=\alpha^2 (I_2)^2+\beta^2 (I_6)^2+\alpha\beta (I_8),
\end{equation}
where $\alpha,\beta$ are positive material constants. Analogously, a possible modification of $K_{2}^{\mathcal{C}_{\infty}}$ is
\begin{equation}
K_{2,\text{pol}}^{\mathcal{C}_{\infty}}=(\eta \vect{d} +\psi \vect{F}\vect{M})\cdot (\eta \vect{d} +\psi \vect{F}\vect{M})=\eta^2 I_7 +\psi^2 I_4 +2\eta \psi K_{2}^{\mathcal{C}_{\infty}},
\end{equation}
where $\eta,\psi$ are positive material constants. The proof of polyconvexity of invariants $I_{8,\text{pol}}$ and $K_{2,\text{pol}}^{\mathcal{C}_{\infty}}$ is included in Appendix \ref{appendix:polyconvex_invariants}.
Modification of $K_2^{\mathcal{D}_{\infty h}}$ can be obtained by adding suitable convex terms depending on $\vect{F}$ and $\vect{D}_0$. However, the physical relevance of the resulting energy is questionable.

 Additionally, in order to fulfill assumptions of Theorem \ref{th:Existence}, we need to ensure that the energy is coercive, cf. (\ref{growth}). For example, invariant $K_1^{\mathcal{C}_{\infty }}=\vect{D}_0\cdot \vect{M}$ is $\mathcal{A}$-polyconvex, but not coercive. Notice that $\vect{D}_0\cdot\vect{M} = 0$ if $\vect{D}_0$ is perpendicular to $\vect{M}$. To ensure coercivity of the overall energy $e$, we must combine non-coercive invariants with coercive ones, e.g., $aJ_6 + bK_{1}^{\mathcal{C}_{\infty}}$ is coercive for $a > 0$.

\section{Numerical examples}\label{sec:examples}

This section presents a series of numerical examples modelling  the performance of transversely isotropic EAPs at large strains. The first numerical example, restricted to the case of homogeneous deformation, circumvents the need to use a Finite Element (FE) spatial discretisation and studies the behaviour, at a local level, of the response of the EAP. A comprehensive study will be conducted where the influence of the deformation and the level of transverse anisotropy has on the stability of the model. In addition, polyconvex and non-polyconvex constitutive models will be studied and compared for a range of deformations and electric fields, emphasising that seemingly similar models might lead to unexpected results. The second numerical example abandons the
assumption of uniform deformation and explores the use of three-dimensional Finite Elements in the case of multi-layered transversely isotropic EAPs in the form of thin rectangular films subjected to appropriate mechanical and electrical boundary conditions. The influence of the orientation of the vector of anisotropy and electro-mechanical properties will be studied, especially into their effect on the stability of the model, so as to prevent the onset of spurious mesh dependence results. From the FE standpoint we adopt an enhanced $(\vect{\phi},\vect{D}_0,\varphi)$ formulation (briefly recalled in Appendix \ref{appendix:FEM} for completeness).   
\subsection{Numerical example 1}
Through this example we aim to:
\begin{itemize}
\item\label{obj Exp1: different performance} To analyse the response of $\mathcal{A}$-polyconvex and non-$\mathcal{A}$-polyconvex energy functionals in uhomogeneous states of deformation and electric fields.

\item\label{obj Exp1: arc-length} Appreciate the importance of using arc-length techniques in order to bypass instability regions and harness actuator performance beyond the moderate regime.

\end{itemize}


\begin{figure}[h!]
	\centering
	\begin{tabular}{c}
		\hspace{-0.6cm}		\includegraphics[width=0.95\textwidth]{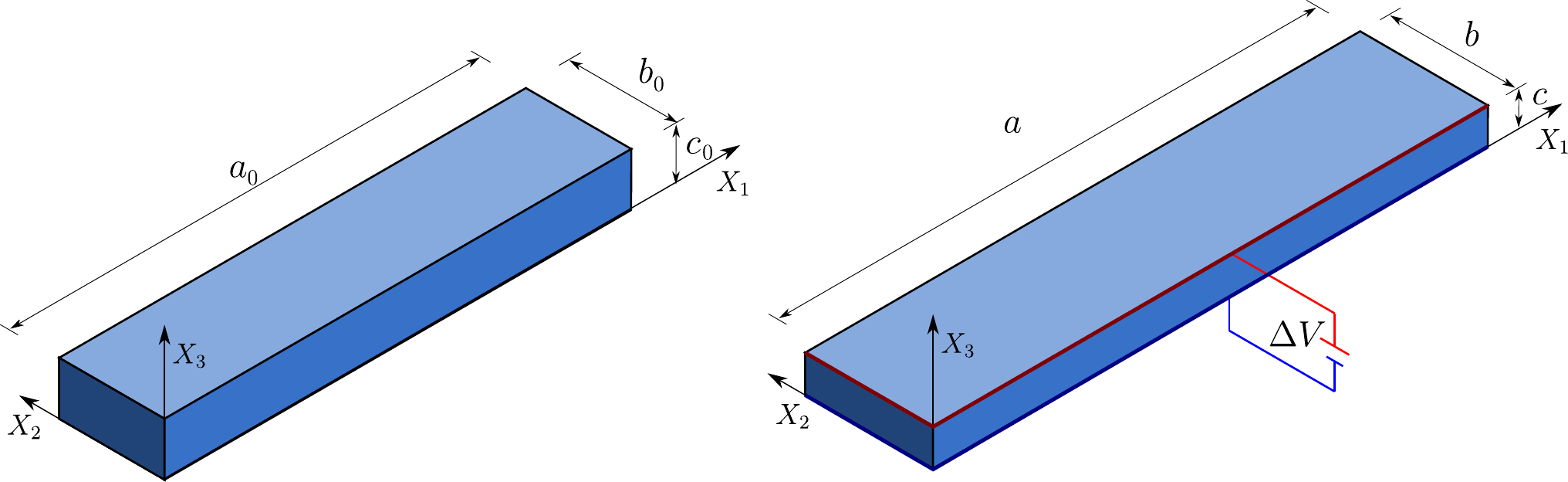}  
	\end{tabular}
	\caption{Numerical example 1. Experimental set-up. The application of a uniform electric field along the $OX_3$ direction causes a stretch of the DE laminated composite along $OX_1$ direction.}
	\label{fig: Exp1 boundary conditions}
\end{figure}

We consider a prototypical set-up very similar to that already explored by several authors in the past, both numerically \cite{Galich2017_Shear_finitely, Wu2018} and experimentally \cite{Zhao2014}. This consists of a Dielectric Elastomer film, such as the one depicted in Figure \ref{fig: Exp1 boundary conditions}, which is subjected to a homogeneous state of deformation and electric field. Two electrodes are placed at opposite faces of the film and an externally controlled Lagrangian electric field $\vect{E}_0$ is applied across and orientated along the $X_3$ axis, whilst maintaining stress-free conditions, in turn generating a state of uniform deformation and electric displacement across the film. This uniform state of deformation and electric field is exploited in order to study the response of the elastomer from a local point of view, without the need to resort to a finite element discretisation. As a result, in the absence of any further loads and electric charges, the homogeneous solution to this problem corresponds to the stationary points $\{\vect{F}^{\ast},\vect{D}_0^{\ast}\}$ of the Helmholtz's free energy functional defined as
\begin{equation}
\Pi(\vect{F}^{\ast},\vect{D}_0^{\ast},\vect{E}_0)=
{\underset{\vect{F}}{\text{inf}}}\;
{\underset{\vect{D}_0}{\text{inf}}}
\left\{
e(\vect{F},\vect{D}_0)-\vect{E}_0 \cdot \vect{D}_0
\right\},
\end{equation}
where $\vect{E}_0$ is the externally controlled electric field. Thus, the stationary conditions of the above functional arise as
\begin{equation}\label{eqn: Exp1 stationary points}
\vect{\mathcal{R}_{\vect{F}}}(\vect{F},\vect{D}_0)=\partial_{\vect{F}}e=\vect{0};\qquad 
\vect{\mathcal{R}_{\vect{D}_0}}(\vect{F},\vect{D}_0)=\partial_{\vect{D}_0}e-\vect{E}_0=\vect0.
\end{equation}
The above nonlinear stationary conditions \eqref{eqn: Exp1 stationary points} are solved in terms of unknowns $\vect{F}^{\ast}$ and $\vect{D}_0^{\ast}$ dependent upon the externally controlled electric field $\vect{E}_0$. A similar set-up has been previously used in many references \cite{Gei2018,Galich2017_Shear}, albeit restricted to the condition of plane strain and strict incompressibility. Here, these kinematic assumptions are relaxed and the deformation gradient tensor is left to adopt a more complex expression. Specifically, $\vect{F}$ and $\vect{E}_0$ are formulated as
\begin{equation}
\vect{F}=\left[\begin{matrix}
F_{11}   &     0      & 0\\
   0     &   F_{22}   & 0\\
   0     &     0      & F_{33}\\
\end{matrix}\right];\qquad 
\vect{E}_0=\begin{bmatrix}
0\\
0\\
E_0
\end{bmatrix},
\end{equation}

Typically, $F_{33}$ is constrained through the condition of strict incompressibility, as this permits to obtain closed-form solutions to the problem. Notice that these simplifying assumption does not apply to this study, where the nonlinear stationary conditions \eqref{eqn: Exp1 stationary points} are solved by an iterative Newton-Raphson method. In addition, in order to track the entire equilibrium path beyond the onset of limit points (i.e., snap-through, snap-back),  an arc-length technique is employed. Two constitutive models have been considered in this study. The first restricts to the isotropic case whereas the second considers the case of the transversely isotropic symmetry group $\mathcal{D}_{\infty}$.

\subsubsection{Isotropic case}

For the isotropic case, two energy functionals have been considered. The first, denoted as $\bar{e}_{\text iso,1}(\vect{F},\vect{D}_0)$ is given in equation \eqref{eqn:non-polyconvex isotropic energy} and it has been additively decomposed into two terms. The first represents strain energy of a Mooney-Rivlin model and that of an ideal dielectric elastomer. This entire contribution has been denoted as $\bar{e}_{\text{ID}}(\vect{F},\vect{D}_0)$. The second term in the additive decomposition contains the non-
$\mathcal{A}$-polyconvex invariant $I_8(\vect{F},\vect{D}_0)$.

\begin{equation}\label{eqn:non-polyconvex isotropic energy}
\begin{aligned}
    \bar{e}_{\text{iso},1}(\vect{F},\vect{D}_0)&=\bar{e}_{\text{ID}}(\vect{F},\vect{D}_0) + 
\frac{1}{\varepsilon_2}  I_8(\vect{F},\vect{D}_0);\\
   \bar{e}_{\text{ID}}(\vect{F},\vect{D}_0)& = \frac{\mu_1}{2}I_1^{\text{dev}}(\vect{F}) + \frac{\mu_2}{2}I_2^{\text{dev}}(\vect{F}) + \frac{\lambda}{2}(I_3(\vect{F})-1)^2 + \frac{1}{2\varepsilon_1 }\frac{I_7(\vect{F},\vect{D}_0)}{I_3(\vect{F})} 
   \end{aligned}
\end{equation}
where
\begin{equation}
I_8(\vect{F},\vect{D}_0)=\vect{H}\vect{D}_0\cdot \vect{H}\vect{D}_0
\end{equation}

Furthermore, a second energy functional, denoted as $\bar{e}_{\text iso,1}(\vect{F},\vect{D}_0)$ has been considered. This is also additively decomposed into the $\bar{e}_{\text{ID}}(\vect{F},\vect{D}_0)$ and the additional invariant $I_{8,\text{pol}}$, which results from the polyconvexication of $I_8$. The explicit expression for  $\bar{e}_{\text {iso},1_{\text{pol}}}(\vect{F},\vect{D}_0)$ can be seen in equation \eqref{eqn:polyconvex isotropic energy}, i.e.
\begin{equation}\label{eqn:polyconvex isotropic energy}
\begin{aligned}
    \bar{e}_{\text{iso},1_{\text{pol}}}(\vect{F},\vect{D}_0)&=\bar{e}_{\text{ID}}(\vect{F},\vect{D}_0) +  I_{8,\text{pol}}(\vect{F},\vect{D}_0)  - 12\alpha^2\log(I_3(\vect{F}))
    \end{aligned}
\end{equation}
where 
\begin{equation}
    I_{8,\text{pol}}=\alpha^2I_2(\vect{F}) + \alpha\beta I_8 + \beta^2 I^2_6(\vect{D}_0)
\end{equation}

\begin{table}[htbp!] \label{Tab:MaterialParametersNonpolyconvex_Isotropic}
\centering
\caption{Material properties for constitutive model in equation \eqref{eqn:non-polyconvex isotropic energy}}
     \label{tab:non-polyconvex isotropic}
     \begin{tabular}{cccccc}
     \toprule
     {$\mu_1$} & {$\mu_2$} & {$\lambda$} & {$\varepsilon_1$}& {$\varepsilon_2$}
     \\
     \midrule
     $1\times10^5$ & $1.0\mu_1$ & $10^3\mu_1$ & $4.82\varepsilon_0$& $24\varepsilon_0$ \\
     \bottomrule
\end{tabular}
 \end{table}

\begin{table}[htbp!] \label{Tab:MaterialParametersPolyconvex_Isotropic}
\centering
\caption{Material properties for the constitutive model in equation \eqref{eqn:polyconvex isotropic energy}}
     \label{tab:polyconvex isotropic}
     \begin{tabular}{ccccccccccccc}
     \toprule
   Parameters  & Pol. Mat. 1  & Pol. Mat. 2 & Pol. Mat. 3 & Pol. Mat. 4 & Pol. Mat. 5
     \\
     \midrule
{$\mu_1$} &$10^5$ &$10^5$ &$10^5$ &$10^5$          &$10^5$\\
{$\mu_2$} &$10^5$ &$10^5$ &$10^5$ &$10^5$          &$10^5$\\
{$\lambda$} &$10^8$              &$10^8$              &$10^8$            &$10^8$               &$10^8$\\
{$\varepsilon$} &$4.82\varepsilon_0$ &$4.82\varepsilon_0$ &$4.82\varepsilon_0$ &$4.82\varepsilon_0$   &$4.82\varepsilon_0$\\
{$\alpha$}      &$1.78\times 10^8$   &$2.90\times 10^9$   &$4.91\times 10^9$ &$6.92\times 10^9$    &$8.93\times 10^9$\\ 
{$\beta$}      &$2.52\times 10^2$            &$4.11\times 10^3$             &$6.95\times 10^3$           &$9.80\times 10^3$              &$1.26\times 10^4$\\ 
     \bottomrule
\end{tabular}
 \end{table}

For the non-$\mathcal{A}$-polyconvex model in \eqref{eqn:non-polyconvex isotropic energy}, the values of the material parameters can be found in Table \ref{Tab:MaterialParametersNonpolyconvex_Isotropic}. For the $\mathcal{A}$-polyconvex model in \ref{eqn:polyconvex isotropic energy}, five combinations of material parameters can be found in Table \ref{Tab:MaterialParametersPolyconvex_Isotropic}. It can be seen that the material parameters $\{\mu_1\mu_2,\lambda,\varepsilon\}$ for the polyconvex models have been kept the same as their counterparts $\{\mu_1\mu_2,\lambda,\varepsilon_1\}$, respectively, for the non-polyconvex case. Several values for the two remaining material parameters $\{\alpha,\beta\}$ have been considered in order to see their influence in the equilibrium path of the elastomer. From Figures \ref{fig:isotropic model}$_{a,b}$, the combination of values for $\{\alpha,\beta\}$ that yields a closer response to the non-polyconvex model is that corresponding with the polyconvex material 2 (Pol. Mat. 2), whose values for $\{\alpha,\beta\}$ can be found in the third column of Table \ref{Tab:MaterialParametersPolyconvex_Isotropic}. These values have been determined by performing an optimisation problem. Specifically, for each pair of values $\left(\vect{F},\vect{D}_0\right)$ in the equilibrium path of the non-polyconvex model, we have formulated the following minimisation problem
\begin{equation}\label{eqn:optimisation isotropic}
\min_{\alpha,\beta} \,\,  \left\{ \begin{aligned}
&\mathcal{J}\\
&\text{s.t.} \,\,\,\alpha>0,\,\, \beta>0
\end{aligned}\right.
\end{equation}
where the objective function $\mathcal{J}$ is defined as
\begin{equation}
   \mathcal{J}= \sqrt{\sum_{i=1}^n \vert\vert\partial_{\vect{F}}\bar{e}_{\text{iso},{1_{\text{pol}}}}(\vect{F}_i,\vect{D}_{0i})\vert\vert^2}+ \sqrt{\sum_{i=1}^n\frac{\vert\vert\partial_{\vect{D}_0}\bar{e}_{\text{iso},{1_{\text{pol}}}}(\vect{F}_i,\vect{D}_{0i})-\partial_{\vect{D}_0}\bar{e}_{\text{iso},{1}}(\vect{F}_i,\vect{D}_{0i})\vert\vert^2 }{\vert\vert\partial_{\vect{D}_0}\bar{e}_{\text{iso},{1}}(\vect{F}_i,\vect{D}_{0i})\vert\vert^2 }},   
\end{equation}
where $n$ refers to the number of discrete data pairs $\left(\vect{F},\vect{D}_0\right)$ describing the discrete equilibrium path of the non-polyconvex model.

The similarity between the equilibrium paths of the non-polyconvex model and that whose material parameters $\{\alpha,\beta\}$ have been obtained through the optimisation method described is fairly reasonable up to a $10$ percent in the $F_{11}$ component of the deformation gradient tensor, beyond which both equilibrium paths start diverging (see Figure \ref{fig:isotropic model}). Clearly, this owes to the fact that the polyconvex model contains higher nonlinear terms $I_4^2$ and $I_6^2$, yielding a slightly more stable response to that of its non-polyconvex counterpart. From Figure \ref{fig:isotropic model}, it is possible to observe the regions where the Hessian operator loses positive definiteness and hence, where the loss of convexity occurs. Interestingly, the non-polyconvex model loses in addition ellipticity (see the yellow region in Figure \ref{fig:isotropic model}$_c$). This has been checked by monitoring the least of the minors of the acoustic tensor $\vect{Q}$ in equation \eqref{eqn: Exp1 acoustic tensor}. Evidently, this is not appreciated in any of the five $\mathcal{A}$-polyconvex models in Figures \ref{fig:isotropic model}$_d$-\ref{fig:isotropic model}$_h$.

\begin{figure}[htbp!]
    \centering
    \begin{tabular}{cc}
\includegraphics[width=0.36\textwidth]{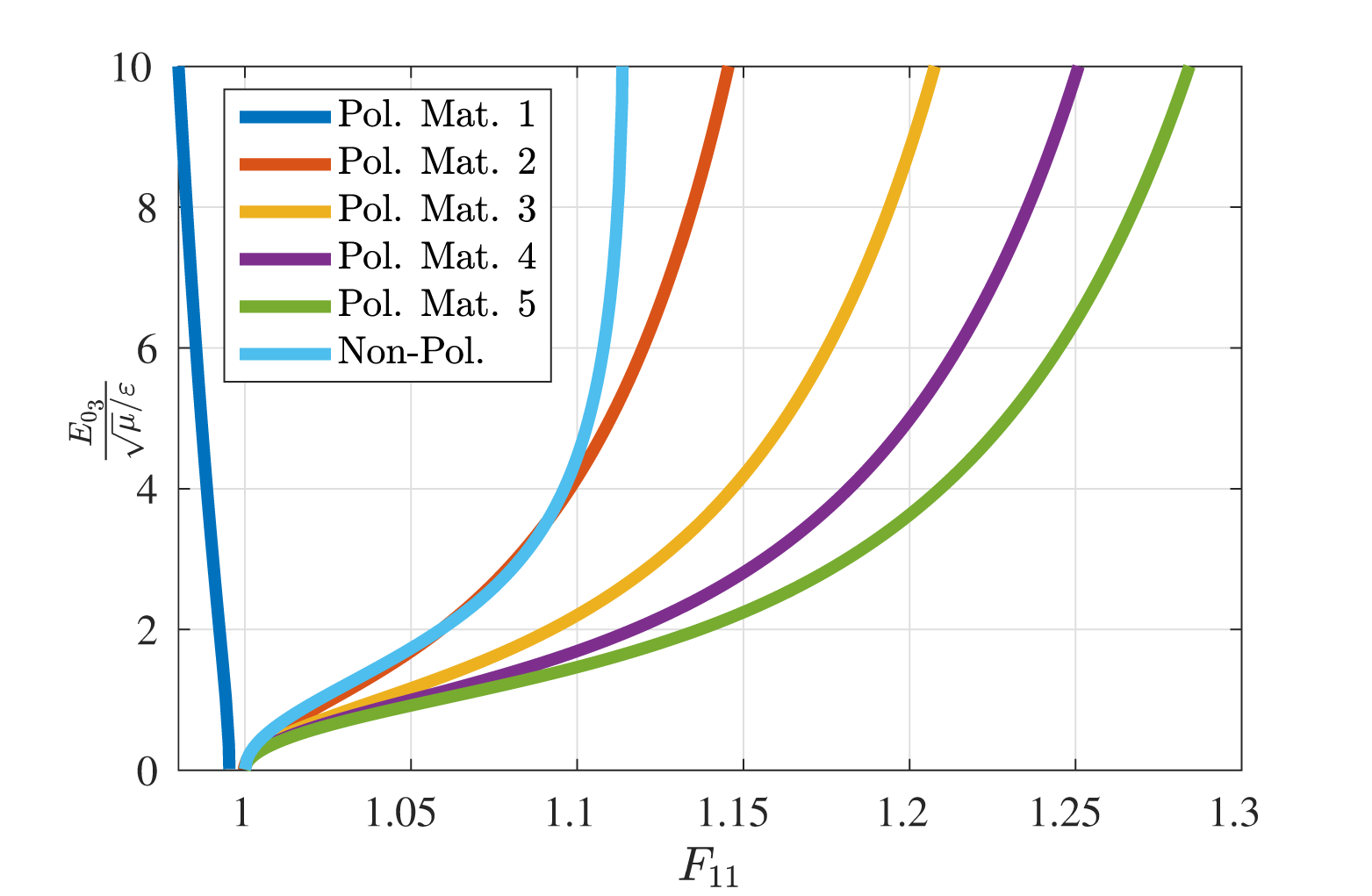}
         &  
\includegraphics[width=0.36\textwidth]{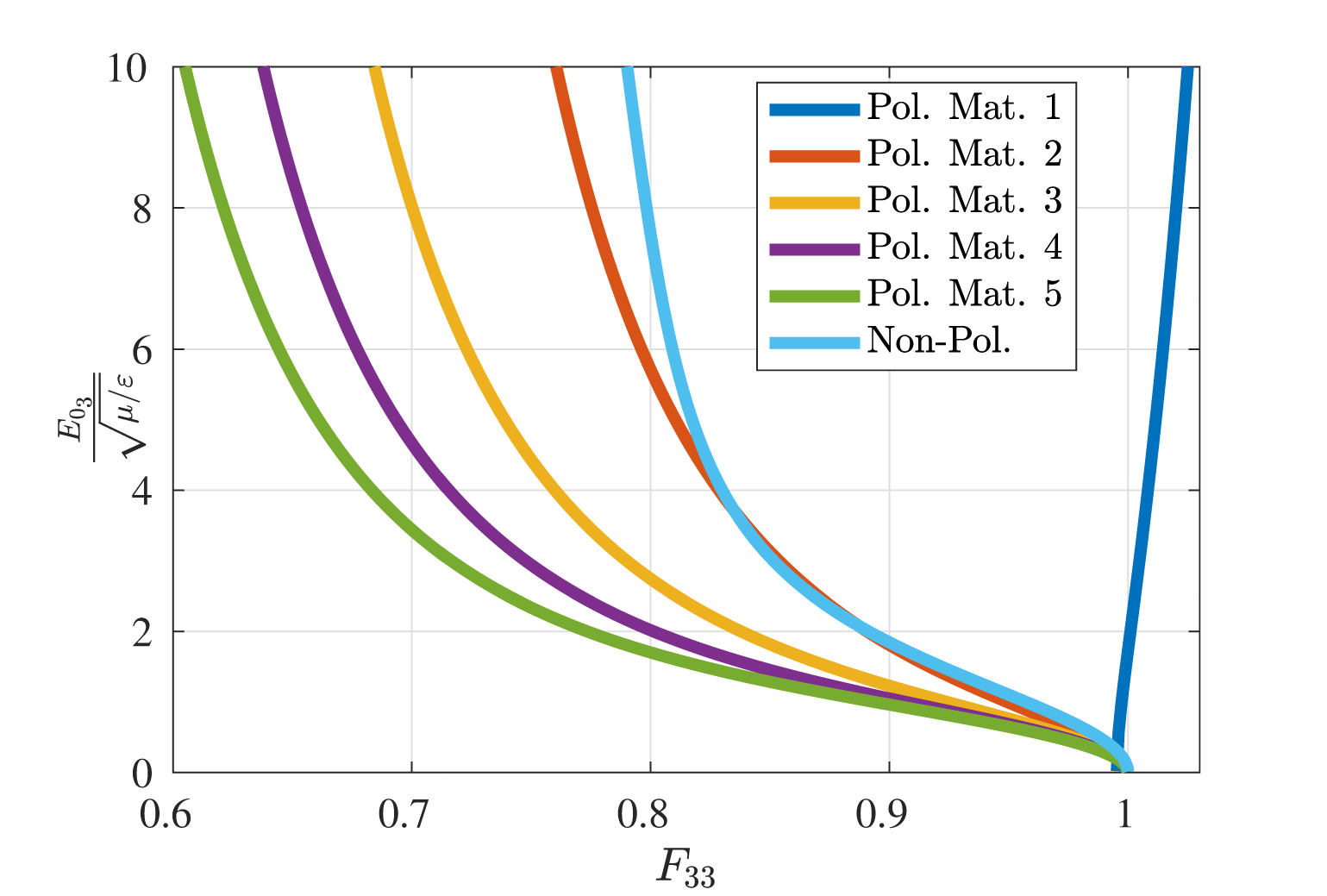}
         \\
(a)         &   (b)\\
\includegraphics[width=0.36\textwidth]{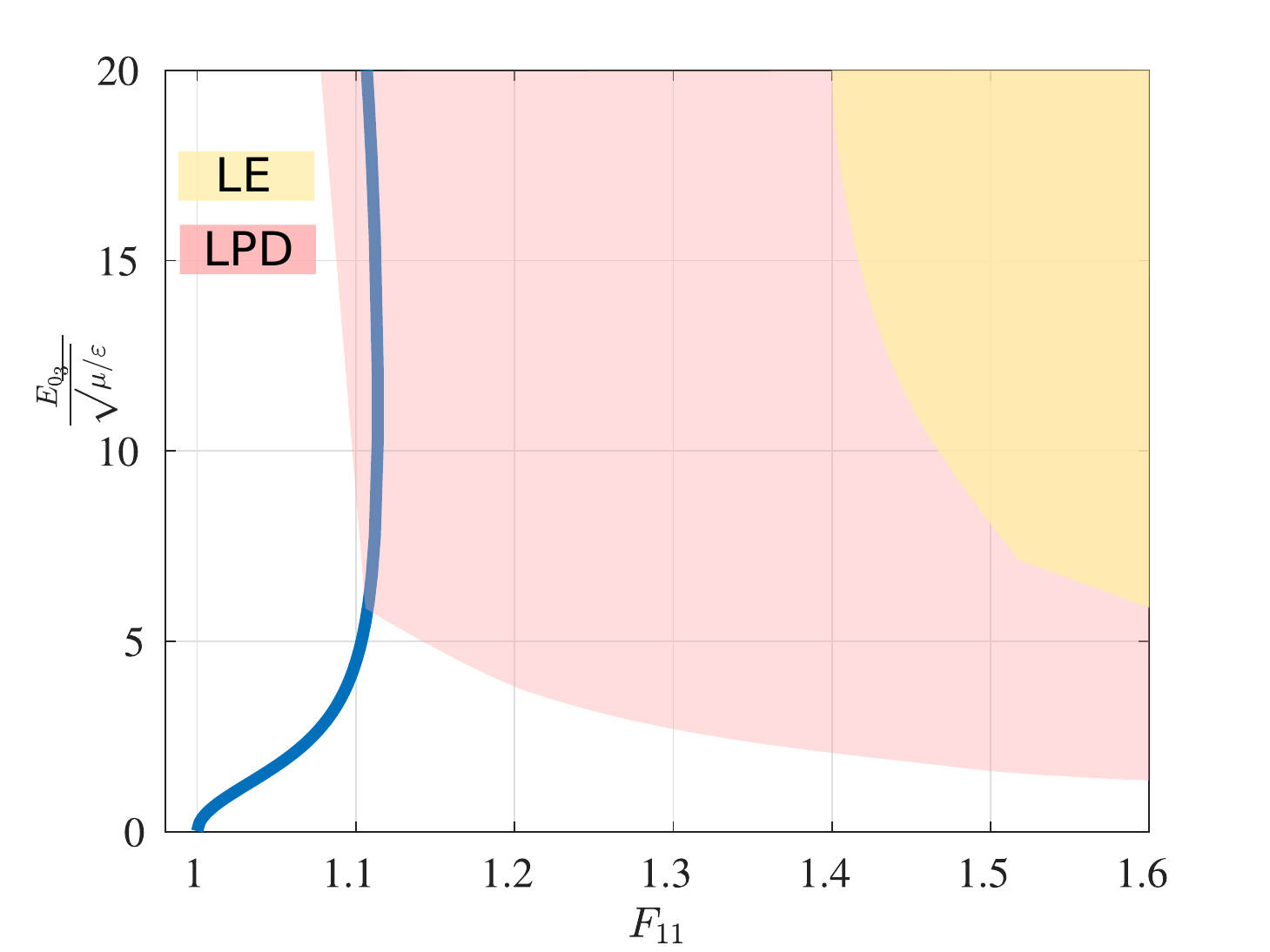}
         &  
\includegraphics[width=0.36\textwidth]{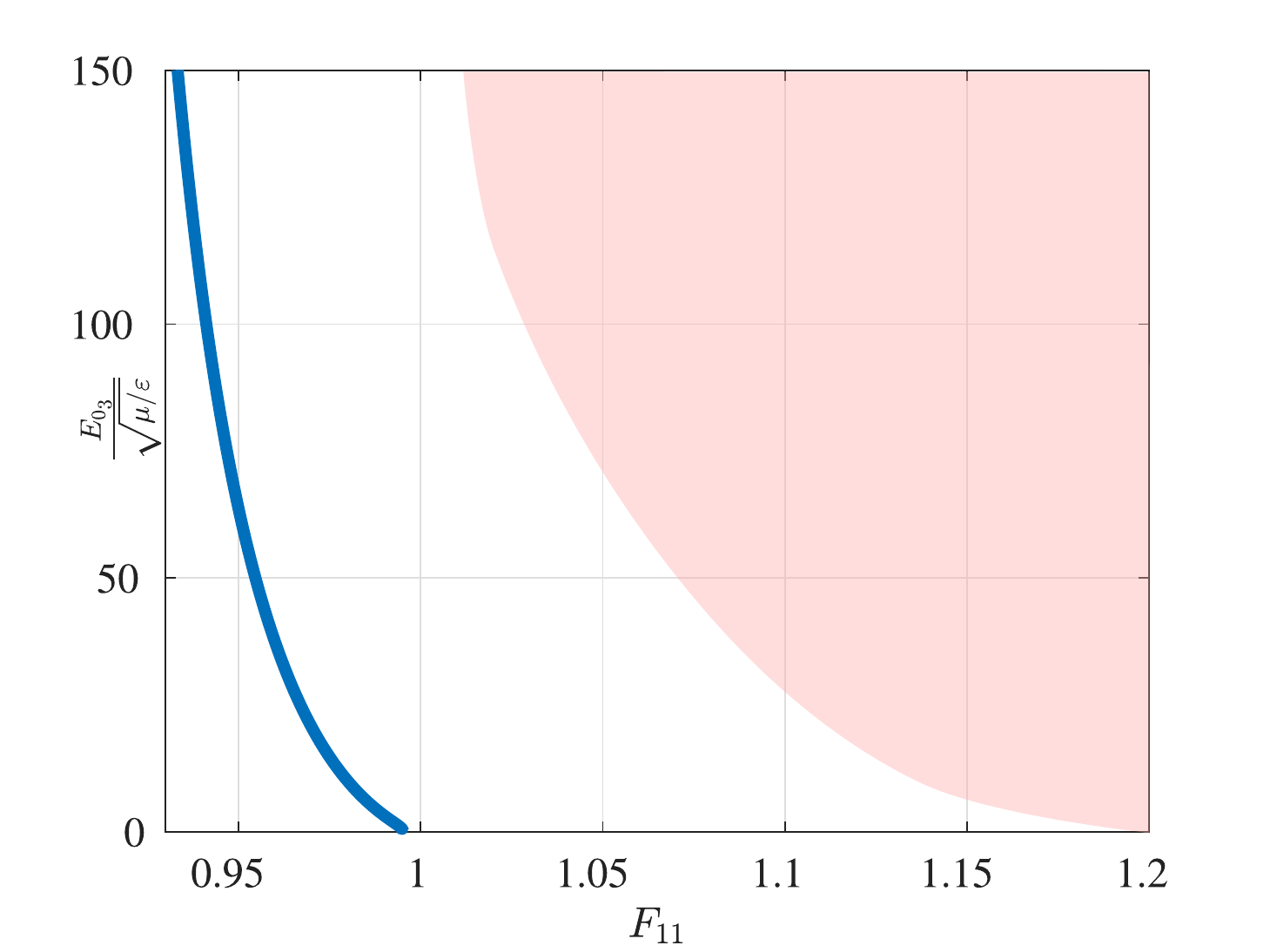}\\
(c)         &   (d)\\
\includegraphics[width=0.36\textwidth]{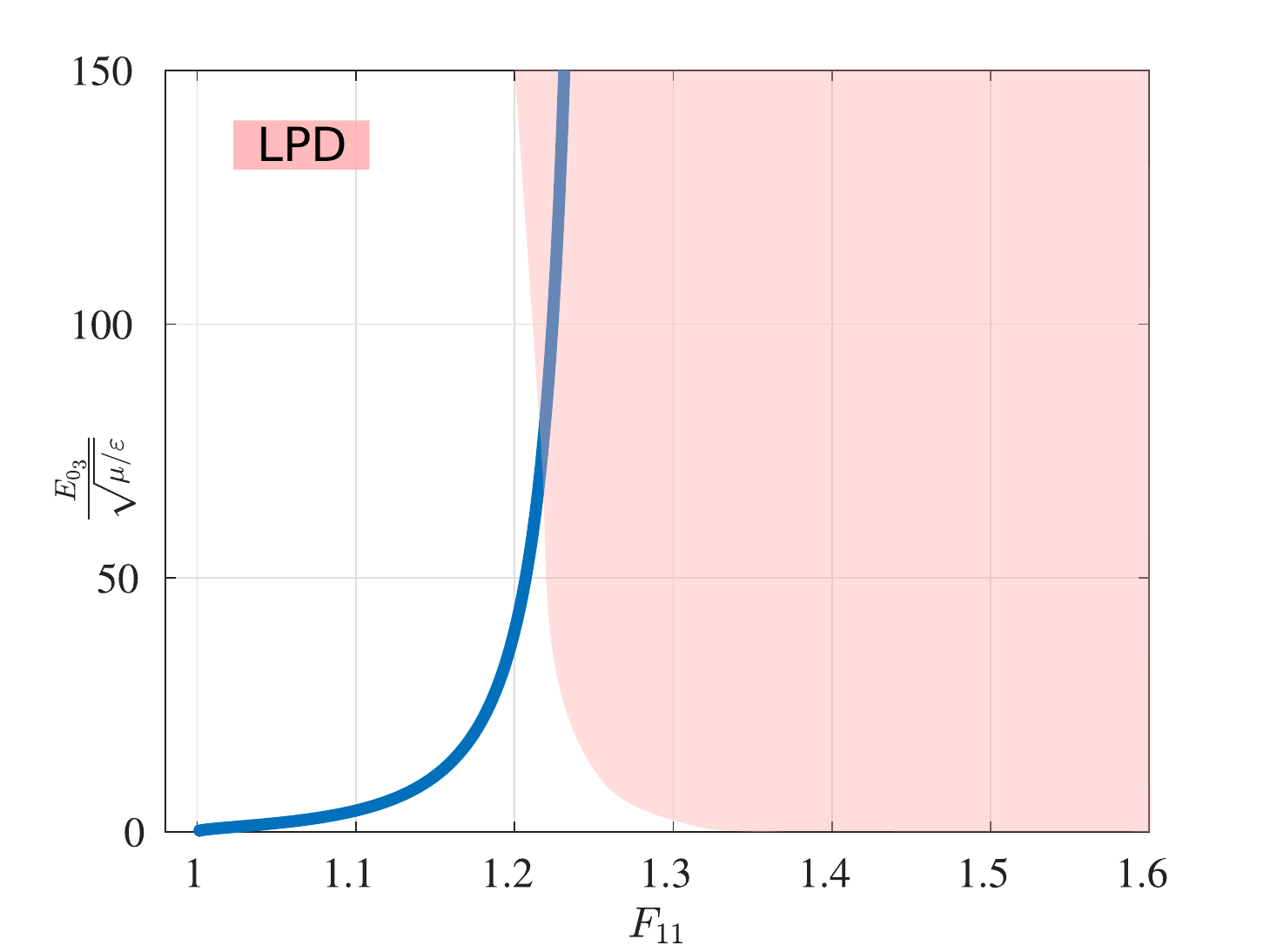}
         &  
\includegraphics[width=0.36\textwidth]{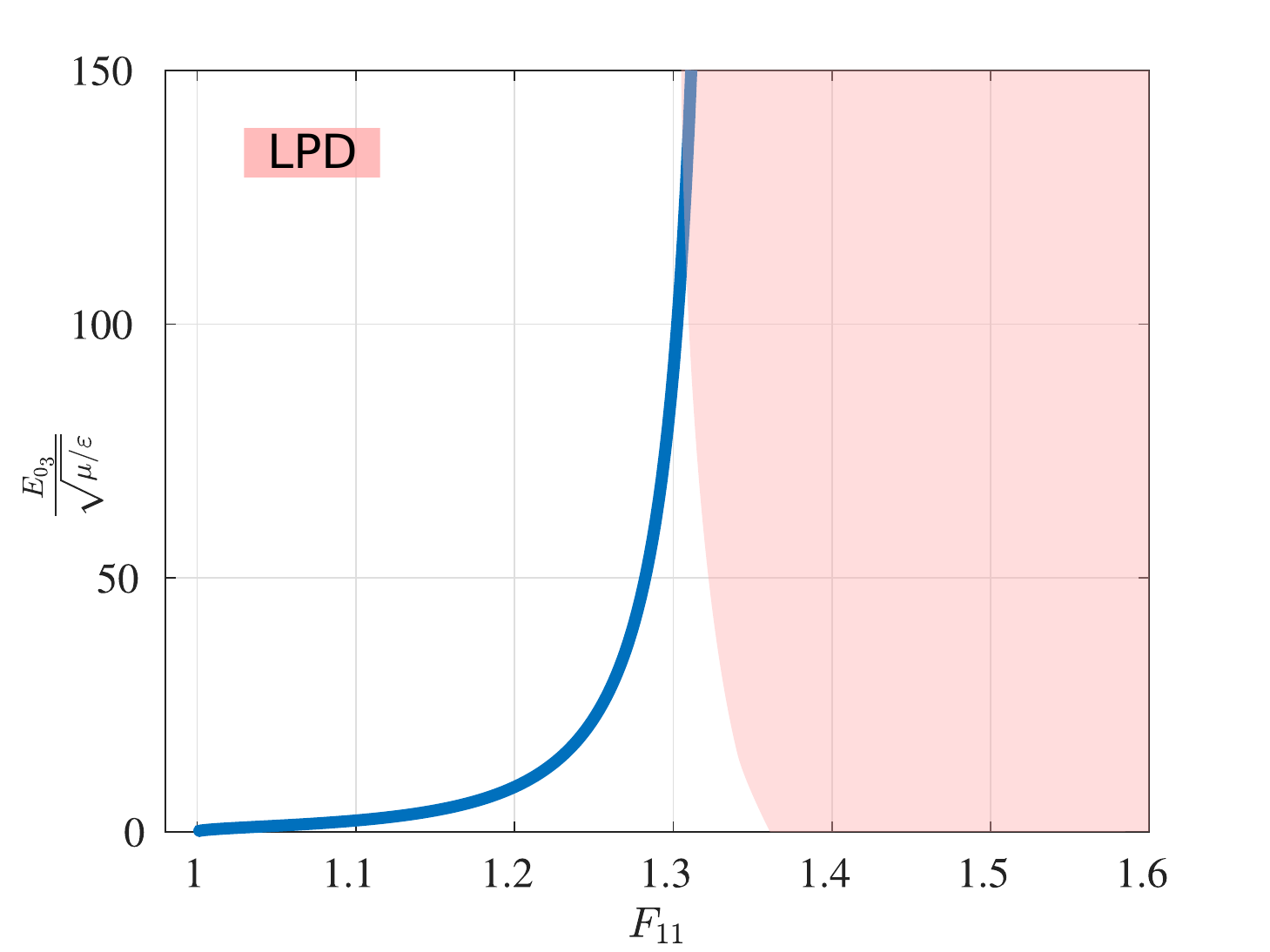}\\
(e)         &   (f)\\
\includegraphics[width=0.36\textwidth]{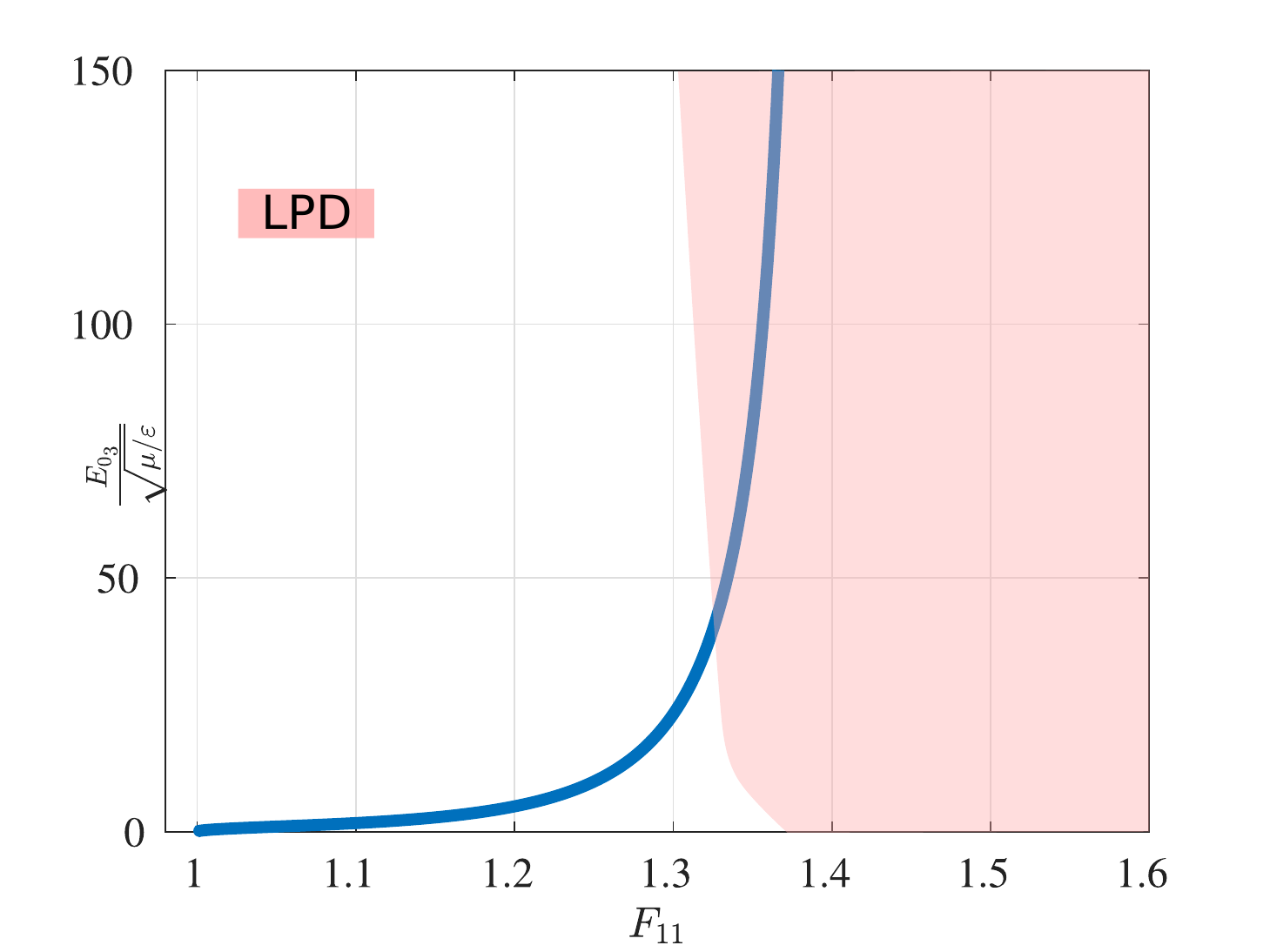}
         &  
\includegraphics[width=0.36\textwidth]{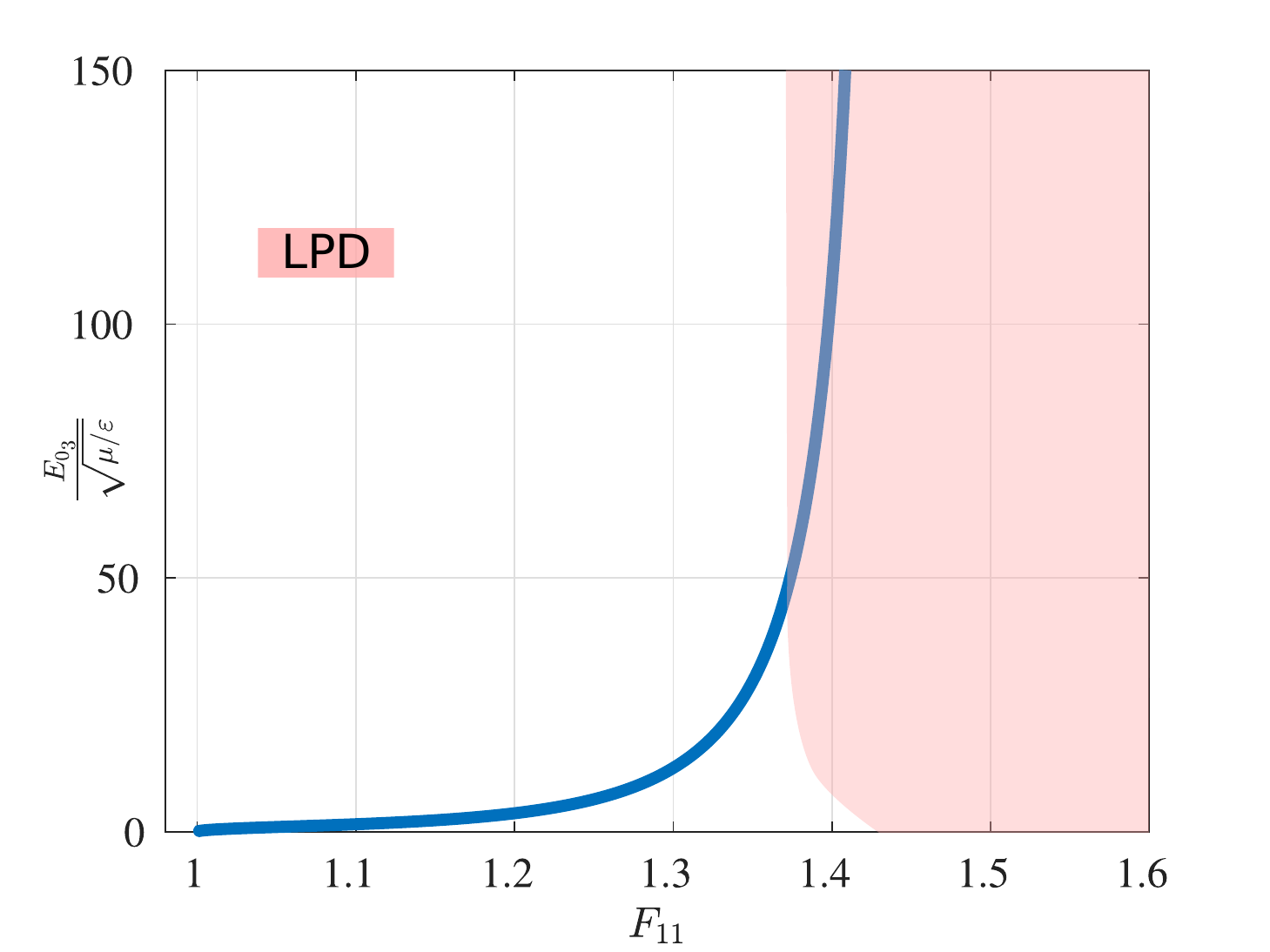}\\
(g)         &   (h)
    \end{tabular}
\caption{(a)-(b) Equilibrium paths for non-polyconvex model in equation \eqref{eqn:non-polyconvex isotropic energy} with material properties in Table \ref{Tab:MaterialParametersNonpolyconvex_Isotropic} and for the five polyconvex models (see equation \eqref{eqn:polyconvex isotropic energy}) with material parameters in Table \ref{Tab:MaterialParametersPolyconvex_Isotropic}. (c)-(h) represent the regions where loss of positive definiteness (LPD) of the Hessian operator and loss of ellipticity (LE) occurs for the non-polyconvex and polyconvex models. The reference values for the dimensionless electric field are $\mu=10^5$ and $\varepsilon=4.8\varepsilon_0$.}
    \label{fig:isotropic model}
\end{figure}

\subsubsection{Anisotropic case}

In this example, we consider the transversely anisotropic group $\mathcal{D}_{\infty h}$. Two energy density functions are considered. The first, denoted as $\bar{e}_{\mathcal{D}_{\infty h,1}}(\vect{F},\vect{D}_0,\vect{M})$, is not $\mathcal{A}$-polyconvex. This is additively decomposed into the isotropic density function $\bar{e}_{\text{ID}}(\vect{F},\vect{D}_0)$ ($\mathcal{A}$-polyconvex) and the non-polyconvex invariant $K_2^{\mathcal{D}_{\infty h,1}}(\vect{F},\vect{D}_0,\vect{M})$. This can be seen in equation \eqref{eqn:anisotropic non polyconvex}.

\begin{equation}\label{eqn:anisotropic non polyconvex}
\begin{aligned}
    \bar{e}_{\mathcal{D}_{\infty h,1}}(\vect{F},\vect{D}_0,\vect{M})&=\bar{e}_{\text{ID}}(\vect{F},\vect{D}_0) + \frac{1}{\varepsilon_2}K_2^{\mathcal{D}_{\infty h}}(\vect{F},\vect{D}_0,\vect{M}),
   \end{aligned}
\end{equation}
where
\begin{equation}
K_2^{\mathcal{D}_{\infty h}}(\vect{F},\vect{D}_0)=\left(\vect{d}\cdot \vect{FM}\right)^2
\end{equation}

The material parameters $\{\mu_1,\mu_2,\lambda,\varepsilon_1,\varepsilon_2\}$ featuring in the resulting non-polyconvex model $\bar{e}_{\mathcal{D}_{\infty h,1}}(\vect{F},\vect{D}_0,\vect{M})$ can be found in Table \ref{Tab:MaterialParametersNonpolyconvex_Anisotropic}.

\begin{table}[htb] \label{Tab:MaterialParametersNonpolyconvex_Anisotropic}
\centering
\caption{Material properties for non-polyconvex model in equation \eqref{eqn:anisotropic non polyconvex}}
     \label{tab:Example2}
     \begin{tabular}{cccccc}
     \toprule
     {$\mu_1$} & {$\mu_2$} & {$\lambda$} & {$\varepsilon_1$}& {$\varepsilon_2$}
     \\
     \midrule
     $1\times10^5$ & $1.0\mu_1$ & $10^3\mu_1$ & $4.82\varepsilon_0$& $24\varepsilon_0$ \\
     \bottomrule
\end{tabular}
 \end{table}

The second strain energy density function considered, denoted as $\bar{e}^{\text{pol}}_{\mathcal{D}_{\infty h,1}}(\vect{F},\vect{D}_0,\vect{M})$ can be seen in equation \eqref{eqn:anisotropic polyconvex model}. Notice that, in contrast to equation \eqref{eqn:anisotropic non polyconvex}, the non-polyconvex invariant $K_2^{\mathcal{D}_{\infty h,1}}(\vect{F},\vect{D}_0,\vect{M})$ has been replaced with its polyconvex counterpart $K_{2,\text{pol}}^{\mathcal{D}_{\infty h}}(\vect{F},\vect{D}_0,\vect{M})$ and an additional term (third term on the right hand side of \eqref{eqn:anisotropic polyconvex model}) which guarantees the stress-free condition in the origin, namely
\begin{equation}\label{eqn:anisotropic polyconvex model}
\begin{aligned}
    \bar{e}^{\text{pol}}_{\mathcal{D}_{\infty h,1}}(\vect{F},\vect{D}_0,\vect{M})&=\bar{e}_{\text{ID}}(\vect{F},\vect{H},J,\vect{d}) +  K_{2,\text{pol}}^{\mathcal{D}_{\infty h}}(\vect{F},\vect{D}_0,\vect{M})  + \beta^2\Big(I_5(\vect{F},\vect{M})-\log(I_3(\vect{F}))\Big);\\
K_{2,\text{pol}}^{\mathcal{D}_{\infty h}}(\alpha_1,\alpha_2,\alpha_3,\alpha_4,\vect{F},\vect{D}_0,\vect{M})&=\alpha^2 I^2_7(\vect{F},\vect{D}_0) + \alpha\beta K_2^{\mathcal{D}_{\infty h}}(\vect{F},\vect{D}_0) + \beta^2I^2_4(\vect{F},\vect{M})   \\
    \end{aligned}
\end{equation}

For all the models, the preferred direction $\vect{M}$ was set to $\vect{M}=\begin{bmatrix}
0  &  0&  1\end{bmatrix}^T$.

\begin{table}[htb] \label{Tab:MaterialParametersPolyconvex_Anisotropic}
\centering
\caption{Material parameters for polyconvex model in equation \eqref{eqn:anisotropic polyconvex model}}
     \label{tab:Example2}
     \begin{tabular}{ccccccccccccc}
     \toprule
   Parameters  & Pol. Mat. 1  & Pol. Mat. 2 & Pol. Mat. 3 & Pol. Mat. 4 & Pol. Mat. 5
     \\
     \midrule
{$\mu_1$}       &$3.67\times 10^5$  &$6.30\times 10^5$    &$5.12\times 10^5$ &$0.44$               &$8.33\times 10^5$\\
{$\mu_2$}       &$21.86$            &$9.56\times 10^3$    &$6.25\times 10^3$ &$2.87\times 10^{-2}$ &$0.25$\\
{$\lambda$}     &$10^8$              &$10^8$              &$10^8$            &$10^8$               &$10^8$\\
{$\varepsilon_1$} &$7.48\varepsilon_0$ &$9.6\varepsilon_0$  &$12.36$           &$4.9\varepsilon_0$   &$9.6\varepsilon_0$\\
{$\alpha$}      &$3.97\times 10^6$   &$4.20\times 10^6$   &$6.50\times 10^6$ &$1.27\times 10^7$    &$2.47\times 10^6$\\ 
{$\beta$}      &$110.13$            &$0.232$             &$37.40$           &$193.12$              &$1.42\times 10^{-2}$\\ 
     \bottomrule
\end{tabular}
 \end{table}

For the $\mathcal{A}$-polyconvex model in \ref{eqn:anisotropic polyconvex model}, five combinations of material parameters can be found in Table \ref{Tab:MaterialParametersPolyconvex_Anisotropic}. Unlike in the isotropic case, (where $\{\mu_1,\mu_2,\lambda,\varepsilon_1\}$) kept the same values as the non-polyconvex model, in this case, all the six material parameters $\{\mu_1,\mu_2,\lambda,\varepsilon_1,\alpha,\beta\}$ have been varied in order to observed their influence in the equilibrium path of the resulting polyconvex model. This can be observed in Figure \ref{fig:anisotropic model}$_{a,b}$. From this figure, the combination of values for $\{\mu_1,\mu_2,\lambda,\varepsilon_1,\alpha,\beta\}$ depends on the interval of $\vect{F}$ where the optimisation problem described in that yields a closer response to the non-polyconvex model is that corresponding with the polyconvex material 2 (Pol. Mat. 2), whose values for $\{\alpha,\beta\}$ can be found in the third column of Table \ref{Tab:MaterialParametersPolyconvex_Isotropic}. These values have been determined by performing an optimisation problem. Specifically, for each pair of values $\left(\vect{F},\vect{D}_0\right)$ in the equilibrium path of the non-polyconvex model, we have formulated the following minimisation problem \eqref{eqn:optimisation isotropic} was carried out. This entails that the minimisation problem in this case can be reformulated as
\begin{equation}
\min_{\mu_1,\mu_2,\lambda,\varepsilon_1,\alpha,\beta} \,\,  \left\{ \begin{aligned}
&\mathcal{J}\\
&\text{s.t.} \,\,\,\{\mu_1,\mu_2,\lambda,\epsilon_1,\alpha,\beta\}>0
\end{aligned}\right.
\end{equation}
where the objective function $\mathcal{J}$ is defined as
\begin{equation}
   \mathcal{J}= \sqrt{\sum_{i=n_1}^{n_2} \vert\vert\partial_{\vect{F}}\bar{e}^{\text{pol}}_{\mathcal{D}_{\infty h,1}}(\vect{F}_i,\vect{D}_{0i})\vert\vert^2}+ \sqrt{\sum_{i=n_1}^{n_2}\frac{\vert\vert\partial_{\vect{D}_0}\bar{e}^{\text{pol}}_{\mathcal{D}_{\infty h,1}}(\vect{F}_i,\vect{D}_{0i})-\partial_{\vect{D}_0}\bar{e}_{\mathcal{D}_{\infty h,1}}(\vect{F}_i,\vect{D}_{0i})\vert\vert^2 }{\vert\vert\partial_{\vect{D}_0}\bar{e}^{\text{pol}}_{\mathcal{D}_{\infty h,1}}(\vect{F}_i,\vect{D}_{0i})\vert\vert^2 }},   
\end{equation}
where $n_1$ and $n_2$ refer to the initial and final elements within the set $\vect{F}_i,\vect{D}_{0i}$, describing the discrete equilibrium path of the non-polyconvex model, which are considered for the optimisation problem. It can be seen that case corresponding with the polyconvex material 4 (Pol. Mat. 4 in Figure \ref{fig:anisotropic model}$_{a,b}$, with material parameters in the fifth column of Table \ref{Tab:MaterialParametersPolyconvex_Anisotropic}) fifts the non-polyconvex model extremely well for values of $F_{11}$ close to $1$. However, this agreeement disappears beyond this point and both models are extremely disimilar past this point. On the other hand, the case corresponding with the polyconvex material 1 (Pol. Mat. 1 in Figure \ref{fig:anisotropic model}$_{a,b}$, with material parameters in the second column of Table \ref{Tab:MaterialParametersPolyconvex_Anisotropic}), only fits very well the equilibrium path of the non-polyconvex model for high values of $F_{11}$.

From Figure \ref{fig:anisotropic model}, it is possible to observe the regions where the Hessian operator loses positive definiteness and hence, where the loss of convexity occurs. Interestingly, the non-polyconvex model loses in addition ellipticity (see the yellow region in Figure \ref{fig:anisotropic model}$_c$). This has been checked by monitoring the least of the minors of the acoustic tensor $\vect{Q}$ in equation \eqref{eqn: Exp1 acoustic tensor}. Evidently, this is not appreciated in any of the five $\mathcal{A}$-polyconvex models in Figures \ref{fig:anisotropic model}$_d$-\ref{fig:anisotropic model}$_h$.

\begin{figure}[htbp!]
    \centering
    \begin{tabular}{cc}
\includegraphics[width=0.36\textwidth]{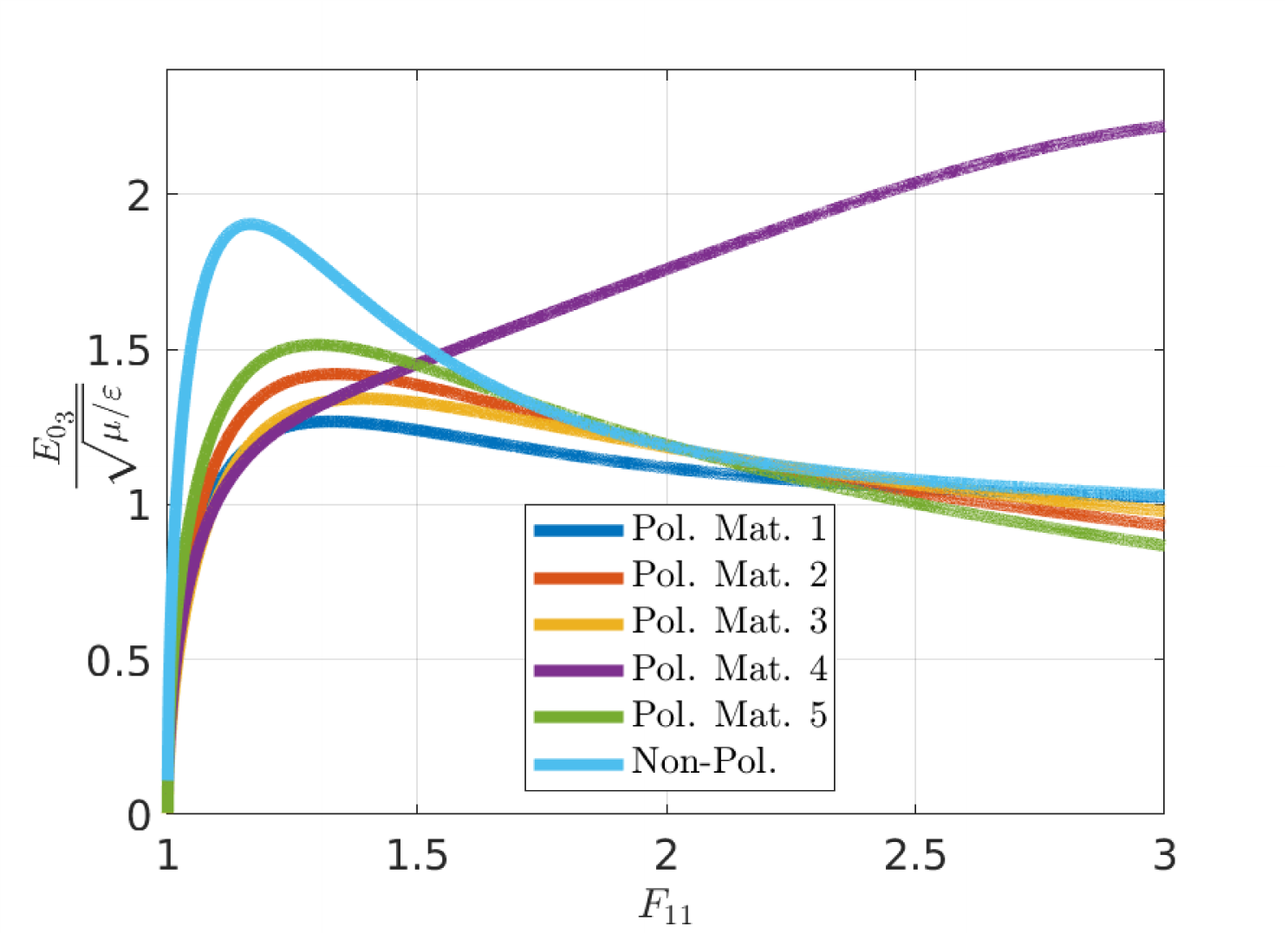}
         &  
\includegraphics[width=0.36\textwidth]{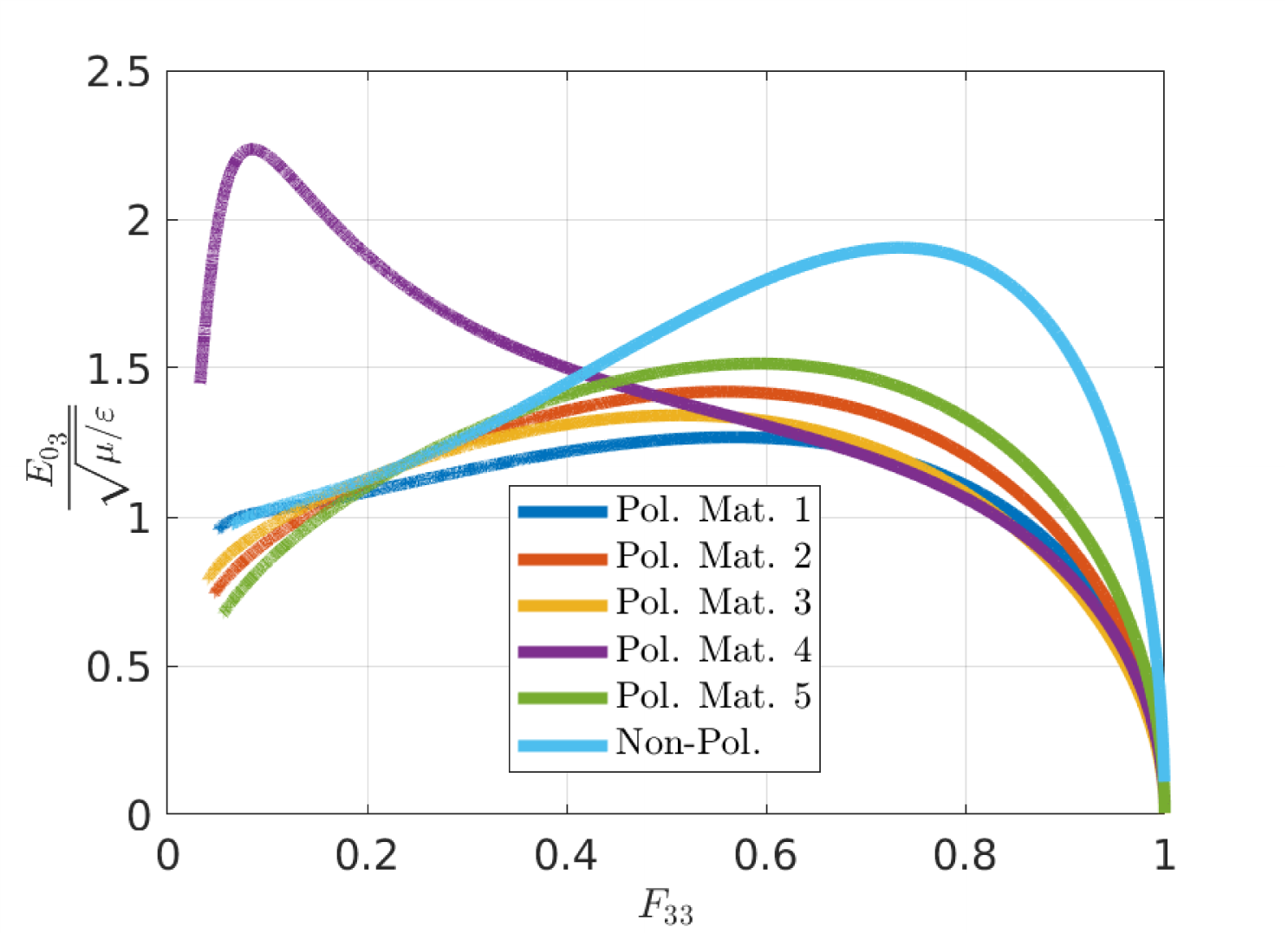}
         \\
(a)         &   (b)\\
\includegraphics[width=0.36\textwidth]{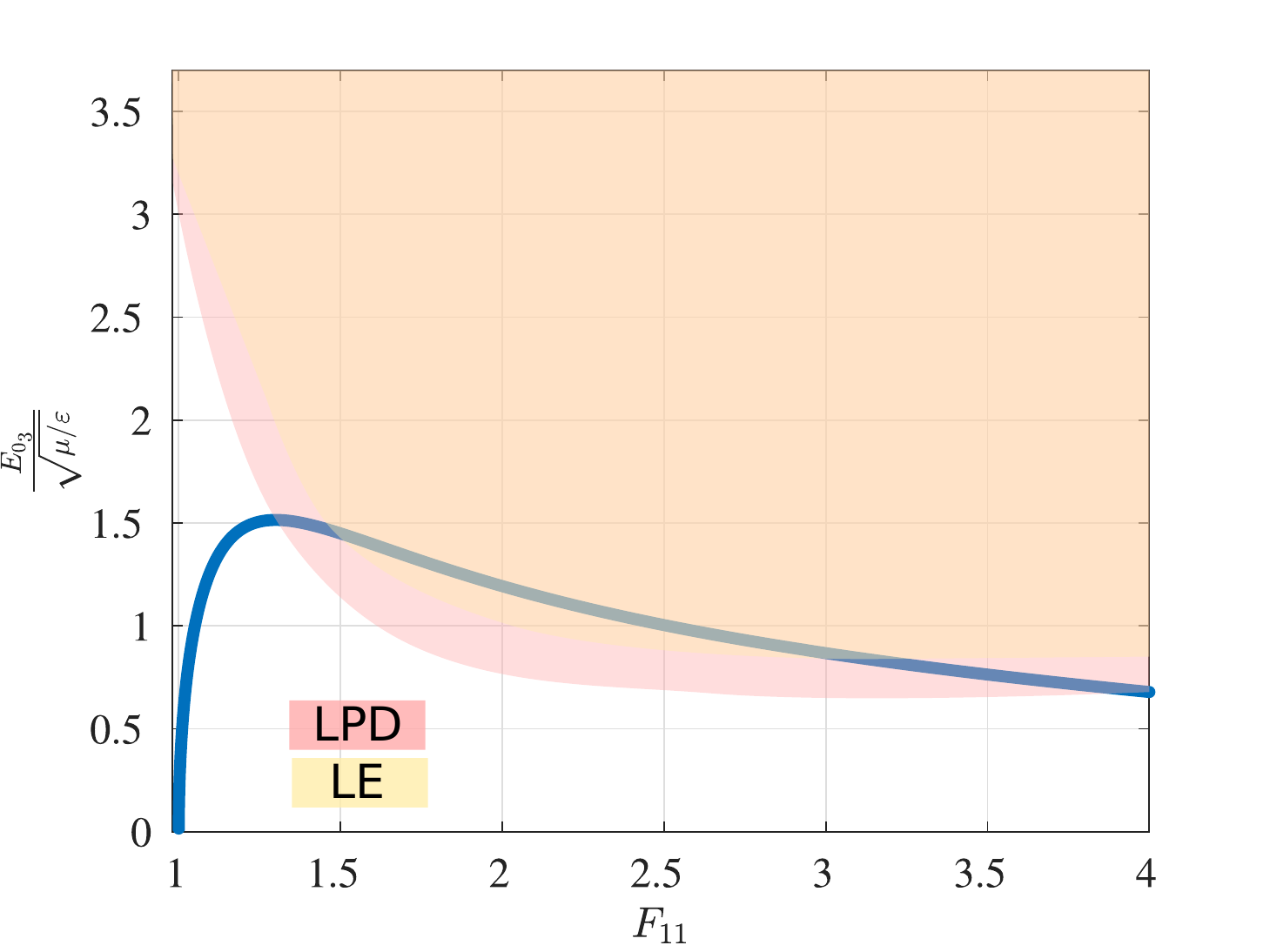}
         &  
\includegraphics[width=0.36\textwidth]{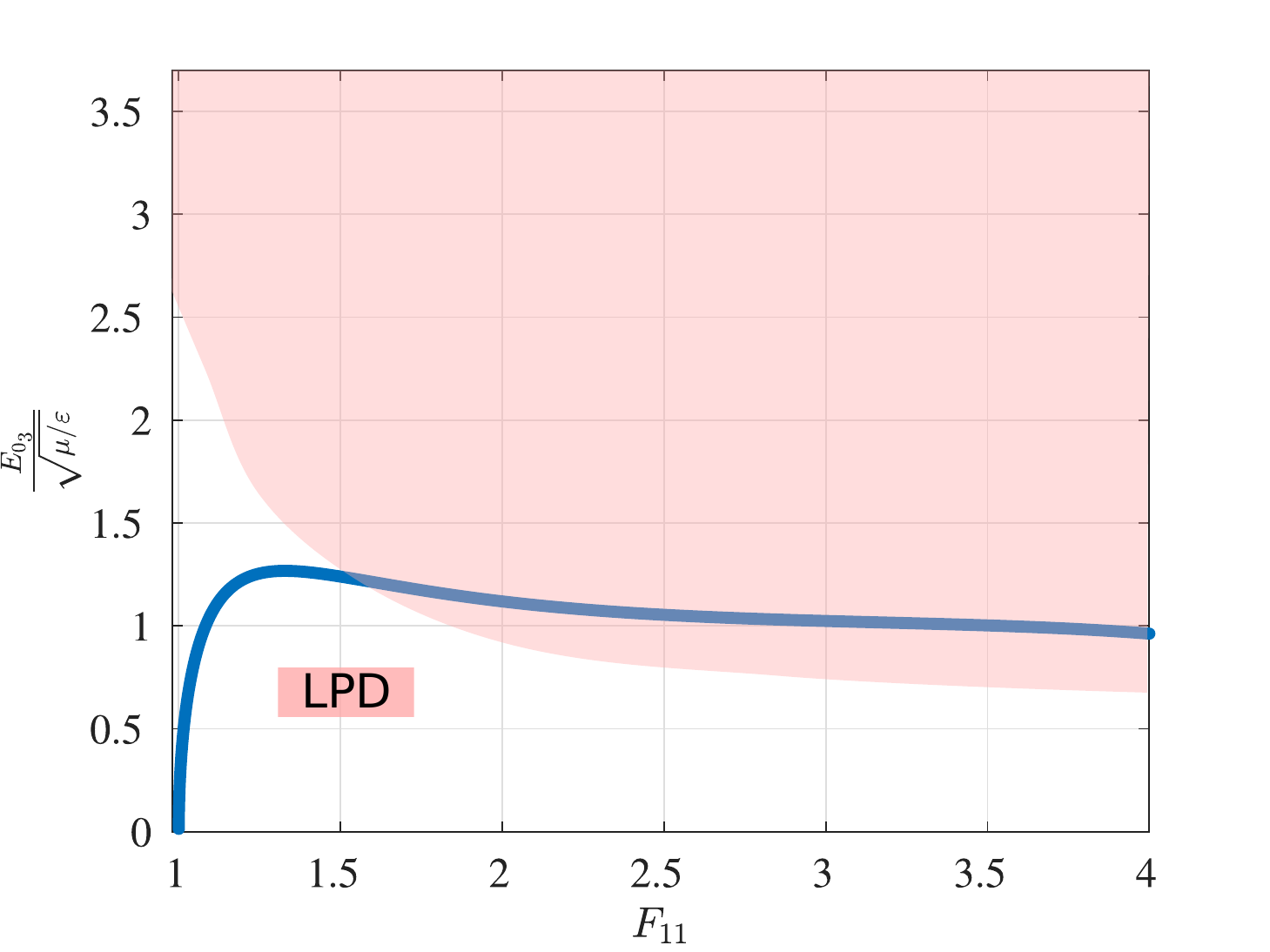}\\
(c)         &   (d)\\
\includegraphics[width=0.36\textwidth]{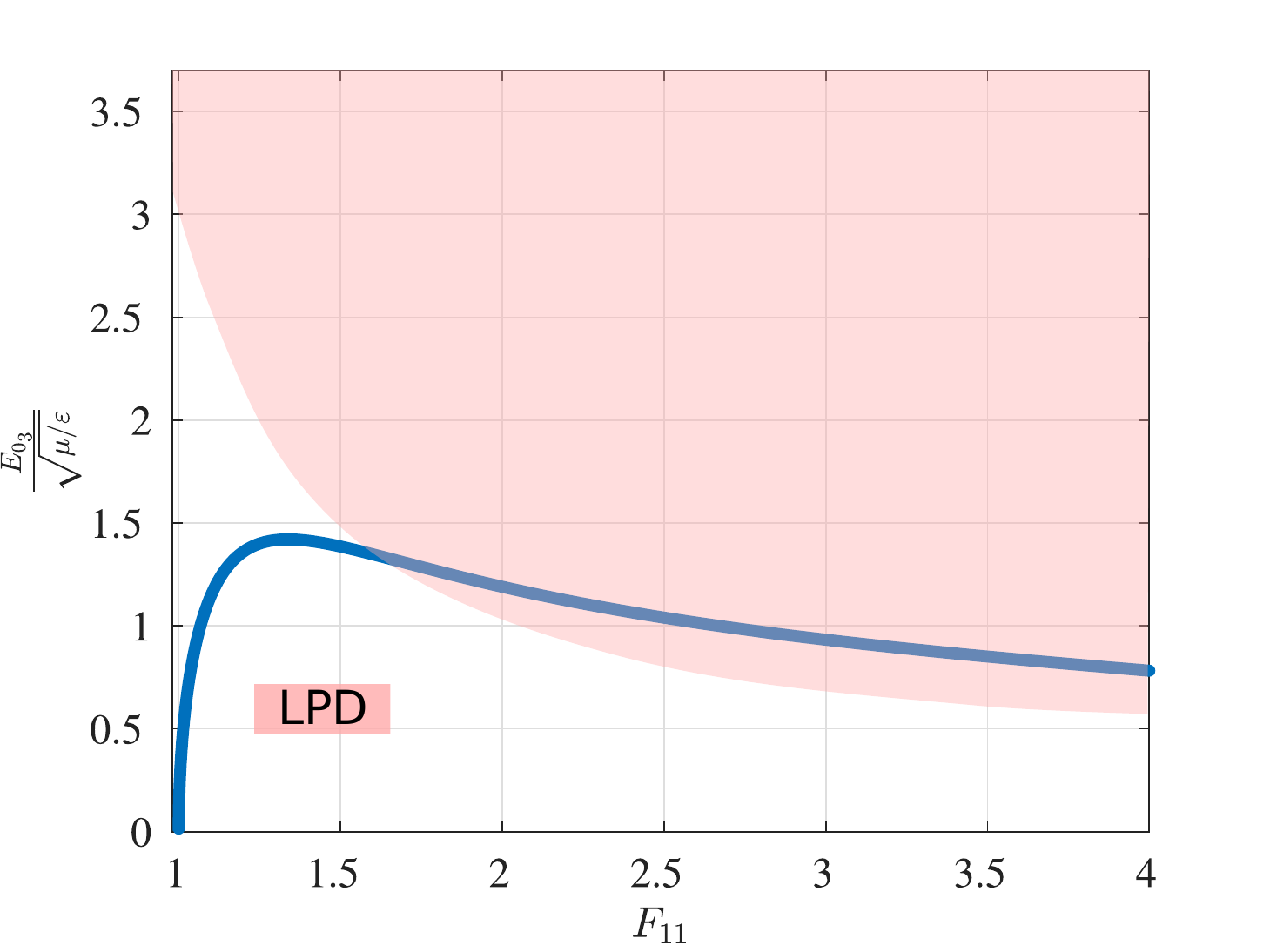}
         &  
\includegraphics[width=0.36\textwidth]{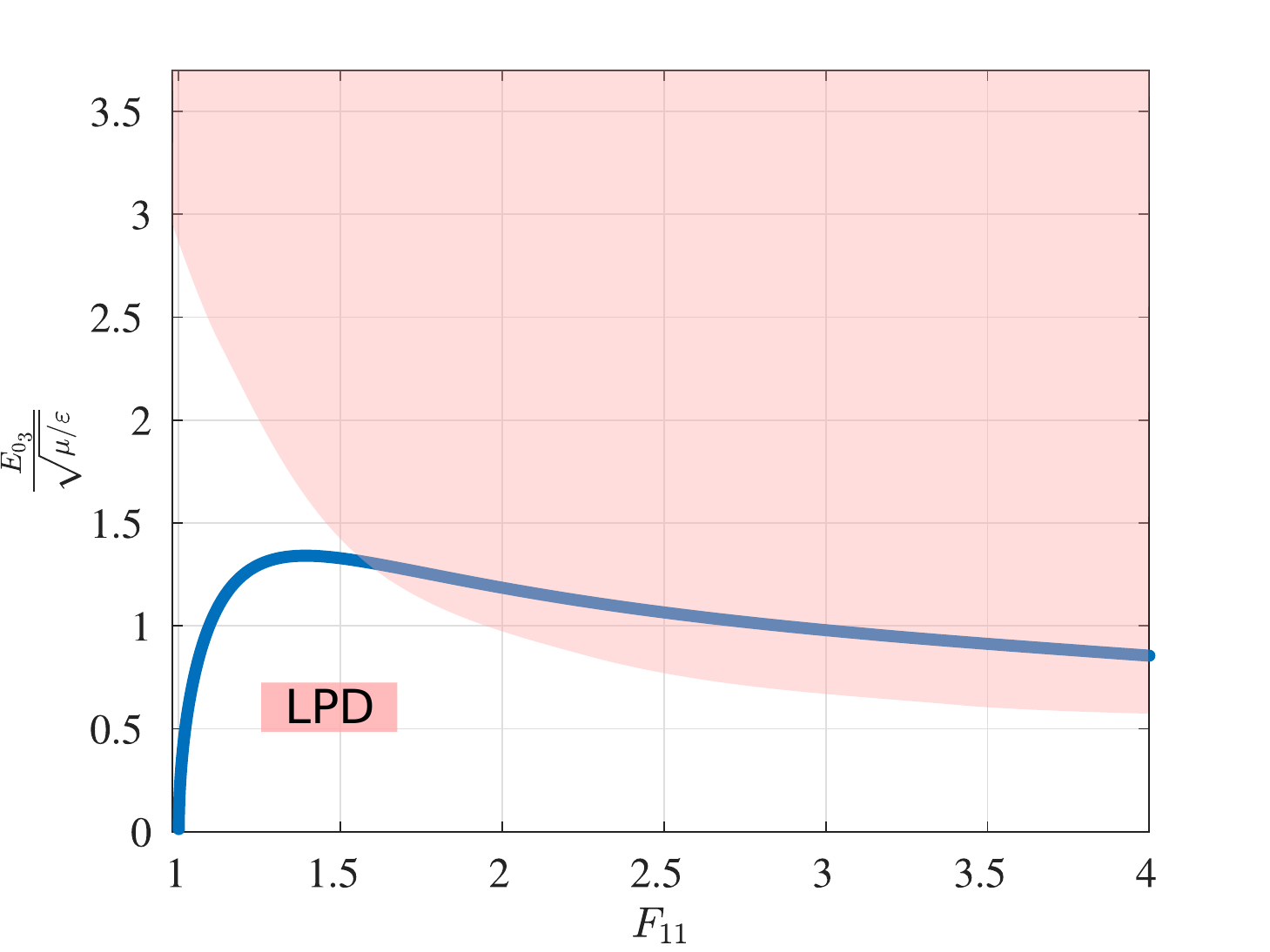}\\
(e)         &   (f)\\
\includegraphics[width=0.36\textwidth]{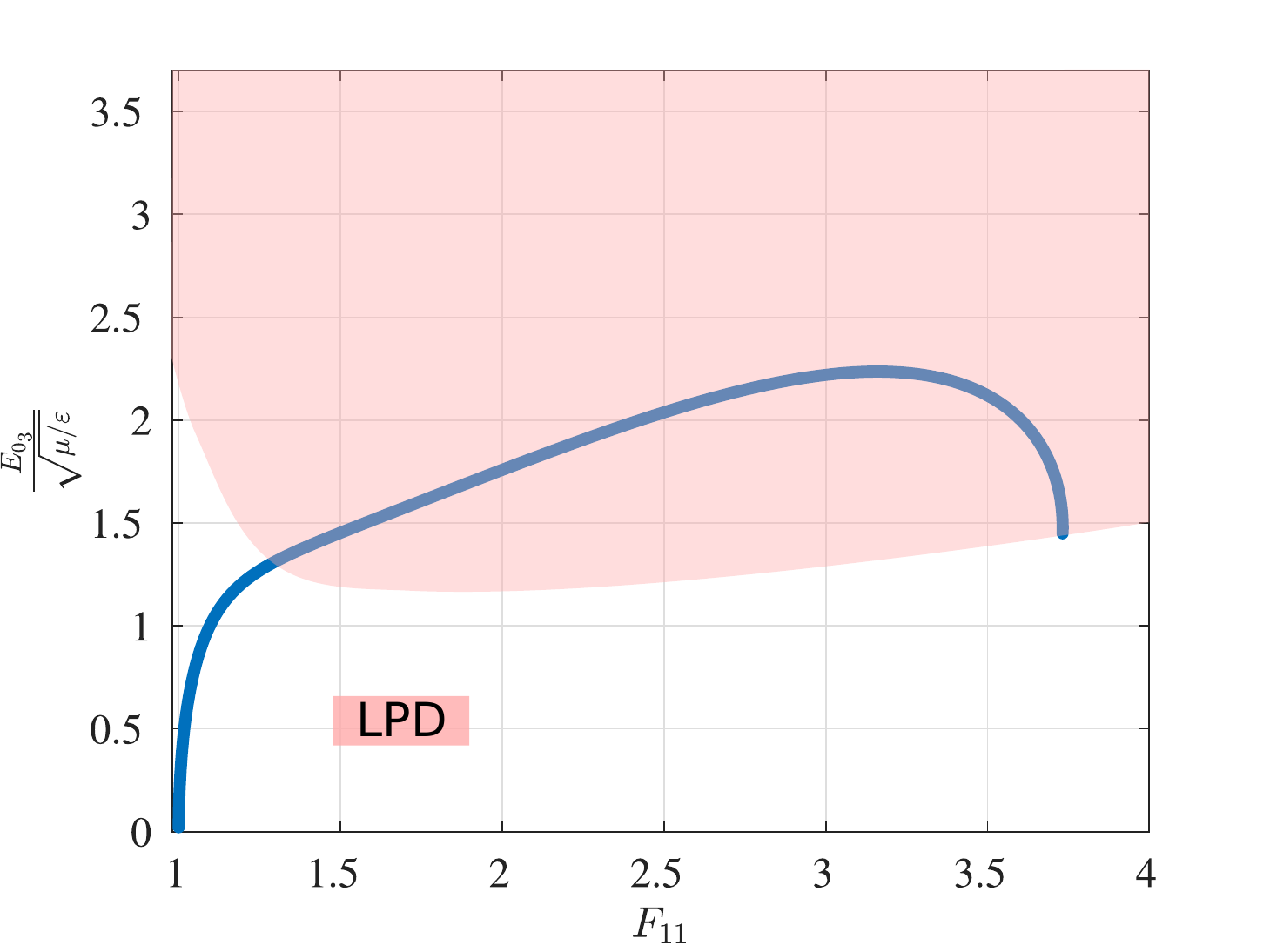}
         &  
\includegraphics[width=0.36\textwidth]{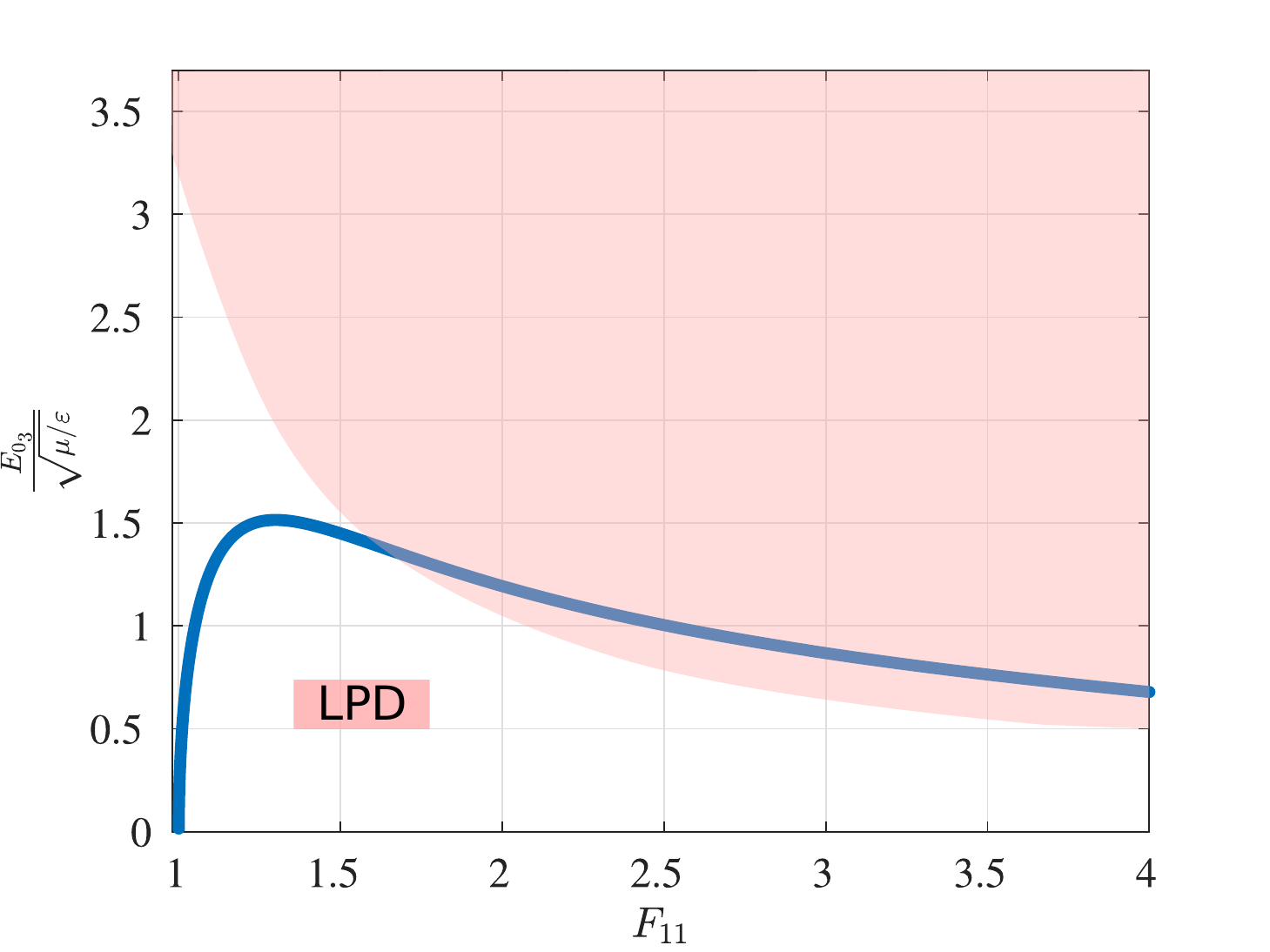}\\
(g)         &   (h)

    \end{tabular}
\caption{(a)-(b) Equilibrium paths for non-polyconvex model in equation \eqref{eqn:anisotropic non polyconvex} with material properties in Table \ref{Tab:MaterialParametersNonpolyconvex_Anisotropic} and for the five polyconvex models (see equation \eqref{eqn:anisotropic polyconvex model}) with material parameters in Table \ref{Tab:MaterialParametersPolyconvex_Anisotropic}. (c)-(h) represent the regions where loss of positive definiteness (LPD) of the Hessian operator and loss of ellipticity (LE) occurs for the non-polyconvex and polyconvex models. The reference values for the dimensionless electric field are $\mu=10^5$ and $\varepsilon=4.8\varepsilon_0$.}
    \label{fig:anisotropic model}
\end{figure}
\subsection{Numerical example 2}
The objective of this example are to:
\begin{itemize}
\item\label{Exp2_obj1} Demonstrate the influence of the fibre orientation on the mode of deformation of the EAP in the case of bending of an EAP in a realistic  three-dimensional setting.
\item\label{obj Exp1: arc-length} Study the loss of the generalized rank-one convexity.
\end{itemize}
Note that a similar example exploring the effect of fibre orientation on bending and twisting actuation of a fibre-reinforced magneto-elastomer has been studied in \cite{stanier2016fabrication}.
\begin{figure}[!h]	
    \centering
	\begin{tabular}{c}
    \includegraphics[width=0.60\textwidth,]{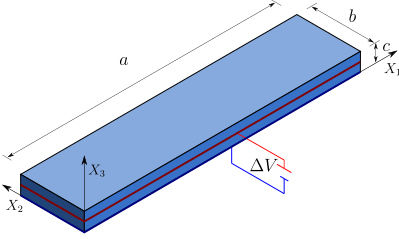}	
\end{tabular}
	\caption{Problem setting}
	\label{fig:elMechCantilever_problemSetting}
\end{figure}
 The ratio between length $a$, width $b$, and height $c$ of the EAP are taken as $c = a/50$, $b = a/6$. The sample, depicted in Figure \ref{fig:elMechCantilever_problemSetting}, is clamped at one end and is loaded by the application of electric potential in  the  electrodes  in the bottom and in the middle of the height. $\Delta \varphi = \Lambda h\sqrt{\mu_1/\varepsilon_1}$ where $\Lambda$ is a loading parameter, and $\mu_1$ and $\varepsilon$ are material parameters, see the model below and Table \ref{Tab:MaterialParametersModel1}. The calculations are performed using  the  second, i.e., mixed  Hu–Washizu variational  principle, see Appendix \ref{appendix:FEM} for more details. In addition, two different material models are considered. Symmetry $\mathcal{D}_{\infty h}$, applicable to electroactive polymers, is assumed in both cases.
\subsubsection{Model 1:  $\mathcal{A}$-polyconvex constitutive law}
 The energy is constructed to be  $\mathcal{A}$-polyconvex meaning that it is composed of $\mathcal{A}$-polyconvex invariants only. The proposed model is defined through the energy written as 
\begin{subequations}\label{eqn:the model in example 2}
\begin{align}
\bar{e}_{\mathcal{D}_{\infty h},1}(\vect{F},\vect{D}_0,\vect{M}) &= \bar{e}_1(\vect{F}) + \bar{e}_2(\vect{F},\vect{M}) + \bar{e}_3(\vect{F},\vect{D}_0)+ \bar{e}_4(\vect{F},\vect{D}_0,\vect{M})
\\
\bar{e}_1(\vect{F}) &= \frac{\mu_1}{2}I_1^{\text{dev}}(\vect{F}) + \frac{\mu_2}{2}I_{2,\text{pol}}^{\text{dev}}(\vect{F}) + \frac{\lambda}{2}(I_3(\vect{F})-1)^2
\\
\bar{e}_2(\vect{F},\vect{M}) &= \frac{\mu_3}{2\alpha}I_4^\alpha(\vect{F},\vect{M}) + \frac{\mu_3}{2\beta}I_5^\beta(\vect{F},\vect{M}) - \mu_3 {\rm ln}I_3(\vect{F})
\\
\bar{e}_3(\vect{F}, \vect{D}_0) &= \frac{1}{2 \varepsilon_1} \frac{I_7(\vect{F},\vect{D}_0)}{I_3(\vect{F})} + \frac{1}{2\varepsilon_2} I_6(\vect{D}_0)
\\
\bar{e}_4(\vect{D}_0) &= \frac{1}{2\varepsilon_3}K_{1}^{\mathcal{D}_{\infty h}}(\vect{D}_0)
\end{align}
\end{subequations}

\begin{figure}[!h]	
\centering
\begin{tabular}{ccc}
\includegraphics[width=0.20\textwidth]{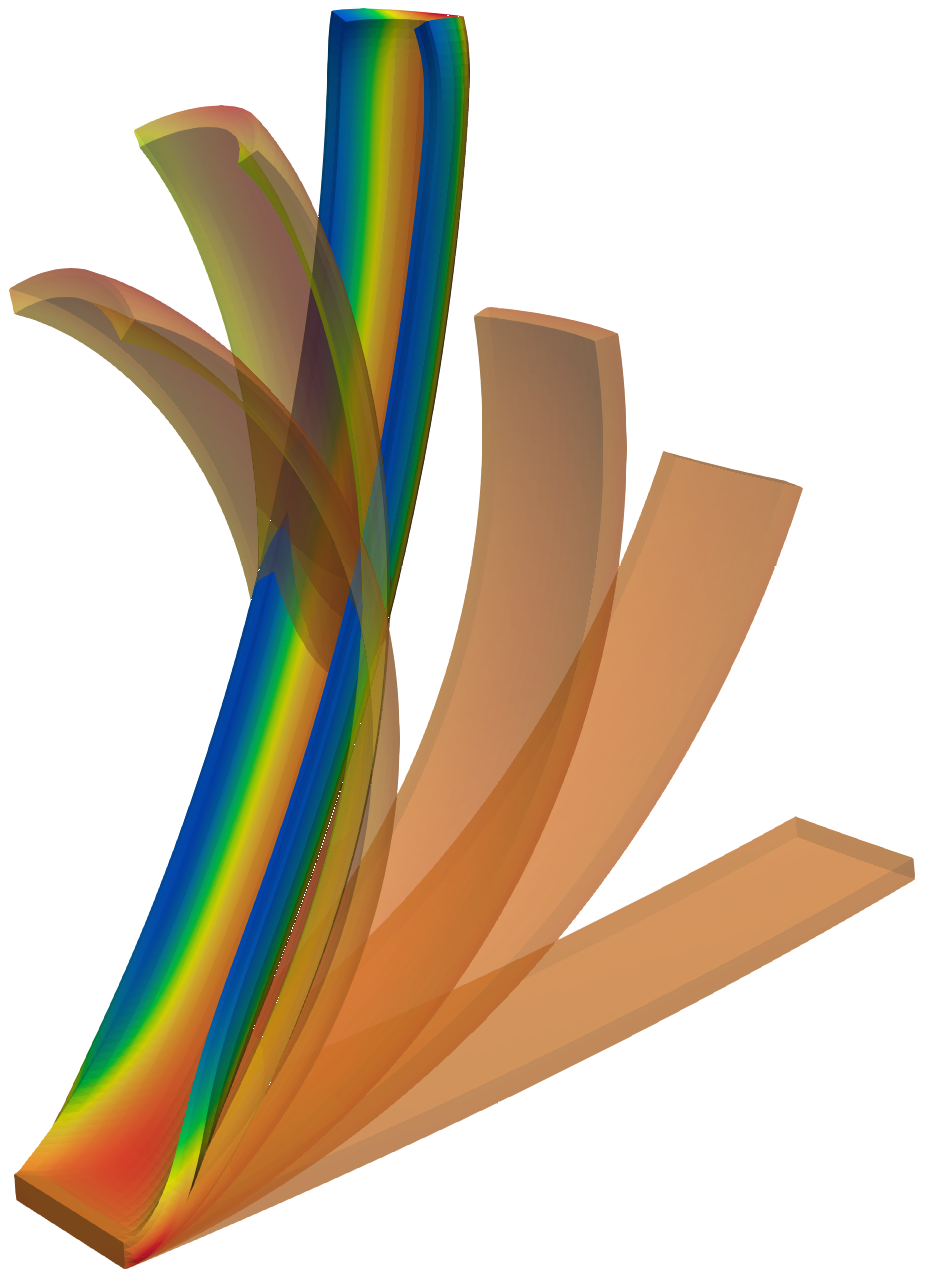}
&
\includegraphics[width=0.20\textwidth]{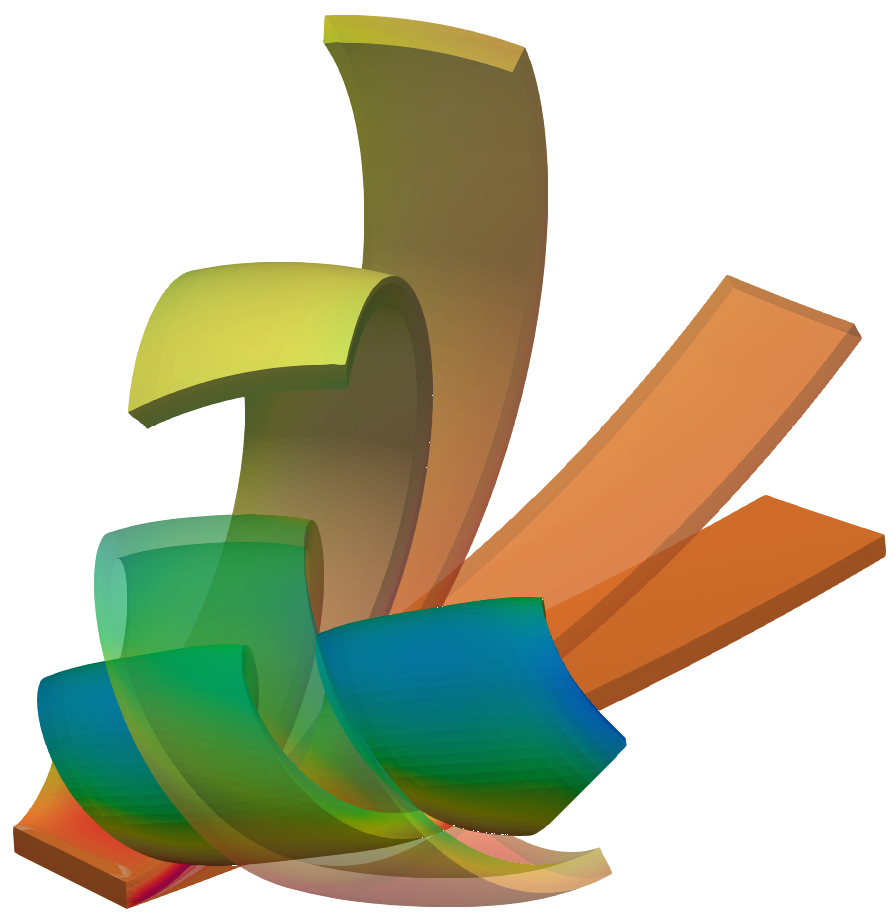}
&
\includegraphics[width=0.20\textwidth]{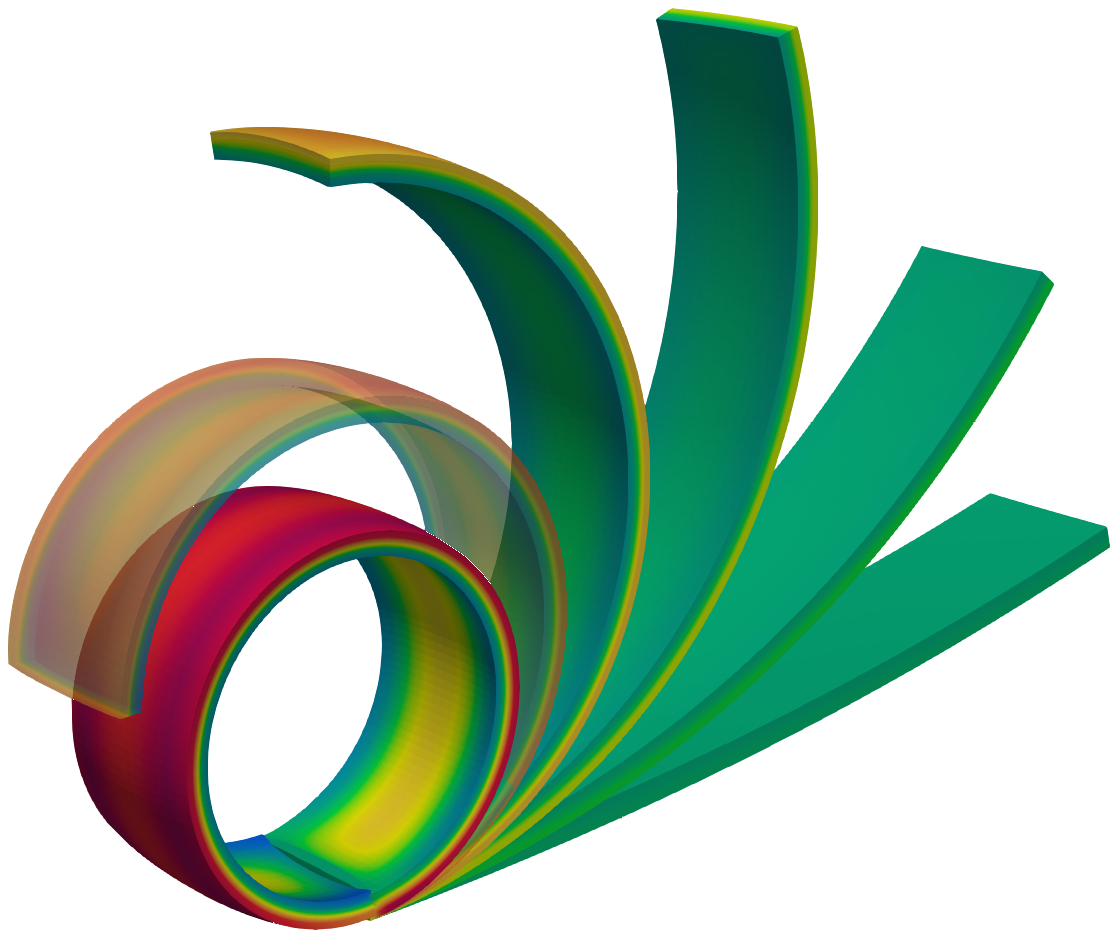}
\\
\includegraphics[width=0.20\textwidth]{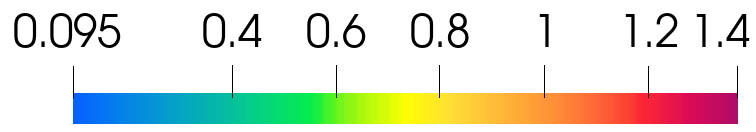}
&
\includegraphics[width=0.20\textwidth]{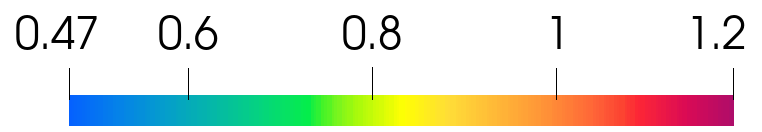}
&
\includegraphics[width=0.20\textwidth]{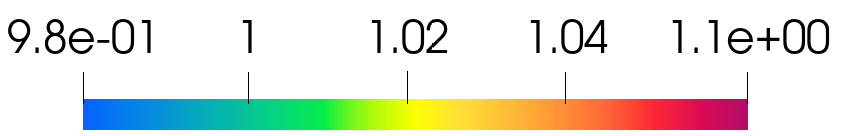}
\\
\\
\\
\includegraphics[width=0.20\textwidth]{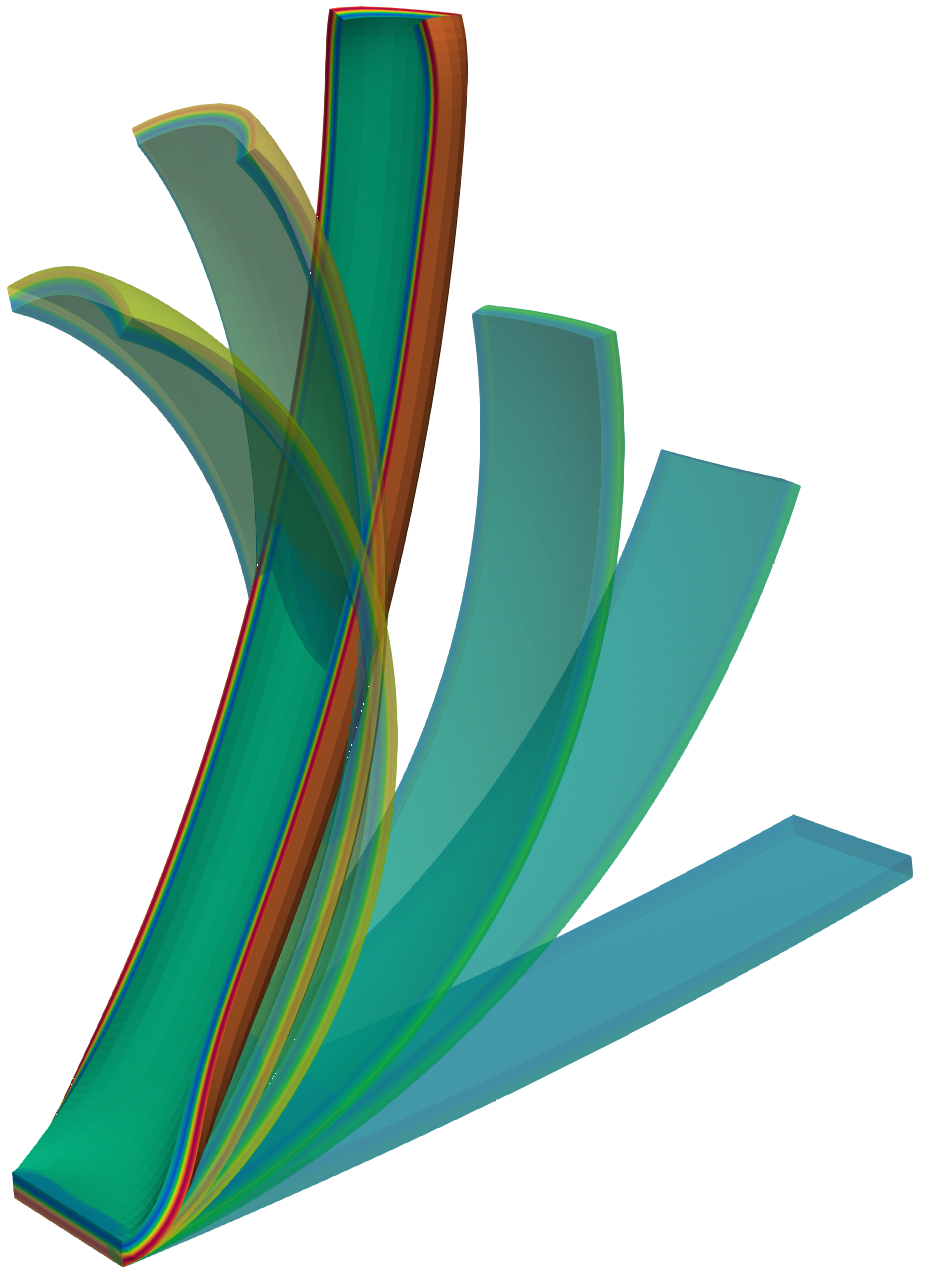}
&
\includegraphics[width=0.20\textwidth]{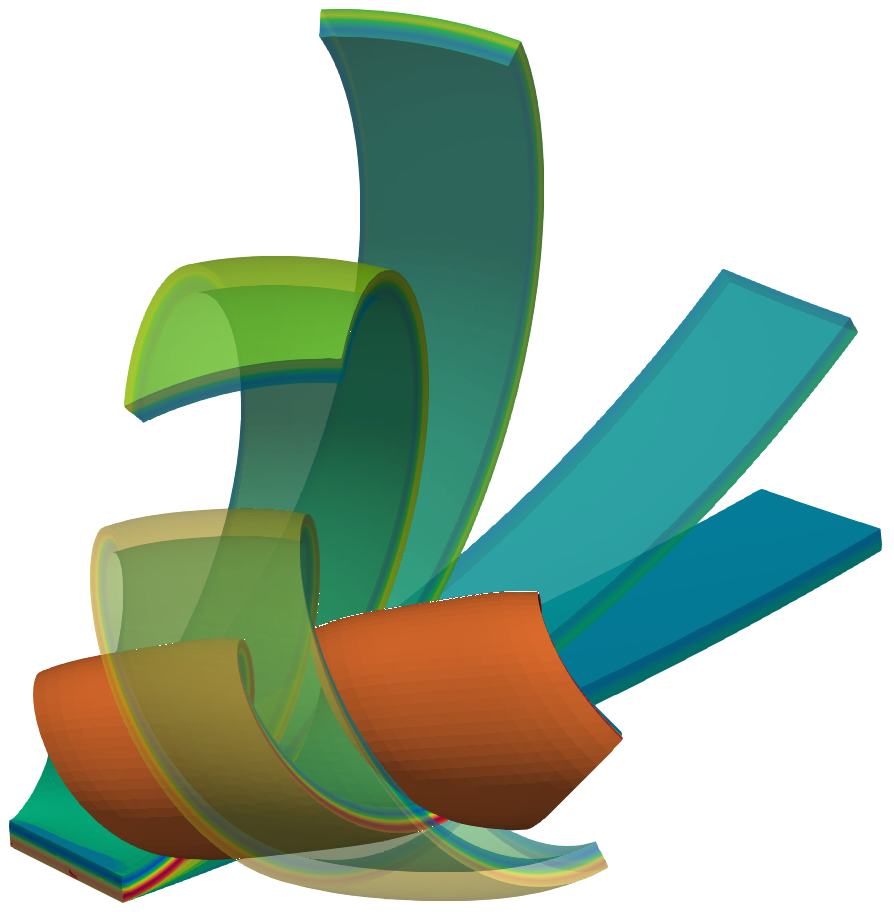}
&
\includegraphics[width=0.20\textwidth]{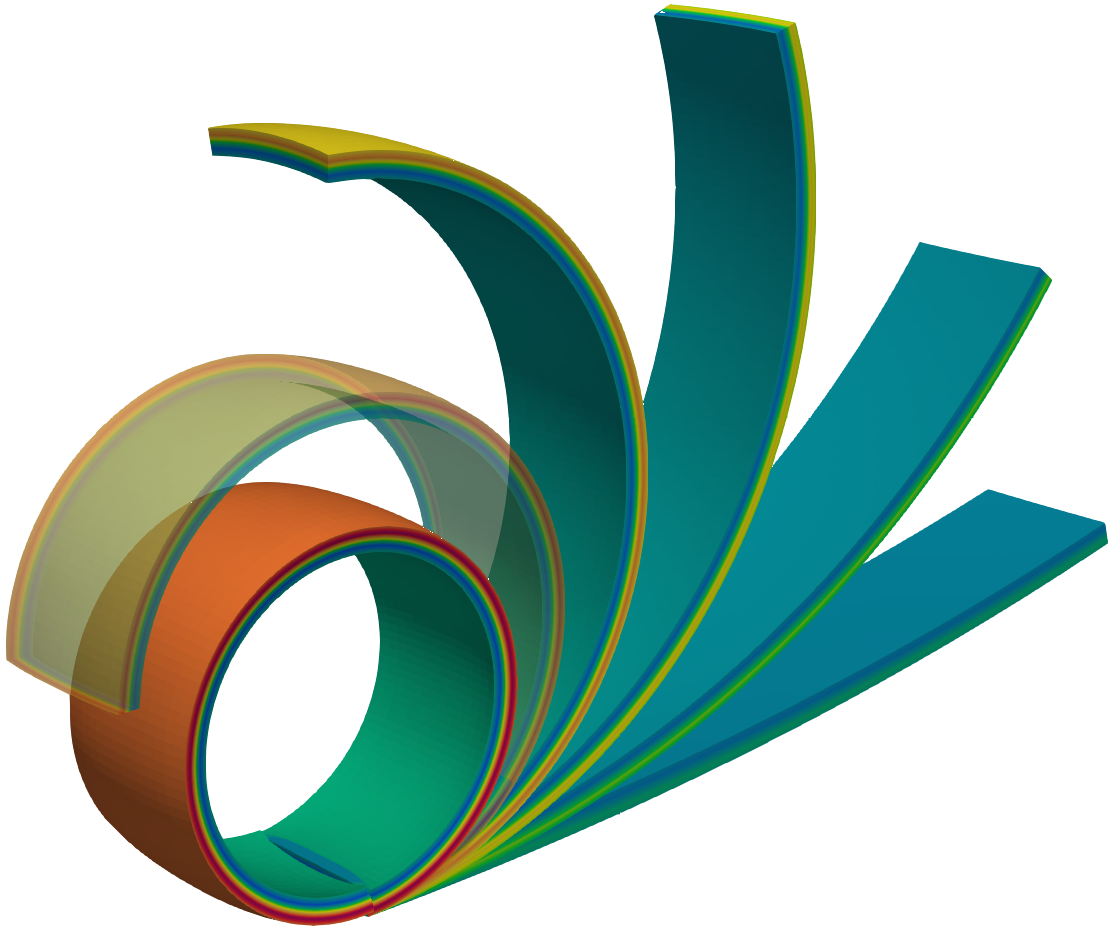}
\\
\includegraphics[width=0.20\textwidth]{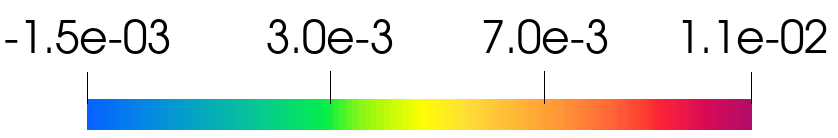}
&
\includegraphics[width=0.20\textwidth]{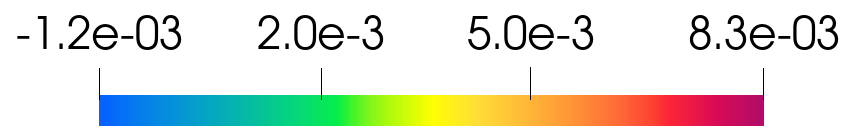}
&
\includegraphics[width=0.20\textwidth]{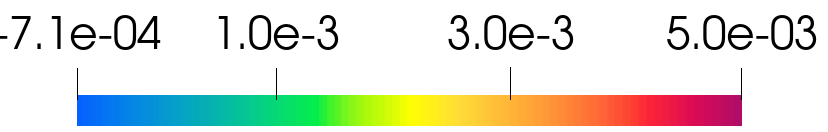}
\\
\\
\\
\includegraphics[width=0.20\textwidth]{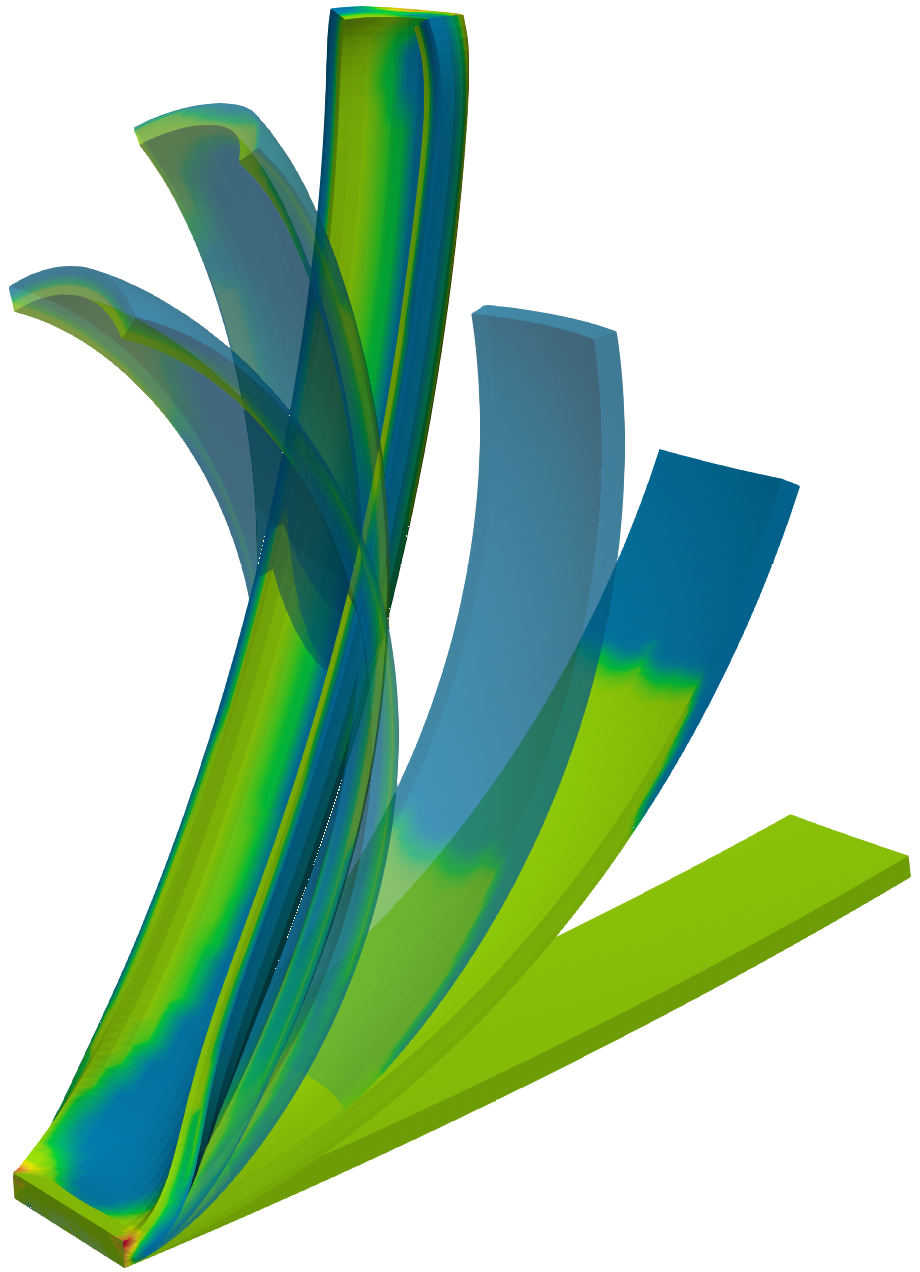}
&
\includegraphics[width=0.20\textwidth]{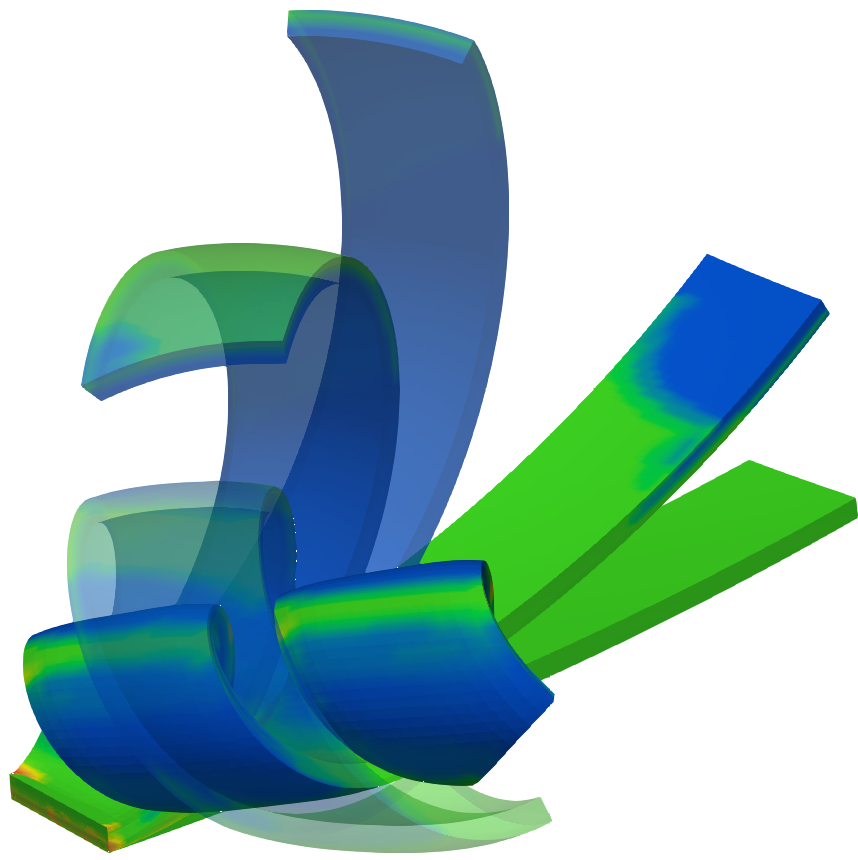}
&
\includegraphics[width=0.20\textwidth]{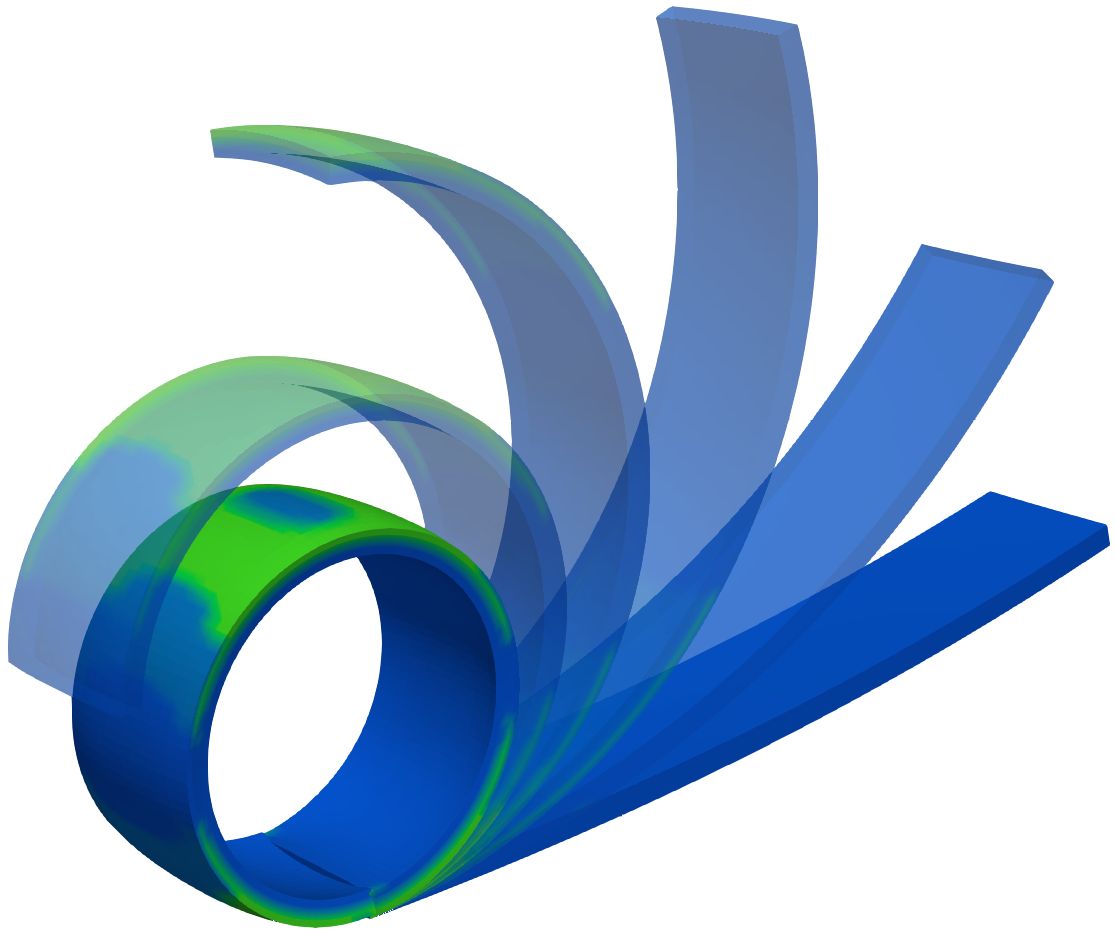}
\\
\includegraphics[width=0.20\textwidth]{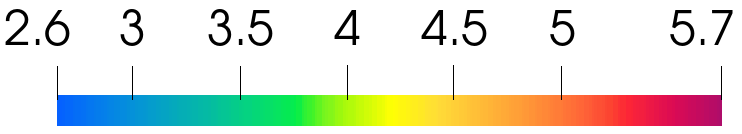}
&
\includegraphics[width=0.20\textwidth]{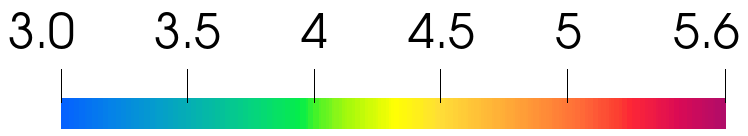}
&
\includegraphics[width=0.20\textwidth]{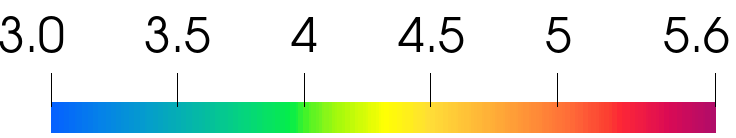}
\end{tabular}
\caption{Deformation of an EAP with different fiber alignment, columns from left to right correspond to $\vect{M} = \begin{bmatrix}1 & 0 & 0  \end{bmatrix}^T$, with the snapshots corresponds to loading parameter $\Lambda = (0,10,15,30,60)$; $\vect{M} = \begin{bmatrix}1/\sqrt{2} & 1/\sqrt{2} & 0  \end{bmatrix}^T$, $\Lambda = (0,5,15,20,30,45)$, and $\vect{M} = \begin{bmatrix}0 & 1 & 0  \end{bmatrix}^T$, $\Lambda = (0,5,10,15,20,30)$. Contour plots in the first row displays deformation gradient component $F_{22}$, the second row displays electric displacement component $(D_0)_z$, and the last row displays minimal value of minor of acoustic tensor normalized by $\mu_1$.}
\label{fig:example_3angles}
\end{figure}
Of course, to assure $\mathcal{A}$-polyconvexity of the model, positive material constants have to be chosen, which is a reasonable physical requirement leading also to coercivity.  The material constants used in the following simulation are presented in Table \ref{tab:Example2}.

\begin{table}[htb] \label{Tab:MaterialParametersModel1}
\centering
\caption{Material properties for Material model 1}
     \label{tab:Example2}
     \begin{tabular}{cccccccc}
     \toprule
     {$\mu_1$} & {$\mu_2$} & {$\mu_3$} & {$\lambda$} & {$\varepsilon_1$}& {$\varepsilon_2$}& {$\varepsilon_3$}
     \\
     \midrule
     $1\times10^5$ & $1.0\mu_1$ & $1.0\mu_1$ & $10^3\times\mu_1$ & $9.6\varepsilon_0$& $2.88\varepsilon_0$& $1.92\varepsilon_0$& \\
     \bottomrule
\end{tabular}
 \end{table}

We study three samples with fibre alignment of $0^{\circ}$, $45^{\circ}$, and $90^{\circ}$ and evaluate the influence of the fibre orientation  upon the electrical loading. The three aforementioned fiber orientations corresponds to $\vect{M}_0 = \begin{bmatrix}1 & 0 & 0  \end{bmatrix}^T$, $\vect{M}_{45} = \begin{bmatrix}1/\sqrt{2} & 1/\sqrt{2} & 0  \end{bmatrix}^T$, and $\vect{M}_{90} =\begin{bmatrix} 0 & 1 & 0  \end{bmatrix}^T$ respectively. The results of simulation are depicted in Figure \ref{fig:example_3angles} showing clearly that design of fibers alignment substantially influence response of an EAP leading to bending, twisting, or combined modes of deformation. Therefore, the computational approach can be used to carefully tailor fiber orientation to obtain response of an EAP required for a specific application.
\subsubsection{Model 2: Non  $\mathcal{A}$-polyconvex constitutive law}
In the previous example, we have examined the behaviour of a fibre-reinforced EAP with a constitutive law composed of $\mathcal{A}$-polyconvex invariants. However, the model does not contain electro-mechanical coupling of the fibres described by invariant $K_2^{\mathcal{D}_{\infty h}}$. The model from the previous section is enhanced as
\begin{align}\label{eq:ed_no}
\bar{e}_{\mathcal{D}_{\infty h},2}(\vect{F},\vect{D}_0,\vect{M}) = \bar{e}_{\mathcal{D}_{\infty h},1}(\vect{F},\vect{D}_0,\vect{M})  + \frac{1}{2\varepsilon_4} K_2^{\mathcal{D}_{\infty h}}(\vect{F},\vect{D}_0,\vect{M})
\end{align}
Note that coupled electro-mechanical behaviour of fibres is often neglected see, e.g., \cite{ahmadi2020nonlinear,allahyari2021fiber}. 

Since the model contains a non $\mathcal{A}$-polyconvex contribution, it is difficult to prove the $\mathcal{A}$-polyconvex of the model. Therefore, a possible loss  of  rank-one convexity is investigated for the case with fixed fiber orientation defined by $\vect{M}=\begin{bmatrix} 1/\sqrt{3}& 1/\sqrt{3} & 1/\sqrt{3} \end{bmatrix}^T$, and material parameter $\varepsilon_4 = 19.2\varepsilon_0$. 

The loss of rank-one convexity, i.e., ellipticity can be detected through evaluation of the  minors of the  acoustic  tensor (\ref{eqn: Exp1 acoustic tensor}). However, evaluation of ellipticity loss is computationally demanding and not practical for real three-dimensional simulations. Note that several efficient strategies for the detection of loss of ellipticity have been proposed, see, e.g., \cite{al2020general} and references therein.

It is well known that homogenization of fibre-reinforced composites with the matrix and fibres described by polyconvex  material laws does not generally lead to a polyconvex model at the macroscale, i.e., polyconvexity is not preserved by homogenization \cite{debotton2006neo, barchiesi2007loss}. It can be demonstrated that the loss of polyconvexity of fibre-reinforced elastomers has a clear physical meaning of a long wave-length instability \cite{slesarenko2017microscopic}.

The results of the simulation for three levels of $\Delta \varphi$ are shown in Figure \ref{fig:example_acousticTensor}. The contour plots display values of the minimal minor of the acoustic tensor normalized by $\mu_1$. The results in case (a) show only positive values of the minor. The value is approaching zero in one element in case (b), and negative values can be detected in case (c), meaning that the rank-one convexity has been lost. The loss of rank-one convexity proves that the proposed energy (\ref{eq:ed_no}) is not $\mathcal{A}$-polyconvex. Even though the loss of rank-one convexity indicates that the governing equations of the electro-mechanical problem become ill-posed and the simulation should be terminated, the evolution of loss of ellipticity is further demonstrated in Figures \ref{fig:example_acousticTensor}(d). 
\begin{figure}[!h]	
\centering
\begin{tabular}{cc}
(a) &(b)
\\
\\
\includegraphics[width=0.25\textwidth]{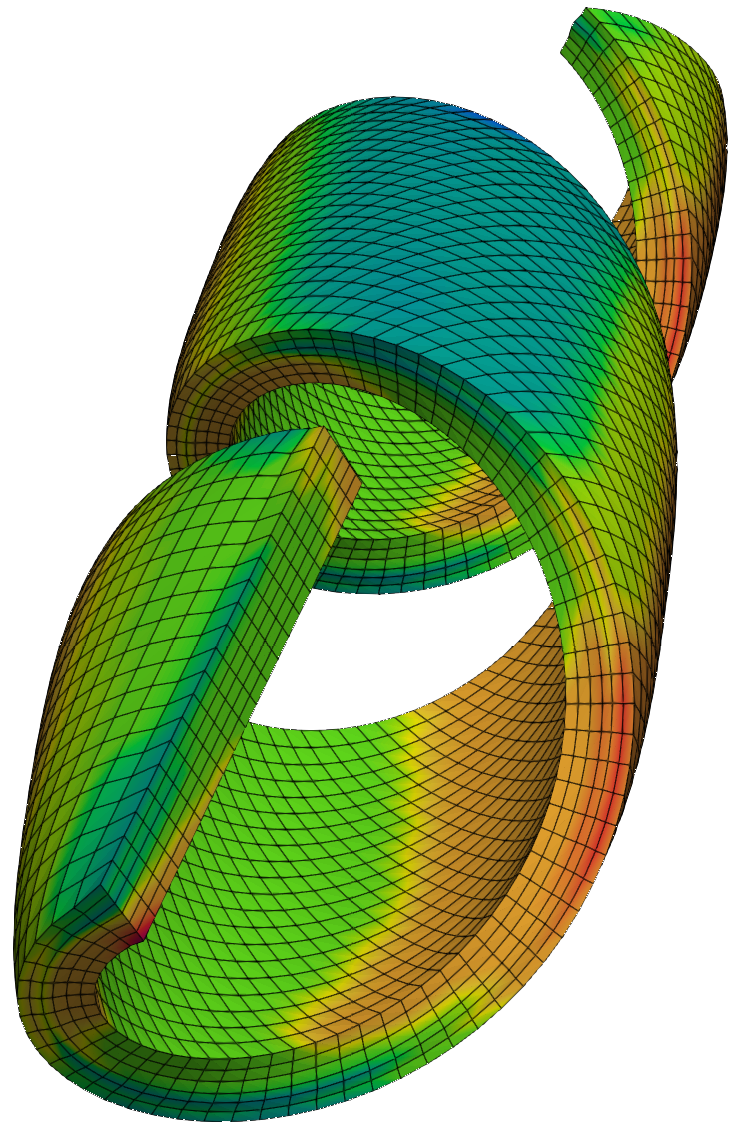}
&
\includegraphics[width=0.25\textwidth]{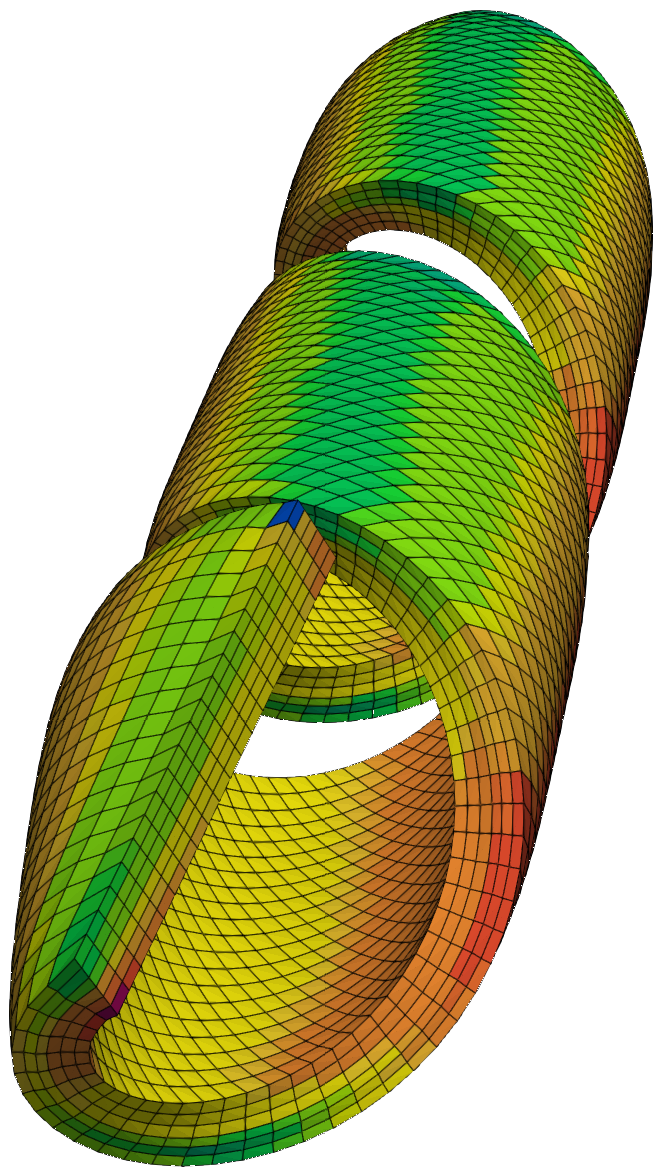}
\\
\includegraphics[width=0.25\textwidth]{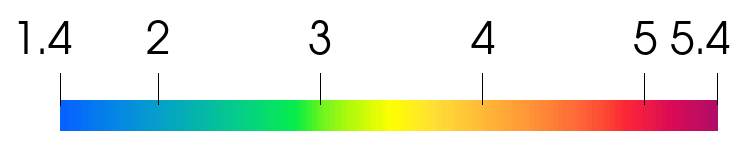}
&
\includegraphics[width=0.25\textwidth]{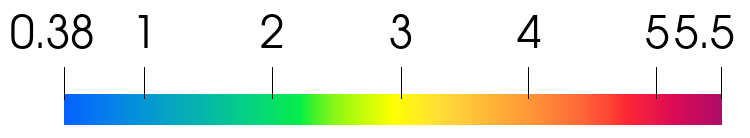}
\\
\\
(c) &(d)
\\
\\
\includegraphics[width=0.25\textwidth]{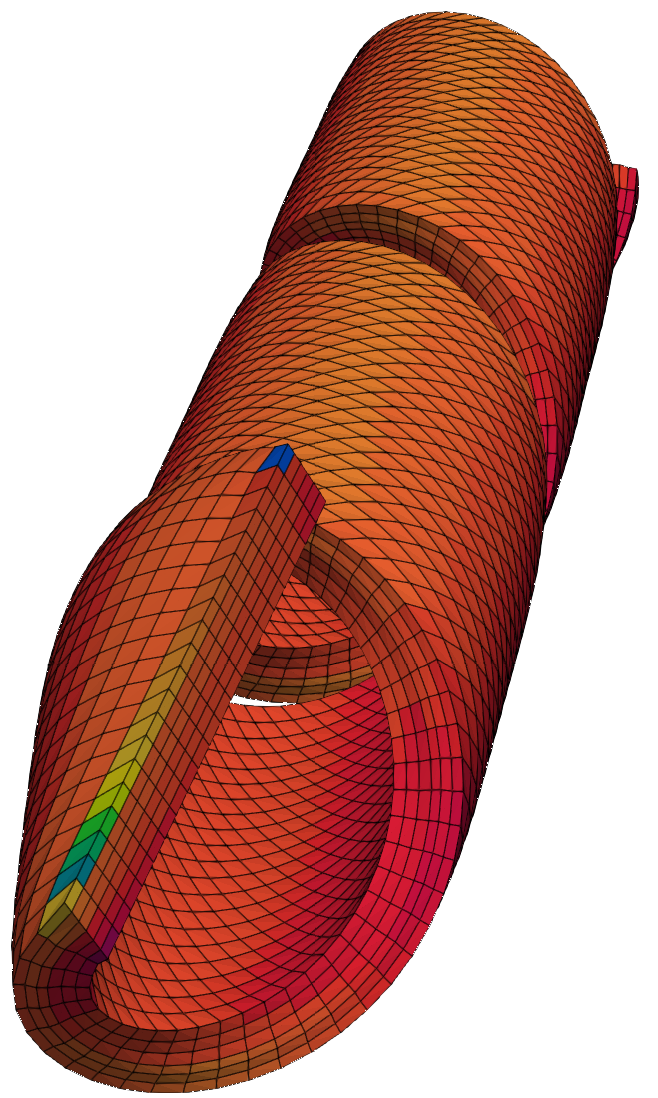}
&
\includegraphics[width=0.25\textwidth]{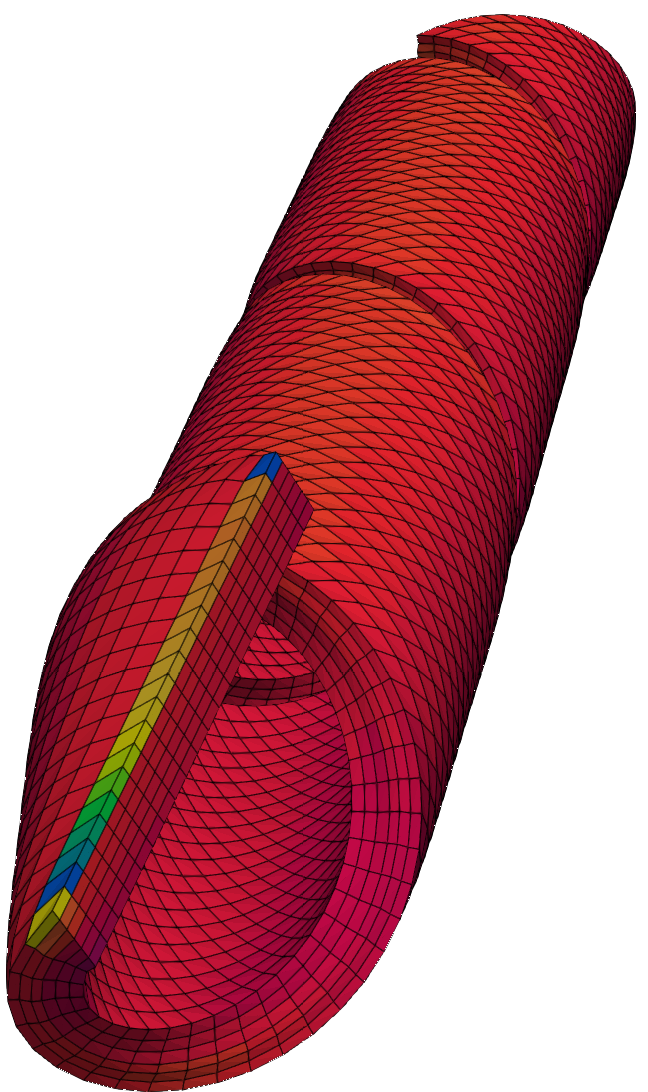}
\\
\includegraphics[width=0.25\textwidth]{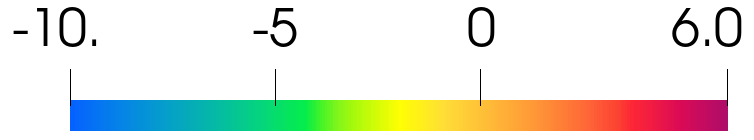}
&
\includegraphics[width=0.3\textwidth]{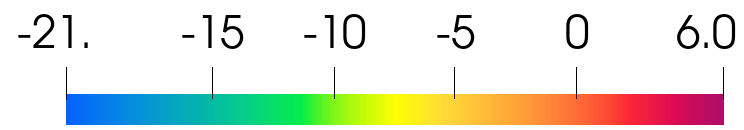}
\end{tabular}
\caption{Deformation of an EAP reinforced with fibers in direction $\vect{M} = \begin{bmatrix} 1/\sqrt{3} & 1/\sqrt{3} & 1/\sqrt{3}\end{bmatrix}^T$ for four levels of loading parameter, namely, $\Lambda = (40,44,47,49)$. Contour plots display minimal minor of the acoustic tensor normalized by $\mu_1$. The negative value of a minor of the acoustic tensor means that rank-one convexity has been lost.}
\label{fig:example_acousticTensor}
\end{figure}

\subsection{Numerical example 3}
The objectives of this example are to:

\begin{itemize}

\item\label{obj Exp3: buckling}  Detect the onset of instabilities on a rank-one DE laminated composite square membrane.
	
\item\label{obj Exp3: pattern} Observe the different buckling patterns in terms of the preferred orientation of transverse isotropy.

\end{itemize}


The geometry for this numerical example is given by a square membrane of side $l=0.06$ m and thickness $h=0.001$ m, clamped along all its side faces, as represented in Figure \ref{fig: Exp3 boundary conditions}. The membrane is subjected to a prescribed electric surface charge on its base while grounded to zero potential on its topside. Geometrical and Finite Element simulation parameters are presented in Table \ref{tab: Exp3 Geometry}. $Q3$ Finite Elements are used to interpolate both displacement and electric potential, being $104,188$ the total number of degrees of freedom, and the surface charge is applied incrementally, being $\Lambda \in [0,1]$ the load factor. The two-field variational formulation $\Pi_{\Psi}\left(\vect{\phi},\varphi\right)$ in equation \eqref{eqn:displacement based formulation} has been considered using the aforementioned Finite Element discretisation. A volumetric force acting along the positive $0X_3$ axis of value $10$, $\text{N}\cdot \text{kg}^{-1}$ is considered.

\begin{figure}[h!]
	\centering
	\begin{tabular}{c}
		\hspace{-0.6cm}		\includegraphics[width=0.5\textwidth]{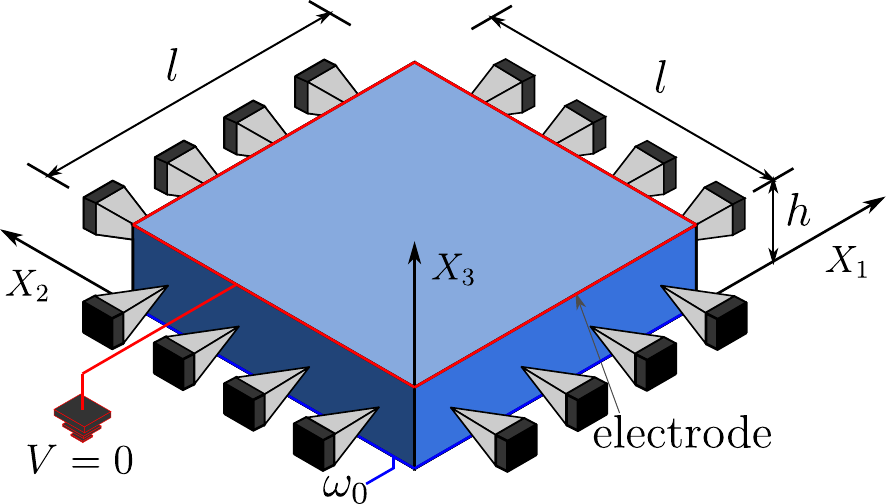}  
	\end{tabular}
	\caption{Numerical example 4. Geometry and boundary conditions. }
	\label{fig: Exp3 boundary conditions}
\end{figure}

\begin{table}[h!]
	\centering
	\begin{tabular}{l c c l l c c l}
		\toprule
		Geometrical parameters &  $l$ & $0.06$    & $\text{m}$  & Simulation parameters  & $N_x$ & $60$ &\\
		&  $h$ & $0.001$   & $\text{m}$  &                        & $N_y$ & $60$ &\\
		&      &           &             &                        & $N_z$ & $2$ &\\
		\vspace{1mm}
		Electric charge	&  $\omega_0$ & $0.02$ & $\text{C}/\text{m}^2$ & Newton tolerance &   & $10^{-6}$ & \\
		\noalign{\smallskip}\hline
	\end{tabular}
	\caption{Numerical example 4. Geometrical and simulation parameters. }
	\label{tab: Exp3 Geometry}
\end{table}

A similar study was conducted in \cite{MARIN_1,MARIN_2}, where a squared membrane, made of a rank-one laminated DE, was subjected to electrical stimuli in order to observe the development of instabilities as it deforms.  In this work, we consider the transversely isotropic constitutive model given in equation \eqref{eqn:the model in example 2}.
The material parameters featuring in the model can be found in Table \ref{Tab:MaterialParametersModel example 3}.

\begin{table}[htb] \label{Tab:MaterialParametersModel example 3}
\centering
\caption{Material properties for example 4}
     \label{tab:Example2}
     \begin{tabular}{cccccccc}
     \toprule
     {$\mu_1$} & {$\mu_2$} & {$\mu_3$} & {$\lambda$} & {$\varepsilon_1$}& {$\varepsilon_2$}& {$\varepsilon_3$}
     \\
     \midrule
     $1\times10^5$ & $1.0\mu_1$ & $1.0\mu_1$ & $10^3\times \mu_1$ & $9.6\varepsilon_0$& $2.88\varepsilon_0$& $1.92\varepsilon_0$& \\
     \bottomrule
\end{tabular}
 \end{table}

With regards to the preferred direction $\vect{M}$, six values have been considered (see Figure \eqref{fig:example_Buckling}). The electrically induced deformation pattern attained by the elastomeric material is indeed highly influenced by the preferred direction $\vect{M}$, as it can be clearly appreciated from Figure \ref{fig:example_Buckling}. Naturally, this influence is also reflected in the contour plot pressure distribution over the elastomeric material, as shown in Figure \ref{fig:example_Buckling pressure}.

\begin{figure}[htbp!]	
\centering
\begin{tabular}{cc}
\includegraphics[width=0.45\textwidth]{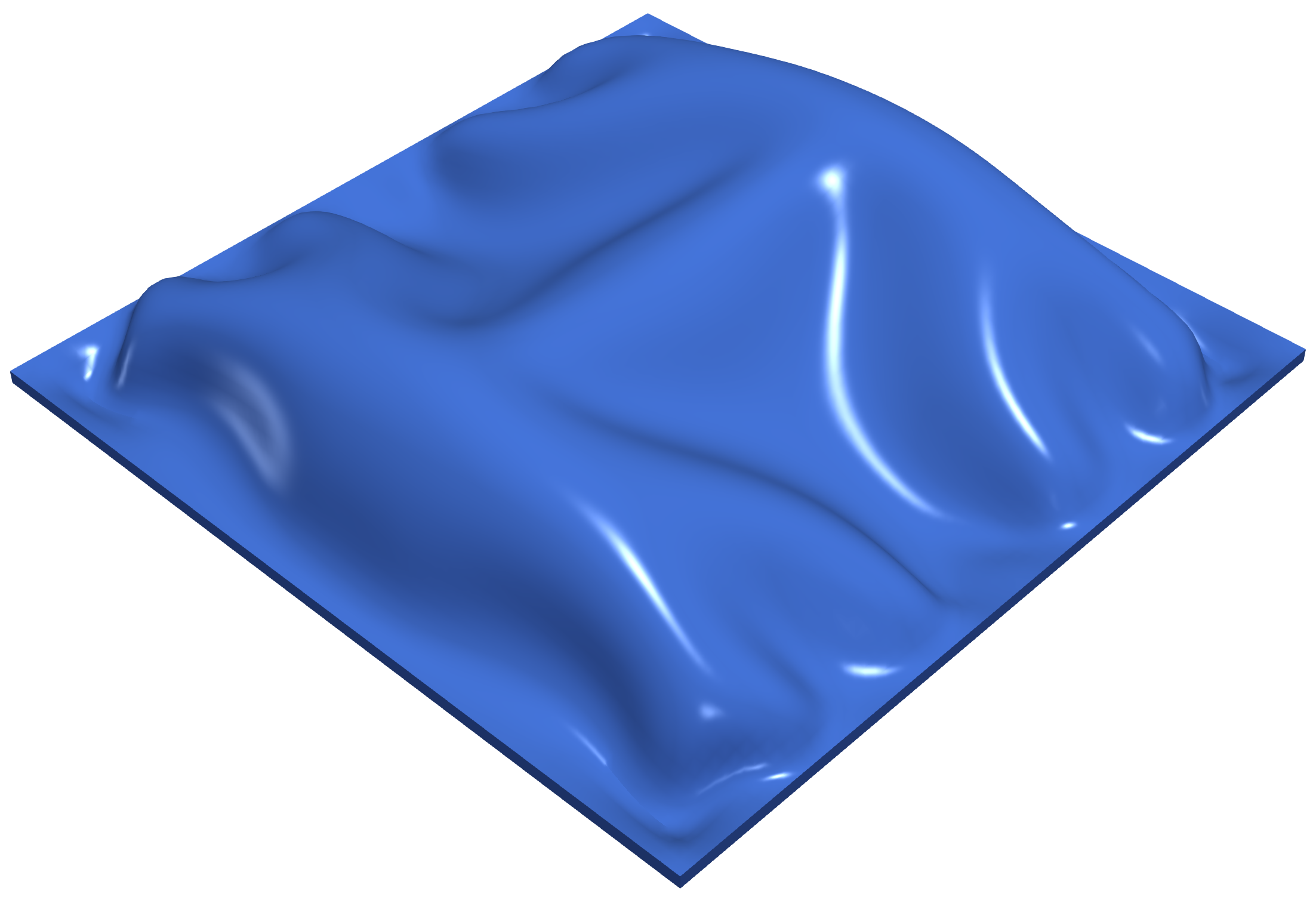}
&
\includegraphics[width=0.45\textwidth]{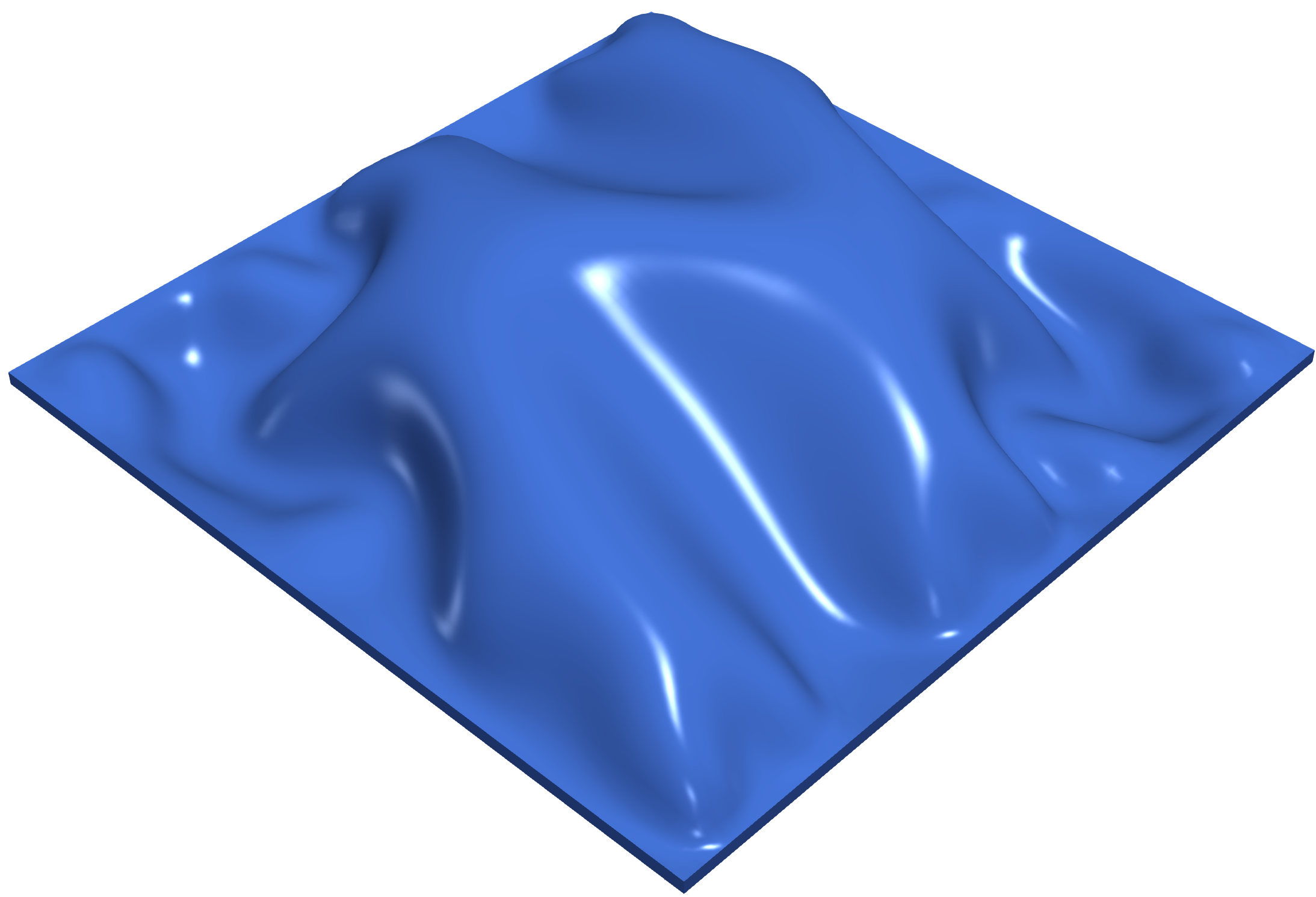}\\
(a)  &  (b)\\

\includegraphics[width=0.45\textwidth]{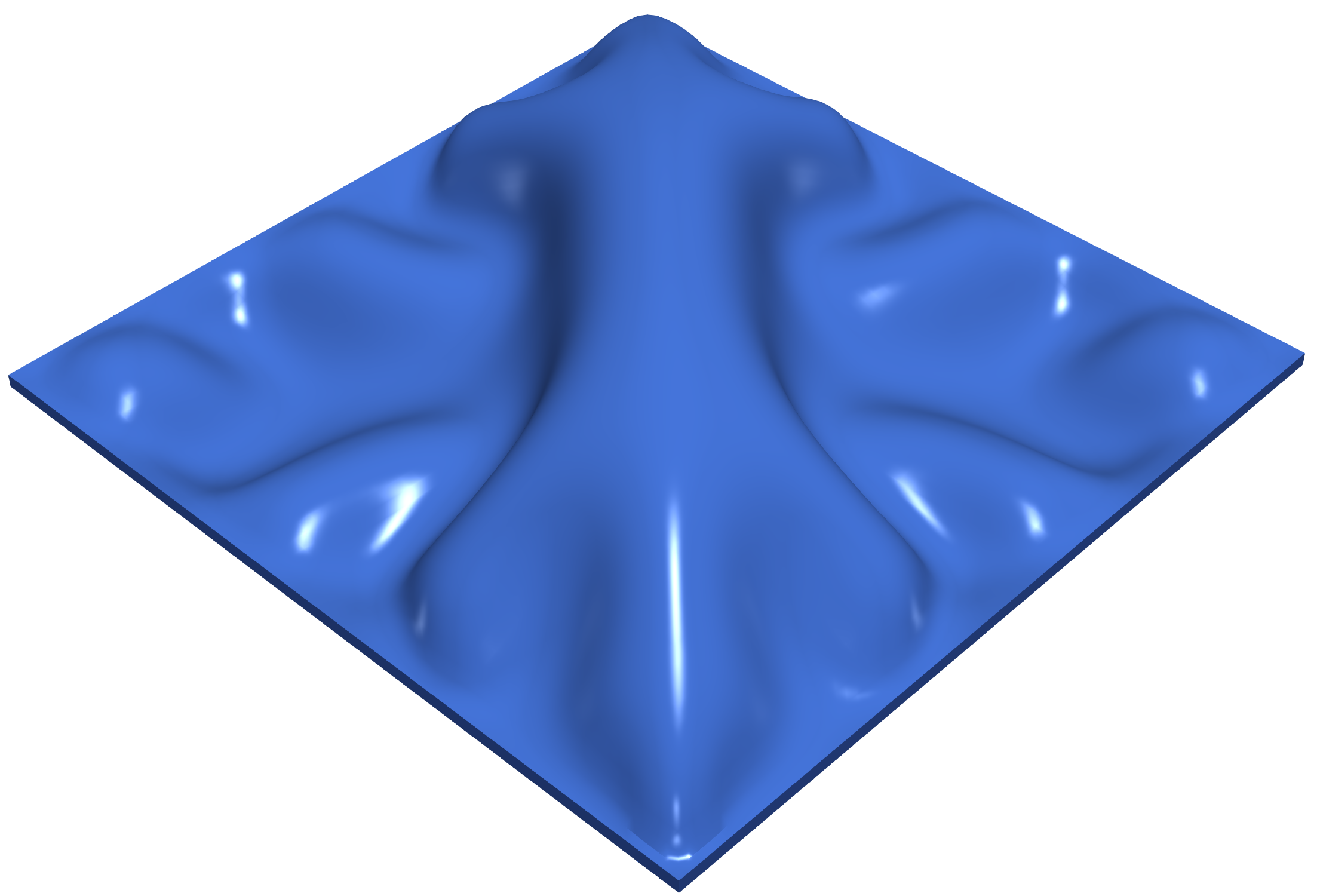}&

\includegraphics[width=0.45\textwidth]{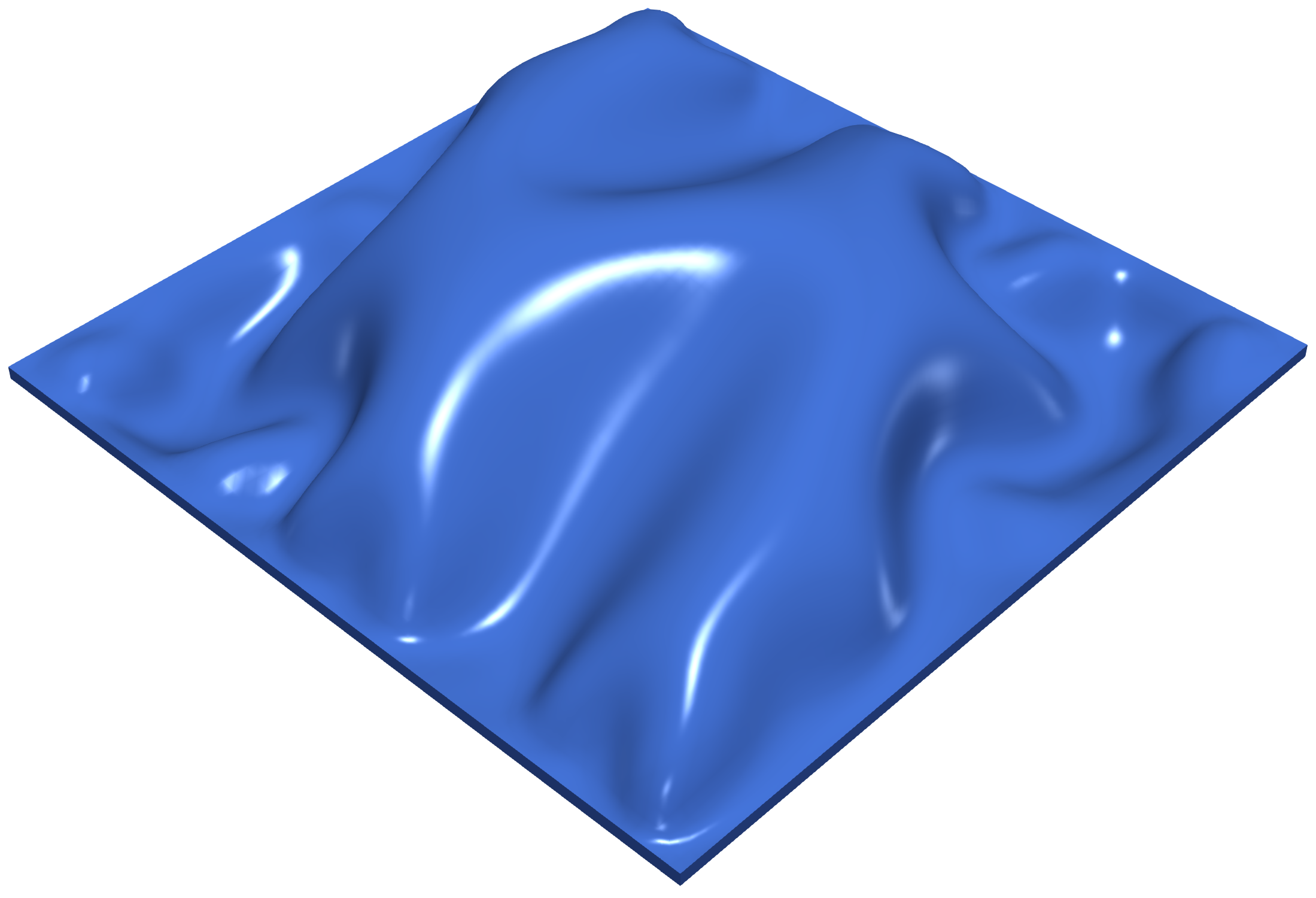}\\
(c)  &  (d)\\%
\includegraphics[width=0.45\textwidth]{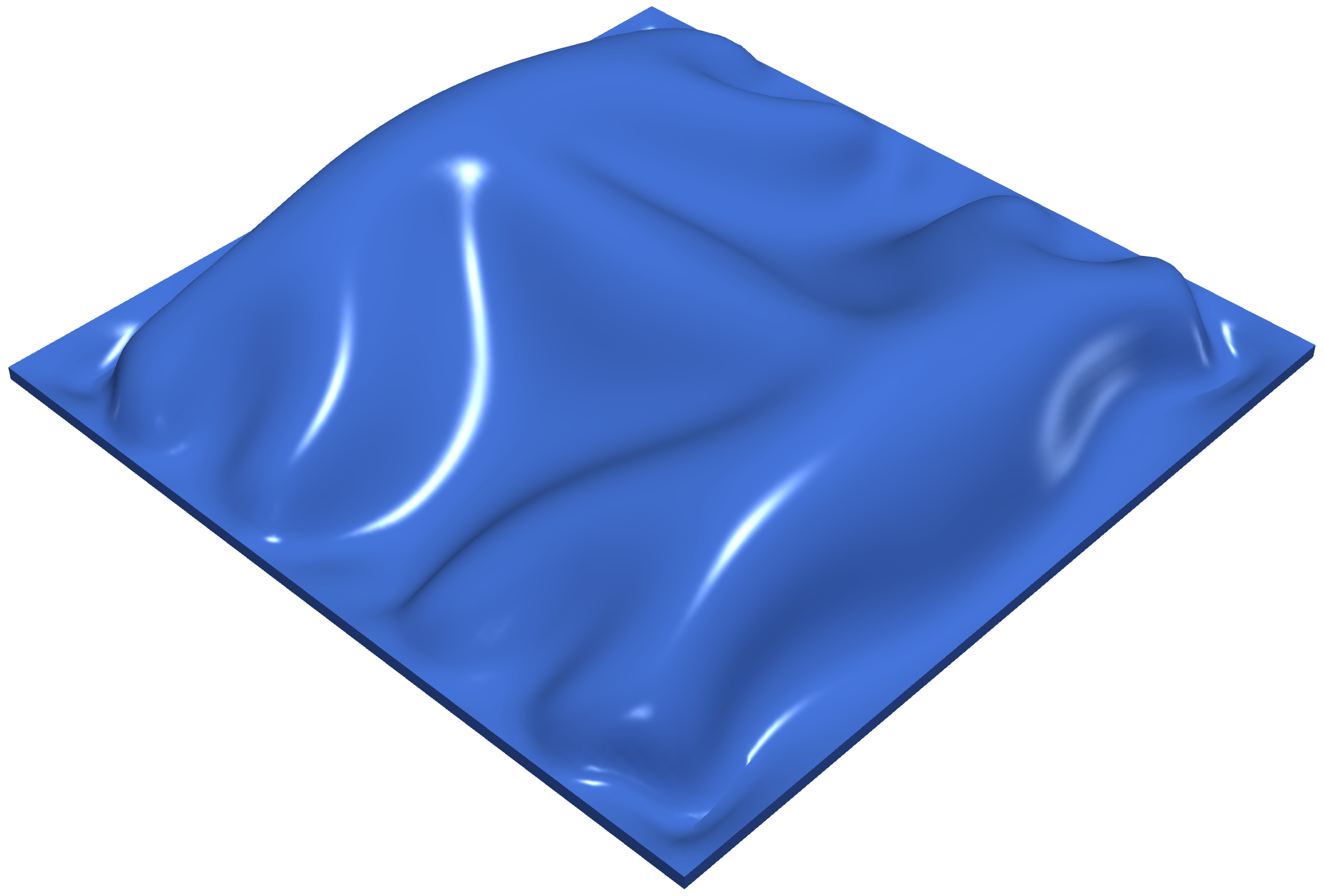}
&
\includegraphics[width=0.45\textwidth]{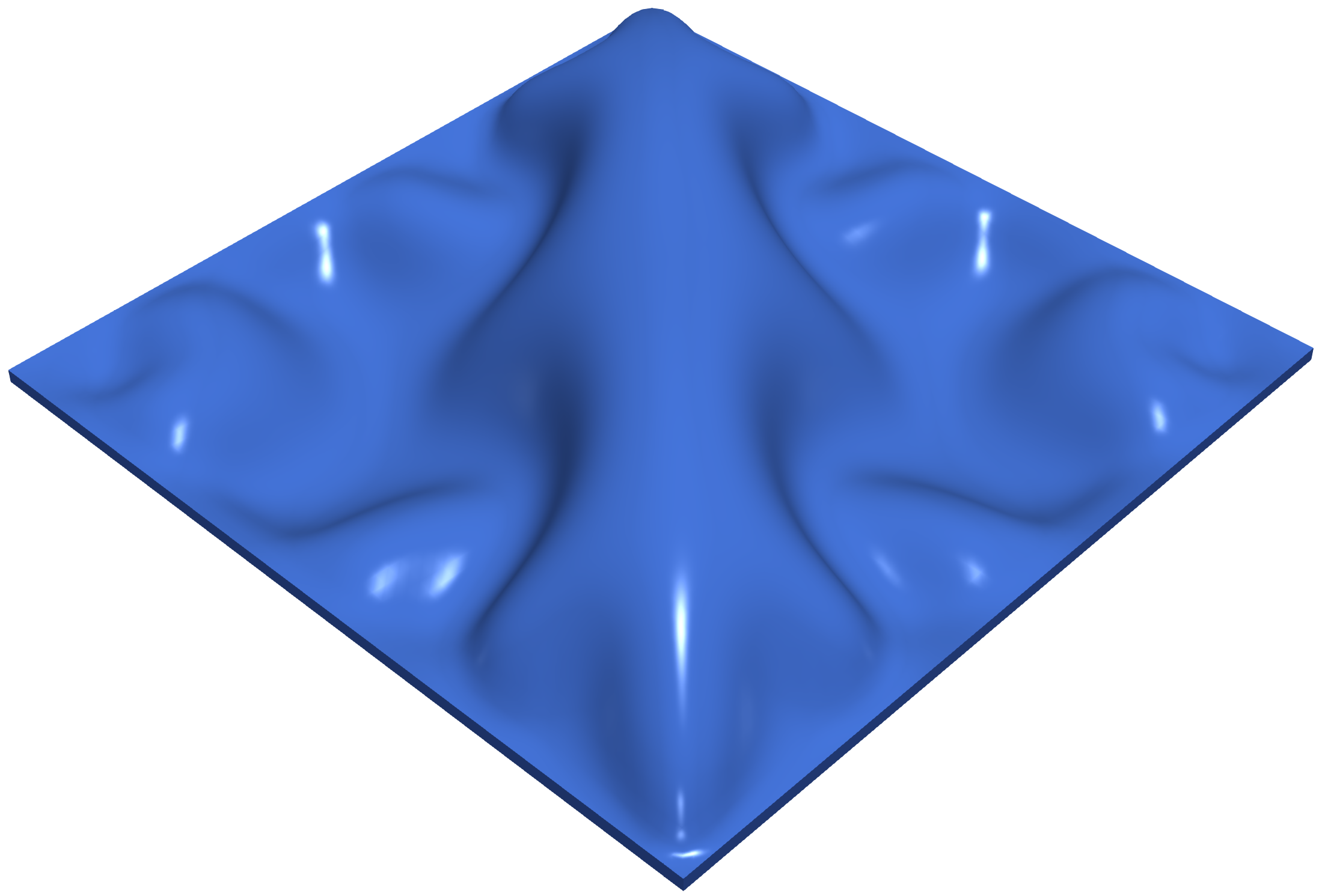}\\
(e)  &  (f)
\end{tabular}
\caption{Electrically induced buckling for numerical example 4 for a value of the accumulated load factor $\Lambda$ of $\Lambda=0.4$. Different preferred directions $\vect{M}$, spherically parametrised according to $\vect{M}=\begin{bmatrix}\cos{\theta}\sin(\phi)& \sin{\theta}\sin(\phi)&\cos(\phi)\end{bmatrix}^T$, with $(\theta,\phi)$: (a) $(0,0)$; (b) $(\pi/8,0)$; (c) $(\pi/4,0)$; (d) $(3\pi/8,0)$; (e) $(\pi/2,0)$; (f) $(\pi/4,\pi/4)$.}
\label{fig:example_Buckling}
\end{figure}

\begin{figure}[htbp!]	
\centering
\begin{tabular}{cc}
\includegraphics[width=0.45\textwidth]{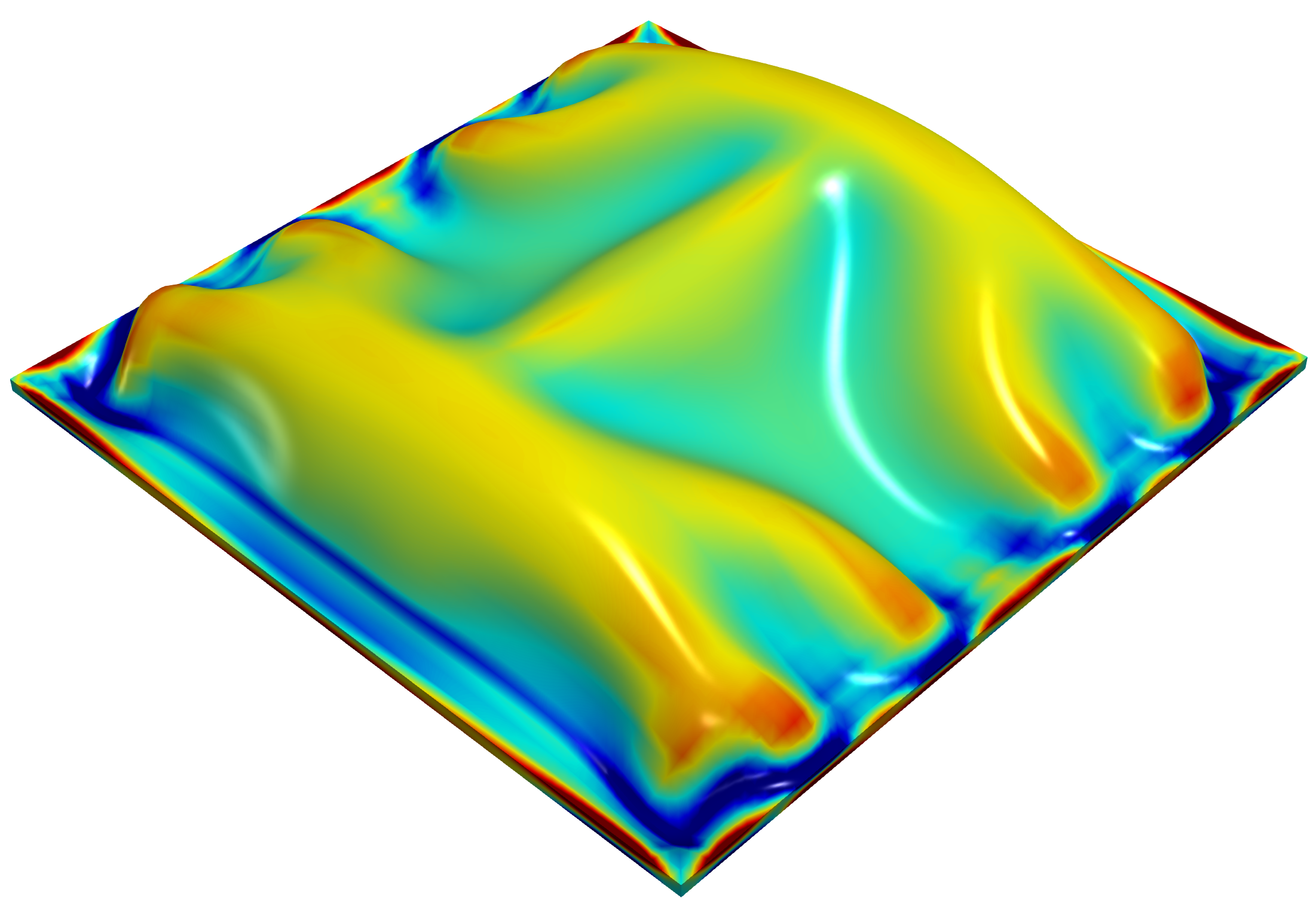}
&
\includegraphics[width=0.45\textwidth]{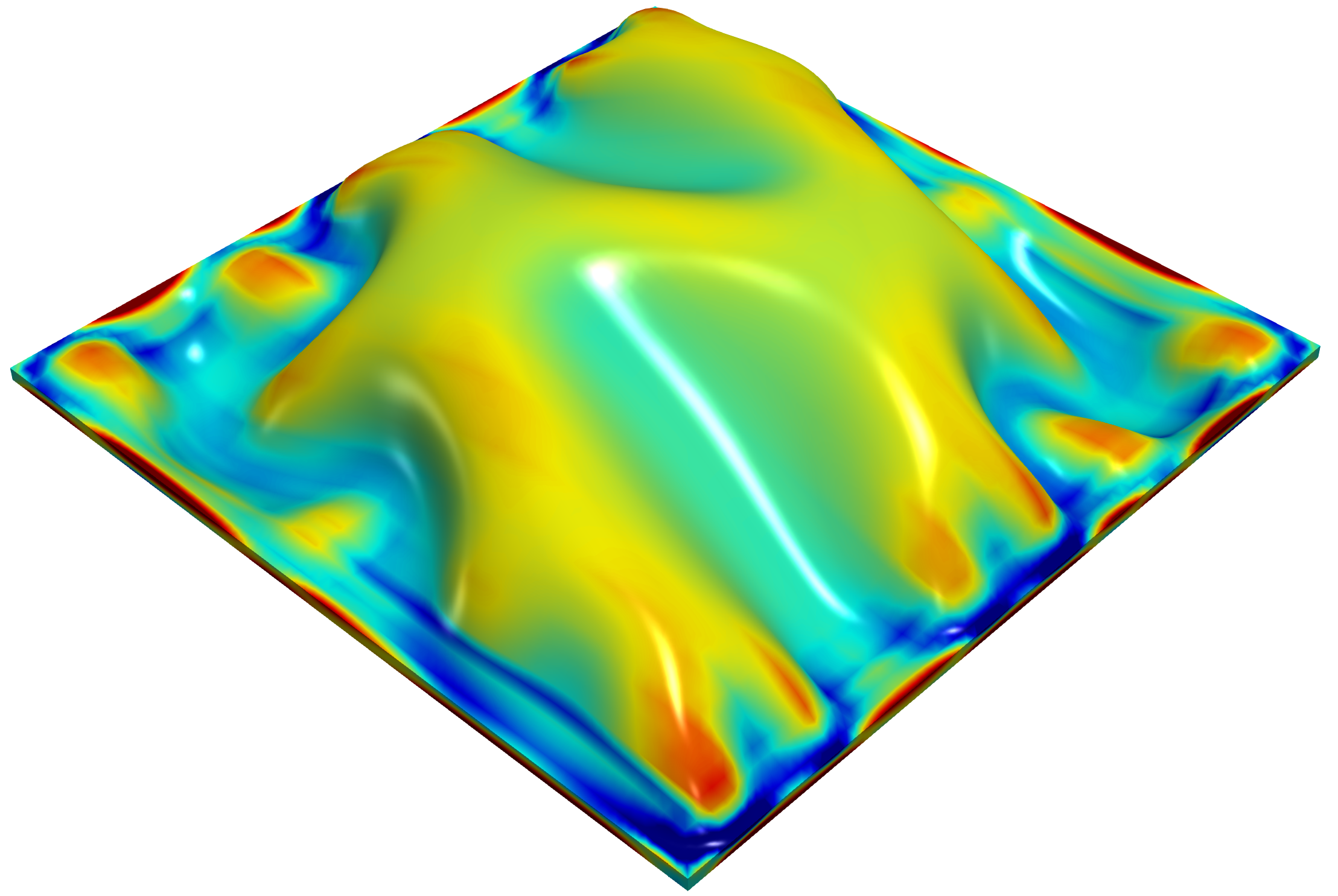}\\
(a)  &  (b)\\

\includegraphics[width=0.45\textwidth]{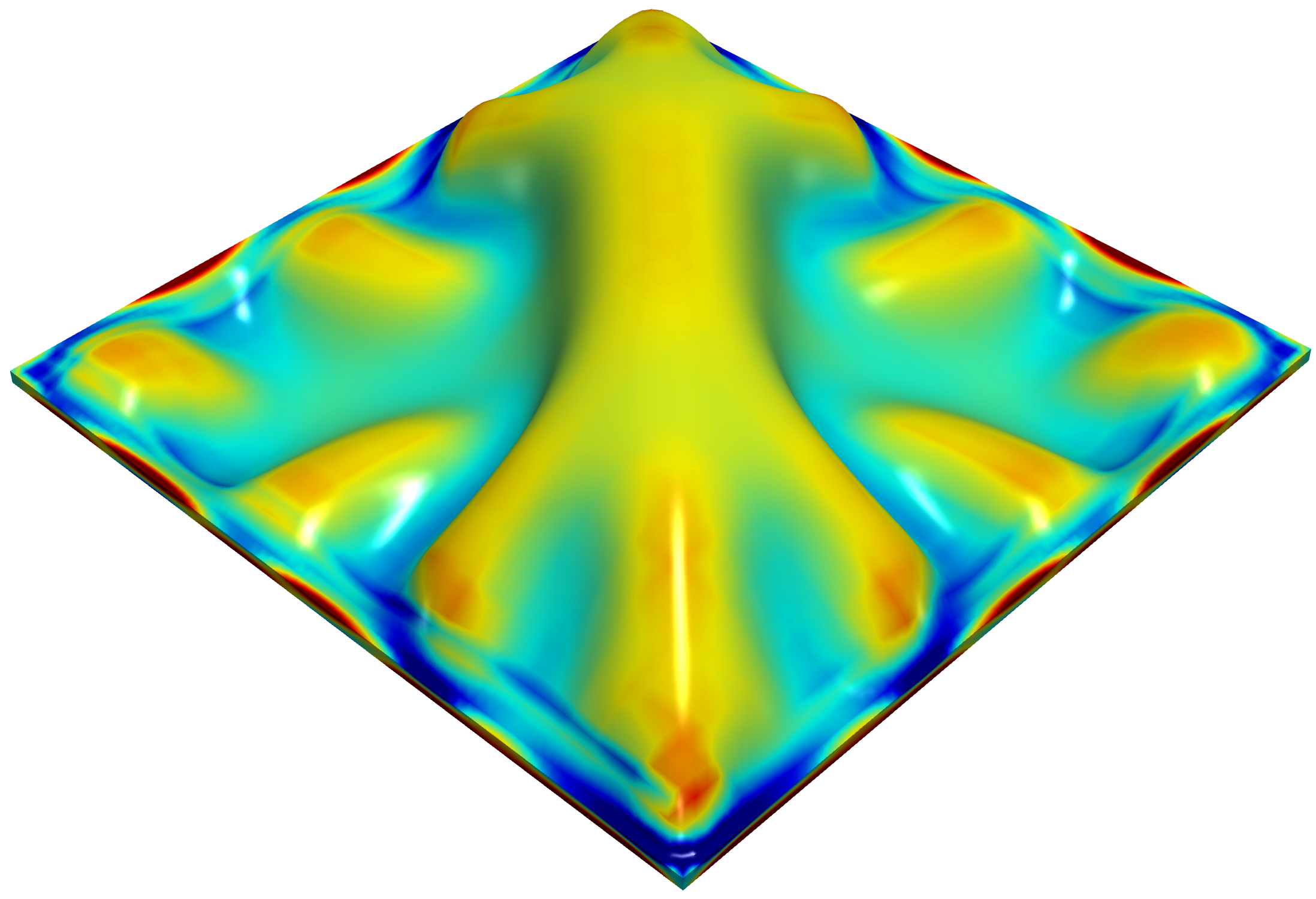}&

\includegraphics[width=0.45\textwidth]{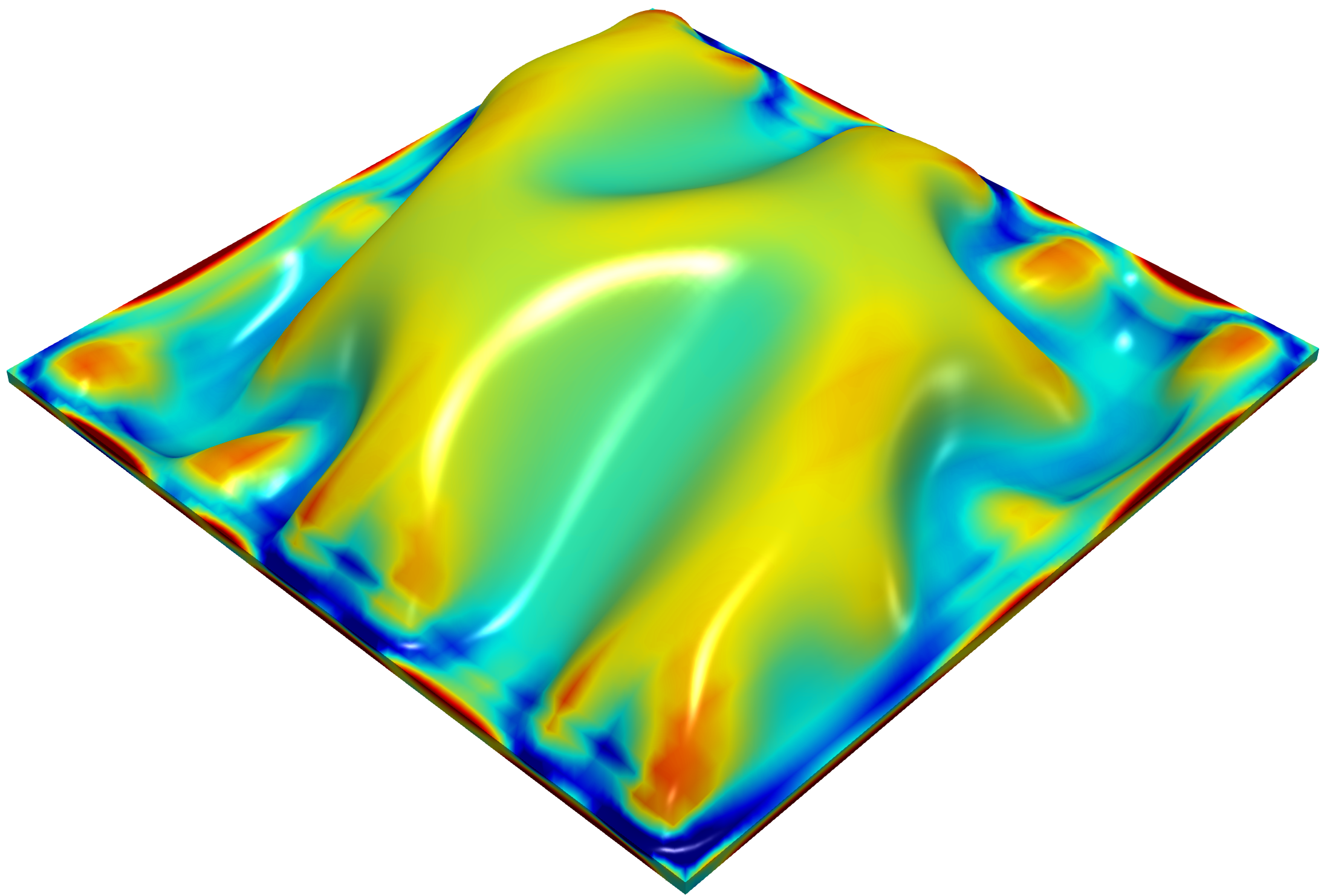}\\

(c)  &  (d)\\
\includegraphics[width=0.45\textwidth]{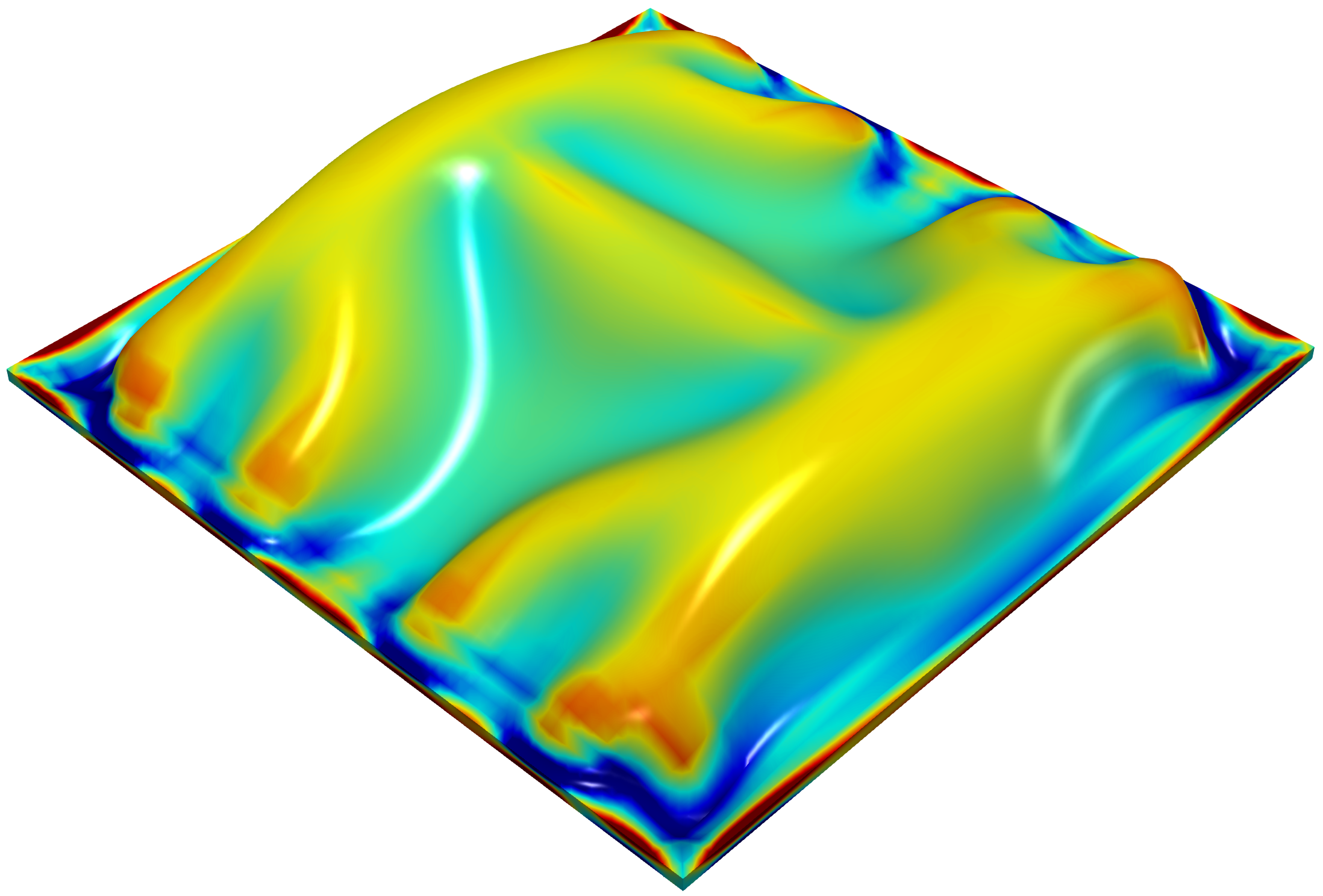}
&
\includegraphics[width=0.45\textwidth]{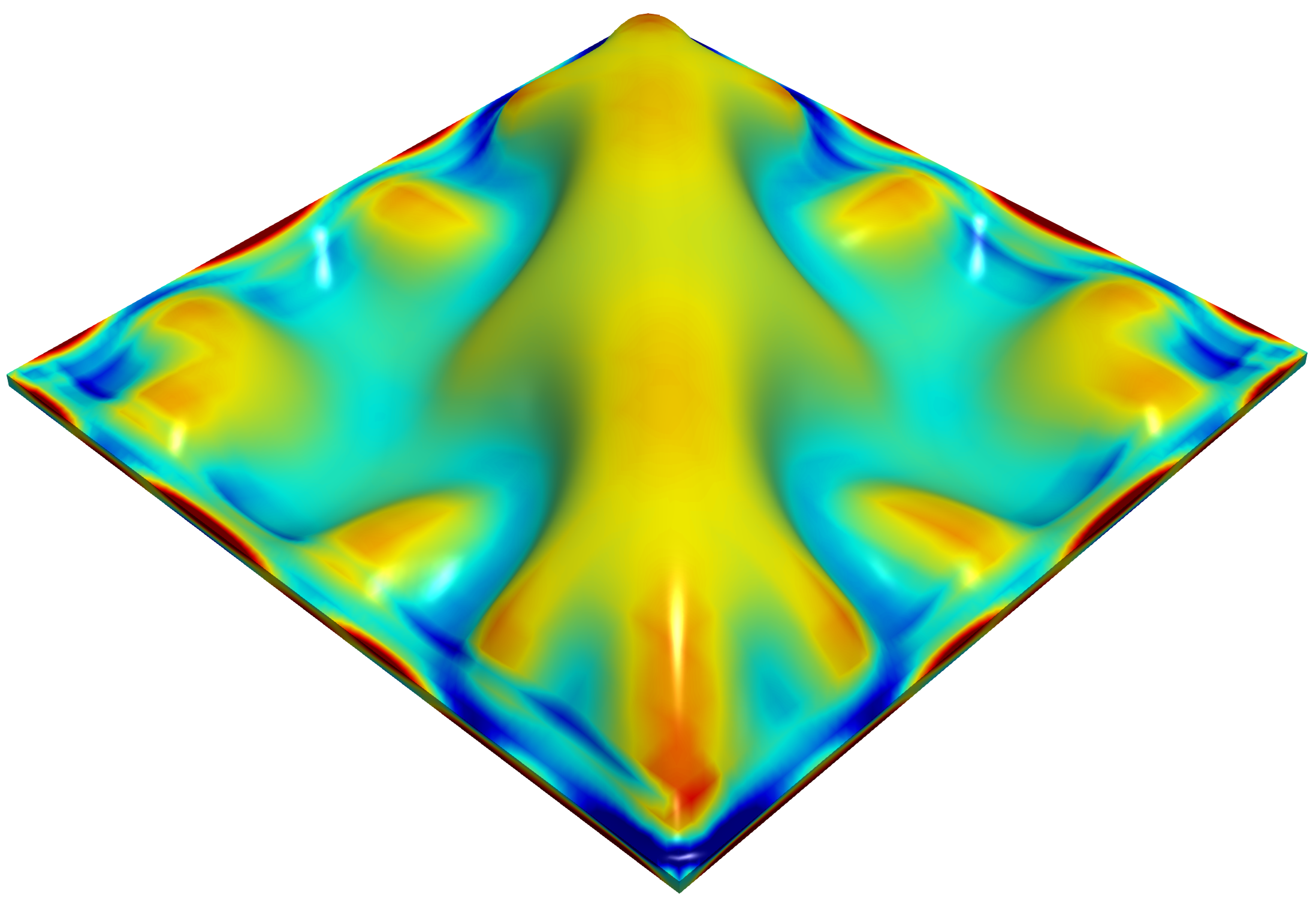}\\
(e)  &  (f)
\end{tabular}
\includegraphics[width=0.45\textwidth]{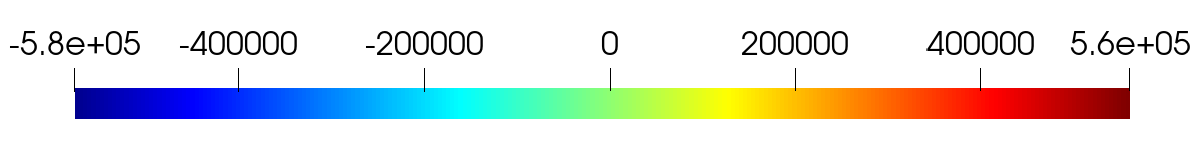}
\caption{Electrically induced buckling for numerical example 4 for a value of the accumulated load factor $\Lambda$ of $\Lambda=0.4$. Contour plot distribution of hydrostatic pressure $p=\text{tr}(J^{-1}\vect{P}\vect{F}^T)$. Different preferred directions $\vect{M}$, spherically parametrised according to $\vect{M}=\begin{bmatrix}\cos{\theta}\sin(\phi)& \sin{\theta}\sin(\phi)&\cos(\phi)\end{bmatrix}^T$, with $(\theta,\phi)$: (a) $(0,0)$; (b) $(\pi/8,0)$; (c) $(\pi/4,0)$; (d) $(3\pi/8,0)$; (e) $(\pi/2,0)$; (f) $(\pi/4,\pi/4)$.}
\label{fig:example_Buckling pressure}
\end{figure}

\section{Conclusions}\label{sec:conclusions}
This paper has presented a phenomenological invariant-based polyconvex transversely isotropic framework for the simulation of EAPs at large strains. Taking inspiration from the pioneering work by Schr\"{o}eder and Neff \cite{schroder2003invariant}, the research builds on previous work by some of authors \cite{Gil2016_NewFramework} for the description of polyconvex isotropic EAPS and enhances it for the modelling of EAPs equipped with an internal transversely isotropic crystallographic structure of the selected types $\mathcal{D}_{\infty h}$ and $\mathcal{C}_{\infty}$ \cite{Vera_thesis}. The paper also bridges the gap between the work developed by \v{S}ilhavý \cite{Silhavy2017} regarding the existence of minimisers and other key computational requirements, such as, rank-one convexity and ellipticity \cite{Gil2016_NewFramework}, the latter necessary to prevent the appearance of spurious mesh-dependent results. 

Three numerical examples are included in order to demonstrate the effect that the anisotropic orientation and the contrast of material properties, as well as the level of deformation and electric field, have upon the response of the EAP when subjected to large three-dimensional stretching, bending, and torsion, including the possible development of wrinkling. Whilst the first example focuses on an academic type of problem, that is, homogeneous deformation, the latter two explore the use of Finite Element discretisations to monitor the response of the framework in realistic three-dimensional in-silico simulations. This manuscript paves the way towards the in-silico simulation of novel soft robotics components made up of predominantly dominant transversely isotropic Electro-Active Polymers (EAPs). Crucially, this type of phenomenologically based model can be enhanced and informed via the use of laboratory data experiments, the up-scaling of computational homogeneisation of complex micro-architecture EAPs, or even the use of Artificial Intelligence type of approaches. These will be the focus of our future work. \\~\\
\noindent \textbf{Declaration of competing interests}\\
The authors declare that they have no known competing financial interests or personal relationships that could have appeared to
influence the work reported in this paper.\\~\\
\noindent \textbf{Acknowledgments}\\
MH and MK acknowledge support by the Czech Science Foundation (EXPRO project No. 19-26143X), AJG acknowledges the financial support received through the European Training Network Protechtion (Project ID: 764636).
\section{Appendices}
\subsection{Ellipticity and acoustic tensor}\label{appendix:acoustic_tensor}
As shown in \cite{Ortigosa2016_NewFramework_ConservationLaws}, the unknown fields $\vect{\phi}$ and $\vect{D}_0$ can be written as a perturbation with respect to equilibrium states $\vect{\phi}^{\text{eq}}$  and $\vect{D}_0^{\text{eq}}$, respectively, by means of the addition of transient travelling (plane) wave functions as
\begin{equation}
\vect{\phi}=\vect{\phi}^{\text{eq}}+\vect{u}f(\alpha);\qquad
\vect{D}_0=\vect{D}_0^{\text{eq}}+\vect{V}_{\perp}g(\alpha);\qquad \alpha=\vect{X}\cdot\vect{V}-ct,
\end{equation}
where $\vect{V}$ represents the polarisation vector of the travelling wave, $c$ the associated speed of propagation of the perturbation with amplitudes $\vect{u}$ and $\vect{V}_{\perp}$ and $f$ and $g$ two arbitrary $\alpha$-scalar functions. Substitution of above ansatzs into \eqref{eqn:definition of F} and  \eqref{eq:DivD_a}, respectively, results in
\begin{subequations}
\begin{align}
\vect{F}&=\underbrace{\vect{\nabla}_0 \vect{\phi}^{\text{eq}}}_{=\vect{F}^{\text{eq}}}+(\vect{u}\otimes \vect{V}) f'(\alpha);\\
\text{DIV}\vect{D}_0-\rho_0&=\underbrace{\text{DIV}\vect{D}^{\text{eq}}_0-\rho_0}_{=0}+\left(\vect{V}_{\perp}\cdot\vect{V}\right)g^{\prime}(\alpha) = 0,
\end{align}
\end{subequations}
hence, $\vect{V}_{\perp}$ must be orthogonal to $\vect{V}$. Linearisation of the first Piola-Kirchhoff stress tensor $\vect{P}$ and the electric field $\vect{E}_0$ about equilibrium states $\vect{P}^{\text{eq}}=\vect{P}(\boldsymbol{\mathcal{U}}^{\text{eq}})$ and $\vect{E}_0^{\text{eq}}=\vect{E}_0(\boldsymbol{\mathcal{U}}^{\text{eq}})$, with $\boldsymbol{\mathcal{U}}^{\text{eq}}=(\boldsymbol{F}^{\text{eq}},\boldsymbol{D}_0^{\text{eq}})$, gives
\begin{equation}\label{eqn:linearisation}
\left[\begin{matrix}
\vect{P}\\
\vect{E}_0
\end{matrix}\right]=\left[\begin{matrix}
\vect{P}^{\text{eq}}\\
\vect{E}_0^{\text{eq}}
\end{matrix}\right]+
\left[\begin{matrix}
\vect{\mathcal{C}}_e & \vect{\mathcal{Q}}^T\\
\vect{\mathcal{Q}} & \vect{\mathcal{\theta}}\\
\end{matrix}\right]\left[\begin{matrix}
:(\vect{u}\otimes \vect{V}) f'(\alpha)\\
\vect{V}_{\perp} g(\alpha)
\end{matrix}\right].
\end{equation}

Substitution of \eqref{eqn:linearisation} into \eqref{eq:DivP} and into Faraday's law \eqref{eqn:faraday}\footnote{In this case, use is made of the more generic expression for Faraday's law, namely $\text{CURL} \vect{E}_0=\vect{0}$.} gives
\begin{subequations}
\begin{align}
\text{DIV}\vect{P}+\vect{b}_0&=
\underbrace{\text{DIV}\vect{P}^{\text{eq}}+\vect{b}_0}_{=\vect{0}}+ \underbrace{\vect{\mathcal{C}_{e,\vect{V}\vect{V}}} \;\vect{u} f''(\alpha)+\vect{\mathcal{Q}_{\vect{V}}}^T \vect{V}_{\perp} g'(\alpha)}_{\rho_R c^2 \vect{u}f''(\alpha)};\label{eq:equil_waves}\\
\text{CURL}\vect{E}_0&=\underbrace{\text{CURL}\vect{E}_0^{\text{eq}}}_{=\vect{0}}+\boldsymbol{\mathcal{E}_{V}}\left[\vect{\mathcal{Q}_{\vect{V}}} \;\vect{u}f''(\alpha)+\vect{\theta}\;\vect{V}_{\perp}g'(\alpha)\right]=\vect{0},\label{eq:Faraday_wave}
\end{align}
\end{subequations}
where 
\begin{equation}
\left[\vect{\mathcal{C}_{e,\vect{V}\vect{V}}}\right]_{ij}=\mathcal{C}_{e,iIjJ}V_{I}V_{J},\qquad
\left[\vect{\mathcal{Q}_{\vect{V}}}\right]_{Ij}=\mathcal{Q}_{IjJ}V_{J},\qquad
\left[\boldsymbol{\mathcal{E}_{V}}\right]_{IJ}=\mathcal{E}_{IJK}V_K,
\end{equation}
where $\mathcal{E}_{IJK}$ denotes the Levi-Civita third-order tensor. Notice that the last underbraced term on the right-hand side of \eqref{eq:equil_waves} is identified with the inertial term corresponding to the acceleration effect with $\rho_R$ the density of the electroactive material. Naturally, a satisfaction of \eqref{eq:Faraday_wave} requires the vector within squared brackets to be collinear to $\vect{V}$, which permits to obtain
\begin{equation}\label{eqn:vperp}
\vect{V}_{\perp}g'(\alpha)=\vect{\theta}^{-1}\left[\beta_{\vect{V}} \vect{V}-\vect{\mathcal{Q}_{\vect{V}}} \;\vect{u}f''(\alpha)\right]
\end{equation}
with $\beta_{\vect{V}}$ a proportionality constant, which can be easily computed from \eqref{eqn:vperp} by projecting both sides of the equation against $\vect{V}$, rendering
\begin{equation}\label{eqn:beta}
\beta_{\vect{V}}=\frac{\vect{V}\cdot \left(\vect{\theta}^{-1}\vect{\mathcal{Q}_{\vect{V}}}\right)\vect{u}}{\vect{V}\cdot \vect{\theta}^{-1} \vect{V}}f''(\alpha).
\end{equation}

Substitution of \eqref{eqn:beta} into \eqref{eqn:vperp} and after some algebraic manipulation gives
\begin{equation}\label{eqn:vperp2}
\vect{V}_{\perp}g'(\alpha)=\vect{\theta}^{-1}\vect{\Omega}\vect{\mathcal{Q}_{\vect{V}}}f''(\alpha);\qquad \vect{\Omega}=\left[\frac{\vect{V}\otimes \vect{\theta}^{-1} \vect{V}}{\vect{V}\cdot \vect{\theta}^{-1} \vect{V}}-\vect{I}\right].
\end{equation}

Substitutions of above equation \eqref{eqn:vperp2} into \eqref{eq:equil_waves}, results in
\begin{equation}\label{eqn:equil_waves_2}
\vect{Q}\vect{u}=\rho c^2 \vect{u};\qquad
\vect{Q}=\vect{\mathcal{C}_{e,\vect{V}\vect{V}}} +\vect{\mathcal{Q}_{\vect{V}}}^T\vect{\theta}^{-1}\vect{\Omega}\vect{\mathcal{Q}_{\vect{V}}},
\end{equation}
where $\vect{Q}$ is the so-called electro-mechanical acoustic tensor which is a function of the constitutive tensors $\boldsymbol{\mathcal{C}}_e$, $\boldsymbol{\mathcal{Q}}$, $\vect{\theta}$ (dependent upon the state of electro-deformation), as well as the polarisation orientation $\vect{V}$. A reasonable physical requirement is that of the existence of real wave speeds, which requires the positive semi-definiteness of the acoustic tensor $\vect{Q}$, that is,
\begin{equation}
\vect{u}\cdot \vect{Q} \vect{u}\geq 0;\qquad \forall\, \boldsymbol{\mathcal{U}},\vect{u},\vect{V}.
\end{equation}
\subsection{Polyconvex invariants}\label{appendix:polyconvex_invariants}
The purpose here is prove that invariant $I_8^{\text{pol}}$ is convex with respect to $\vect{H}$ and $\vect{D}_0$. For that, let us compute $D^2I_8^{\text{pol}} [\delta\vect{H},\delta\vect{D}_0;\delta\vect{H},\delta\vect{D}_0]$, i.e.
\begin{equation}
    \begin{aligned}
    D^2I_8^{\text{pol}} [\delta\vect{H},\delta\vect{D}_0;\delta\vect{H},\delta\vect{D}_0]& =4\vert\vert \vect{H}\vert\vert^2\vert\vert \delta\vect{H}\vert\vert^2 + 8(\vect{H}:\delta\vect{H})^2\\ &+2\vert\vert\delta\vect{H}\vect{D}_0+\vect{H}\delta\vect{D}_0\vert\vert^2+4\vect{H}\vect{D}_0\cdot \delta\vect{H}\delta\vect{D}_0\\
    &+4\vert\vert\vect{D}_0\vert\vert^2\vert\vert\delta\vect{D}_0\vert\vert^2 + 8(\vect{D}_0\cdot \delta\vect{D}_0)^2\\
    &\geq 4\vert\vert \vect{H}\vert\vert^2\vert\vert \delta\vect{H}\vert\vert^2   +4\vect{H}\vect{D}_0\cdot \delta\vect{H}\delta\vect{D}_0+4\vert\vert\vect{D}_0\vert\vert^2\vert\vert\delta\vect{D}_0\vert\vert^2
    \end{aligned}
\end{equation}
Application of the Cauchy-Schwarz inequality in the previous equation yields
\begin{equation}
    \begin{aligned}
    D^2I_8^{\text{pol}} [\delta\vect{H},\delta\vect{D}_0;\delta\vect{H},\delta\vect{D}_0]& 
    \geq 4\Big(\vert\vert \vect{H}\vert\vert^2\vert\vert \delta\vect{H}\vert\vert^2   -\vert\vert\vect{H}\vect{D}_0\vert\vert \vert\vert\delta\vect{H}\delta\vect{D}_0\vert\vert+\vert\vert\vect{D}_0\vert\vert^2\vert\vert\delta\vect{D}_0\vert\vert^2\Big)
    \end{aligned}
\end{equation}
and by definition, the norm of a matrix $\vect{H}$ verifies that
\begin{equation}
    \vert\vert\vect{H}\vect{v}\vert\vert\leq \vert\vert\vect{H}\vert\vert \vert\vert\vect{v}\vert\vert
\end{equation}
for any vector $\vect{v}$. Therefore, it is possible to conclude that
\begin{equation}
    \begin{aligned}
    D^2I_8^{\text{pol}} [\delta\vect{H},\delta\vect{D}_0;\delta\vect{H},\delta\vect{D}_0]& 
    \geq 4\Big(\vert\vert \vect{H}\vert\vert^2\vert\vert \delta\vect{H}\vert\vert^2   -\vert\vert\vect{H}\vert\vert\vert\vert\vect{D}_0\vert\vert \vert\vert\delta\vect{H}\vert\vert\vert\vert\delta\vect{D}_0\vert\vert+\vert\vert\vect{D}_0\vert\vert^2\vert\vert\delta\vect{D}_0\vert\vert^2\Big)\\
    &\geq 4\Big(\vert\vert \vect{H}\vert\vert^2\vert\vert \delta\vect{H}\vert\vert^2   -2\vert\vert\vect{H}\vert\vert\vert\vert\vect{D}_0\vert\vert \vert\vert\delta\vect{H}\vert\vert\vert\vert\delta\vect{D}_0\vert\vert+\vert\vert\vect{D}_0\vert\vert^2\vert\vert\delta\vect{D}_0\vert\vert^2\Big)\\
    &=4\Big(\vert\vert \vect{H}\vert\vert\vert\vert \delta\vect{H} \vert\vert- \vert\vert\vect{D}_0\vert\vert\vert\vert\delta\vect{D}_0\vert\vert\Big)^2\geq 0
    \end{aligned}
\end{equation}
which therefore, concludes the proof. A similar derivation can be carried out in order to proof convexity with respect to $\vect{d}$ and $\vect{F}$ of the polyconvex invariant $K_{2,\text{pol}}^{\mathcal{D}_{\infty}}$.

\subsection{Finite Element implementation}\label{appendix:FEM}
This Appendix briefly summarises the Finite Element spatial discretisation procedure implemented for the second and third numerical examples. Specifically, we focused on the mixed Hu-Washizu type of variational principle $\tilde{\Pi}_e$ \eqref{eqn:mixed variational principle II}, whose stationary conditions yield
\begin{equation}\label{eq:stationary_1}
\begin{aligned}
D \tilde{\Pi}_e [\delta \vect{\phi}]&=\int_{{\Omega}_0} \partial e_{\vect{F}} : \vect{\nabla}_0 \delta \vect{\phi} \;dV - D \Pi_{\text{ext}} [\delta \vect{\phi}]=0;\\
D \tilde{\Pi}_e [\delta \vect{D}_0]&=\int_{{\Omega}_0} (\partial e_{\vect{D}_0}+ \vect{\nabla}_0 \varphi)\cdot \delta \vect{D}_0\;dV=0;\\
D \tilde{\Pi}_e [\delta \varphi]&=\int_{{\Omega}_0} \vect{D}_0 \cdot \vect{\nabla}_0 \delta \varphi \;dV - D \Pi_{\text{ext}} [\delta \varphi]=0,\\
\end{aligned}
\end{equation}
where $\{\vect{\phi},\vect{D}_0,\varphi\} \in \mathbb{V}^{\vect{\phi}}_{\bar{\vect{\phi}}} \times \mathbb{V}^{\vect{D}_0} \times \mathbb{V}^{\varphi}_{\bar{\varphi}}$ and with admissible variations $\{\delta \vect{\phi}, \delta \vect{D}_0, \delta \varphi\} \in \mathbb{V}^{\vect{\phi}}_{\vect{0}} \times \mathbb{V}^{\vect{D}_0} \times \mathbb{V}^{\varphi}_0$, defined as
\begin{eqnarray}
\nonumber
\mathbb{V}^{\vect{\phi}}_{\bar{\vect{\phi}}} &=&\begin{Bmatrix}\vect{\phi}:\O_0\rightarrow\R^3;\; [\vect{\phi}]_i\in H^1(\O_0);\; J>0;\;\vert\; \vect{\phi} = \bar{\vect{\phi}}\;\text{on}\; \partial\O_0^u  \end{Bmatrix}; \\
\nonumber
\mathbb{V}^{{\varphi}}_{\bar{\varphi}} &=&\begin{Bmatrix}{\varphi}:\O_0\rightarrow\R;\; {\varphi}\in H^1(\O_0)\;\vert\; {\varphi} = \bar{{\varphi}} \;\text{on}\; \partial\O_0^\varphi  \end{Bmatrix};\\
\nonumber
\mathbb{V}^{\vect{D}_0}&=&\begin{Bmatrix}\bm{D}_0:\O_0\rightarrow\R^3;\; \bm{D}_0\in L_2(\O_0) \end{Bmatrix}.
\end{eqnarray}

Further, we apply standard isoparametric Finite Elements and the domain $\O_0$ is divided into elements and the displacement, electric potential and electric displacement fields are approximated by means of approximation finite element spaces $\mathbb{V}^{\vect{\phi}}_{\bar{\vect{\phi}},h} \subset \mathbb{V}^{\vect{\phi}}_{\bar{\vect{\phi}}}$, $V^{\varphi}_{\bar{\varphi},h} \subset V^{\varphi}_{\bar{\varphi}}$, and $V^{\vect{D}_0}_h \subset V^{\vect{D}_0}$ defined as
\begin{eqnarray}
\nonumber
\mathbb{V}^{\vect{\phi}}_{\bar{\vect{\phi}},h} &=&\begin{Bmatrix}\vect{\phi}:\O_0\rightarrow\R^3;\; \vect{\phi} = \sum_{a=1}^{n^{\vect{\phi}}_{nod}}N_{\vect{\phi}}^a \vect{d}^a_{\vect{\phi}}\;\vert\; \vect{\phi} = \bar{\vect{\phi}}\;{\rm on}\; \partial\O_0^u  \end{Bmatrix}; \\
\nonumber
V^{\varphi}_{\bar{\varphi},h} &=&\begin{Bmatrix}{\varphi}:\O_0\rightarrow\R^3;\; \varphi = \sum_{a=1}^{n^\varphi_{nod}}N_\varphi^a d^a_\varphi\;\vert\; \varphi = \bar{\varphi}\;{\rm on}\; \partial\O_0^\varphi  \end{Bmatrix}; \\
\nonumber
V^{\vect{D}_0}_h&=&\begin{Bmatrix}\bm{D}_0:\O_0\rightarrow\R^3;\; \bm{D}_0= \sum_{a=1}^{n^D_{nod}}N_{\vect{D}_0}^a \vect{d}^a_{\vect{D}_0} \end{Bmatrix},
\end{eqnarray}
where for any field $i = \begin{Bmatrix}{{\vect{\phi}},\vect{D}_0,\varphi}\end{Bmatrix}$, $n^i_{nod}$ denotes number of interconnected nodes, $N_i^a$ is the $a$-th basis function and $\bm{d}_i^a$ is the $a$-th node associated with the appropriate field. Use of above Finite Element interpolation spaces into \eqref{eq:stationary_1} leads to a system of nonlinear algebraic equations, whose solution is obtained via a consistent Newton-Raphson linearisation iterative procedure. In order to reduce the computational
cost of the proposed formulation, a piecewise discontinuous interpolation of the field $\vect{D}_0$ is followed. A standard static condensation procedure \cite{Ortigosa2016_NewFramework_FiniteElements} is used to condense out the degrees of freedom of the field $\vect{D}_0$, yielding a formulation
with a cost comparable to that of a two-field $(\vect{\phi},\varphi)$ formulation.

\def\cprime{$'$} \def\ocirc#1{\ifmmode\setbox0=\hbox{$#1$}\dimen0=\ht0
  \advance\dimen0 by1pt\rlap{\hbox to\wd0{\hss\raise\dimen0
  \hbox{\hskip.2em$\scriptscriptstyle\circ$}\hss}}#1\else {\accent"17 #1}\fi}

%
\end{document}